\documentclass[11pt]{article}
\usepackage{etoolbox}
\usepackage{inputenc}
\usepackage[T1]{fontenc}
\providetoggle{hidestuff}
\settoggle{hidestuff}{true}

\usepackage[svgnames]{xcolor}
\usepackage{amsmath, amssymb, amsthm}
\usepackage{mathtools}
\usepackage{stix}
\usepackage{microtype}
\usepackage{xfrac}
\usepackage[margin=1in,left=1in,right=1in,marginpar=0.75in]{geometry}
\usepackage{graphicx}
\usepackage[colorlinks=true,urlcolor=blue,citecolor=blue,linkcolor=blue]{hyperref}
\usepackage{natbib}
\usepackage{setspace}
\iftoggle{hidestuff}{\usepackage[disable]{todonotes}}{\usepackage{todonotes}}
\usepackage{cleveref}
\usepackage{soul}
\usepackage{accents}
\newcommand{\bl}[1]{\underaccent{\bar}{#1}}
\newcommand{\GPV}{\mathrm{gpv}}
\newcommand{\older}[1]{}
\usepackage[inline]{enumitem}

\allowdisplaybreaks
\newcommand{\diam}[1]{\accentset{\diamond}{#1}}
\newcommand{\Var}{\mathbb{V}}

\newtheoremstyle{jorisstyle}
{3pt}
{3pt}
{}
{}
{\bfseries}
{}
{ }
{\thmname{#1}\thmnumber{ #2'}\thmnote{\bfseries~#3}.}

\newcommand{\cbT}{\sqrt[3]{T}}
\newcommand{\mle}{\text{mle}}
\newcommand{\C}{\mathbb{C}}
\renewcommand{\S}{\mathbb{S}}
\DeclareRobustCommand\<{\begin{equation}}
\DeclareRobustCommand\>{\end{equation}}

\newcommand{\pd}[2]{\frac{\partial #1}{\partial #2}} 
\newcommand{\pdd}[2]{\frac{\partial^2 #1}{\partial #2^2}}

\newcommand\ot\leftarrow
\renewcommand{\Re}{\mathbb{R}}								
\newcommand{\mathb}[1]{\underbar{$#1$}}						
\DeclareMathOperator*{\argmax}{arg\,max}						
\newcommand{\convd}{\stackrel{d}{\to}}
\theoremstyle{definition}
\newtheorem{defn}{Definition}
\theoremstyle{jorisstyle}

\newtheorem*{just*}{Justification}
\newcommand{\oo}{o\parens{1}}
\newcommand{\sI}{\mathscr{I}}
\theoremstyle{plain}
\newtheorem{thm}{Theorem}
\newcommand{\sB}{\mathscr{B}}
\newcommand{\sH}{\mathscr{H}}

\theoremstyle{plain}
\newtheorem{lem}{Lemma}
\newtheorem*{thm*}{Theorem}

\newtheorem{ass}{Assumption}
\newtheorem{ex}{Example}

\definecolor{dkgreen}{rgb}{0,0.6,0}
\definecolor{gray}{rgb}{0.5,0.5,0.5}
\definecolor{mauve}{rgb}{0.58,0,0.82}

\newcommand{\beT}{\breve e_T}
\newcommand{\alT}{\alpha_T}

\newcommand{\G}{\mathbb{G}}
\newcommand{\beTml}{\breve e_T^{\text{MLE}}}
\newcommand{\balTml}{\breve \alpha_T^{\text{MLE}}}

\DeclarePairedDelimiter{\parens}{(}{)}
\DeclarePairedDelimiter\cparens\{\}
\DeclarePairedDelimiter\sparens[]

\DeclarePairedDelimiter{\halfopen}[)

\DeclarePairedDelimiter{\abs}\vert\vert

\DeclarePairedDelimiterX\condr[1](){#1}
\DeclarePairedDelimiterX\condc[1]\{\}{#1}
\DeclarePairedDelimiterX\conds[1][]{#1}
\DeclarePairedDelimiterX\condn[1]{}{}{#1}

\newcommand{\Exp}{\mathbb{E}}
\newcommand{\Expr}[2][]{\Exp\parens[#1]{#2}}

\newcommand{\Expc}[2][]{\Exp\cparens[#1]{#2}}

\newcommand{\sA}{\mathscr{A}}
\newcommand{\dif}{\:\mathrm{d}}
\let\Set\condc

\usepackage{xspace}
\newcommand{\convw}{\rightsquigarrow}
\newcommand{\joris}[1]{\normalmarginpar\todo[color=Lime,fancyline,size=\tiny]{#1}\xspace}
\newcommand{\karl}[1]{\todo[color=Pink,fancyline,size=\tiny]{#1}\xspace}
\newcommand{\thought}[1]{\todo[color=Magenta,fancyline,size=\tiny]{#1}\xspace}
\crefname{equation}{}{}
\Crefname{equation}{Equation}{Equations}
\crefname{proposition}{proposition}{propositions}
\creflabelformat{proposition}{#2#1#3}
\crefname{figure}{figure}{figures}
\creflabelformat{figure}{#2#1#3}
\crefname{table}{table}{tables}
\creflabelformat{table}{#2#1#3}
\crefname{lem}{lemma}{lemmas}
\creflabelformat{lemma}{#2#1#3}
\crefname{thm}{theorem}{theorems}
\creflabelformat{theorem}{#2#1#3}
\crefname{corollary}{corollary}{corollaries}
\creflabelformat{lemma}{#2#1#3}
\crefname{ass}{assumption}{assumptions}
\crefname{enumi}{part}{parts}
\creflabelformat{enumi}{#2#1#3}

\newcommand{\maligned}[1]{\left\{\begin{aligned}#1\end{aligned}\right.}
\newcommand{\ehat}{\hat e_T}
\newcommand{\ehatpsi}{\hat e_{T\psi}}

\newcommand{\lazyint}[2][s]{\int_{-\infty}^\infty #2 \dif #1}
\newcommand{\busyint}[2][s]{\int_{\bl\upsilon_\psi}^{\bar\upsilon_\psi} #2 \dif #1}
\newcommand{\sV}{\mathscr{V}}
\newcommand{\sP}{\mathscr{P}}
\newcommand{\opone}{\ensuremath{o_p\parens{1}}}
\newcommand{\comment}[1]{}
\usepackage{marginnote}
\newcommand{\jorisr}[1]{\reversemarginpar\todo[color=Lime,size=\tiny]{#1}}

\newcommand{\one}{\mathbb{1}}
\DeclareMathOperator*{\argmin}{argmin}
\newcommand{\jorisinline}[1]{\todo[inline,color=Lime]{#1}}
\newcommand{\karlinline}[1]{\todo[inline,color=Pink]{#1}}
\newcommand{\thoughtinline}[1]{\todo[inline,color=Magenta]{#1}}

\DeclareMathOperator{\Med}{Med}
\newcommand{\ahT}{\hat\alpha_T}
\newcommand{\ahTpsi}{\hat\alpha_{T\psi}}
\newcommand{\ahTmt}{\hat\alpha_{T,-t}}
\newcommand{\ahJ}{\breve\alpha_{TJ}}
\newcommand{\ahJs}{\hat\alpha_{TJ}}
\newcommand{\beTmt}{\breve e_{T,-t}}
\newcommand{\BS}{\mathrm{BS}}
\newcommand{\BShat}{\widehat{\mathrm{BS}}}
\newcommand{\MV}{\mathrm{MV}}
\newcommand{\MVhat}{\widehat{\mathrm{MV}}}
\newcommand{\PR}{\mathrm{PR}}

\newcommand{\uint}[2][p]{\int_0^1 #2 \dif #1}
\newcommand{\infint}[2][z]{\int_{-\infty}^\infty #2 \dif #1}
\newcommand{\meanT}[1][]{\frac{1}{T#1}\sum_{t=1}^T}
\newcommand{\fpp}{f_p\parens{p}}
\DeclareFontFamily{U}{mathx}{}
\DeclareFontSubstitution{U}{mathx}{m}{n}
\DeclareFontShape{U}{mathx}{m}{n}{ <-> *[.4]mathx10}{}
\DeclareSymbolFont{mathx}{U}{mathx}{m}{n}
\DeclareFontFamily{T1}{mathfrak}{}
\DeclareFontSubstitution{T1}{mathfrak}{m}{n}
\DeclareFontShape{T1}{mathfrak}{m}{n}{ <-> *stix-mathfrak}{}
\DeclareSymbolFont{mathfrak}{T1}{mathfrak}{m}{n}
\DeclareMathDelimiter{\bkdelim}{4}{mathfrak}{'071}{mathx}{"37}

\newcommand{\sL}{\mathscr{L}}
\newcommand{\grey}[1]{{\color{Grey}#1}}
\allowdisplaybreaks

\DeclareMathOperator{\bs}{BS}
\DeclareMathOperator{\bshat}{\widehat{\bs}}

\newcommand{\symm}{\mathrm{symm}}

\renewcommand{\pd}[2]{\partial_{#2}#1}

\newcommand{\drop}[1]{\todo[inline,color=LightGrey]{stuff at least one of us thinks can be dropped has been omitted here.}}

\newcommand{\ahTR}{\bar {\alpha}_{T{\psi}}^R}

\usepackage{xpatch}

\xpatchcmd{\proof}
{\itshape}
{\bfseries}
{}
{}
\usepackage{afterpage}

\title{Estimation of Auction Models\\with Shape Restrictions}
\date{\today}
\author{Joris Pinkse and Karl Schurter\thanks{We thank Yanqin Fan, Paul Grieco, Brent Hickman, Kei Hirano, Sung Jae Jun, Laurent Lamy, Ruixuan Liu, Orville Mondal, Peter Newberry, and participants of the Penn State brown bag workshop, the Lancaster auctions conference, the EARIE conference in Barcelona, and the PennCorn conference on econometrics and industrial organization in Ithaca for their helpful comments.}\\ 
}

\usepackage{tikzsymbols}

\usetikzlibrary{calc}

\iftoggle{hidestuff}{
\renewcommand{\grey}[1]{}
}{}
\renewcommand{\drop}[1]{}
\renewcommand{\grey}[1]{}

\begin{document}
\maketitle


\begin{abstract}\noindent
We introduce several new estimation methods that leverage shape constraints in auction models to estimate various objects of interest, including the distribution of a bidder's valuations, the bidder's ex ante expected surplus, and the seller's counterfactual revenue.  The basic approach applies broadly in that (unlike most of the literature) it works for a wide range of auction formats and allows for asymmetric bidders.  Though our approach is not restrictive, we focus our analysis on first--price, sealed--bid auctions with independent private valuations. We highlight two nonparametric estimation strategies, one based on a least squares criterion and the other on a maximum likelihood criterion.  We also provide the first direct estimator of the strategy function.  We establish several theoretical properties of our methods to guide empirical analysis and inference. In addition to providing the asymptotic distributions of our estimators, we identify ways in which methodological choices should be tailored to the objects of their interest. For objects like the bidders' ex ante surplus and the seller's counterfactual expected revenue with an additional symmetric bidder, we show that our input--parameter--free estimators achieve the semiparametric efficiency bound. For objects like the bidders' inverse strategy function, we provide an easily implementable boundary--corrected kernel smoothing and transformation method in order to ensure the squared error is integrable over the entire support of the valuations. An extensive simulation study illustrates our analytical results and demonstrates the respective advantages of our least--squares and maximum likelihood estimators in finite samples. Compared to estimation strategies based on kernel density estimation, the simulations indicate that the smoothed versions of our estimators enjoy a relatively large degree of robustness to the choice of an input parameter.
\end{abstract}

\clearpage

\section{Introduction}
%
%

We develop several new estimators that leverage shape constraints implied by the bidder's incentive--compatibility condition in auction models.  Unlike most existing methods, the basic approach applies broadly in that it works both for a wide range of auction formats and allows for asymmetric bidders. For the case of first price auctions, we establish asymptotic results for multiple unsmoothed (piecewise--constant) and smoothed estimators of the inverse strategy function, the first direct estimator of the bid function, and estimators of a variety of other objects, including a bidder's value density function, her expected surplus, and the mean of her value distribution.  We consider each of these objects separately in order to provide guidance to applied researchers on ways in which our estimators can be optimized for the specific objects of their interest. For objects like the expected surplus, our approach (unlike most existing methods) does not require the researcher to choose an input parameter (e.g.~a kernel bandwidth) and achieves the semiparametric efficiency bound. For objects like the valuation distribution, we use a boundary--corrected kernel smoothing method so that our estimators converge at the same optimal nonparametric rate as popular alternatives.  We purposely propose a relatively large number of estimation options because different choices work better depending on the nature of a given problem, as is borne out by our simulation study. We further provide simulation evidence to confirm our theoretical finding that, because our approach imposes shape restrictions a priori, it is more robust to the choice of inputs compared with alternative estimation strategies and also appears robust to the choice of design.  


%
%
The key insight behind our approach is that, despite the great diversity of auction formats we might consider, the fundamental nature of a bidder's decision problem is the same. Specifically, given the strategies of a bidder's competitors and the distribution of their private values, the bidder chooses her bid to optimally trade off the probability of winning with her expected payment to the seller. Though the details of this trade--off as a function of the bid are complicated and depend on the specifics of the auction rules, an envelope theorem argument demonstrates that the equilibrium payment function $e$ must be convex in the probability with which the bidder expects to win. Moreover, the first--order condition of the bidder's problem requires that the slope of the equilibrium expected payment function at the optimally chosen win probability is equal to her private valuation. Thus, the derivative $\alpha$ of the convex payment function is equivalently viewed as the inverse strategy function, which maps optimally chosen win--probabilities to values. These facts have been used to establish the revenue equivalence theorem \citep{myerson1981optimal,Milgrom1982} and were subsequently invoked as a generic nonparametric identification strategy \citep{Larsen2018}. Our paper exploits this change of variables from bids to win--probabilities further in order to relate the literature on nonparametric estimation in auctions to the large literature on nonparametric estimation under shape constraints.

%
%
The main benefits of this change of variables are threefold. First, by reformulating the target of estimation as the slope of a convex function, we open the door to a variety of well--known estimation strategies such as (shape--)constrained (nonparametric) least squares and (a new version of) nonparametric maximum likelihood estimation (MLE),\footnote{See e.g.\ \citet{brunk1955maximum} for an early example of nonparametric estimation subject to shape constraints.} as well as some more adventurous estimators like a jackknife estimator.
Second, this framework generalizes the large and growing toolkit for nonparametric estimation and testing in first--price auctions to generic auction mechanisms. And, finally, it allows the econometrician to easily impose the structure of symmetric equilibria, namely that the marginal distribution of a bidder's optimally chosen win--probability is a known function that only depends on the number of bidders.  Importantly, the distribution function does \emph{not} depend on the unknown distribution of the bidders' private valuations. This a priori knowledge of the win--probability distribution yields sizable improvements in the asymptotic distribution of our estimators.

%
%
The latter observation is especially useful when we consider the estimation of objects that are less primitive than the value density $f_v$ but may be of more direct interest to the researcher. For example, the bidder's ex ante expected surplus can be expressed as an integral of $\alpha(p) p - e(p)$ with respect to the  win--probability distribution. The win--probability distribution can be precisely estimated in symmetric or asymmetric equilibria. In the symmetric case, however, we show that one can significantly reduce the asymptotic variance by substituting the win--probability distribution which is known to prevail in any symmetric equilibrium rather than an estimate thereof.

%
%
Similar to the bidder's surplus, the mean of the bidder's valuations and the seller's expected profit as a function of the number of bidders can also be expressed as a (weighted) integral of the inverse strategy function. We show that such objects can be estimated at a square--root rate of convergence when ${\alpha}$ is replaced with an unsmoothed estimate. The reason is that although the first step estimator of ${\alpha}$ converges at a cube--root rate, it has little bias.\footnote{Indeed, the asymptotic distribution is centered at zero.} The act of integration in the second stage acts as an average and hence reduces the asymptotic variance. Moreover, the resultant estimators achieve the semiparametric efficiency bound when one fully exploits the symmetric structure of the equilibrium. If the researcher does not assume bidders are symmetric or only observes one competitor bid per auction, the semiparametric efficiency bound on the asymptotic variance is larger, but we again show that the unsmoothed plug--in estimators for the mean valuation and bidder's surplus attain the efficiency bound.

%
%
Thus, smoothing the estimate of ${\alpha}$ does not necessarily improve the asymptotic performance of the desired object. Indeed, it can be detrimental. For example, in order to achieve square--root consistency of the mean valuation using a smoothed estimate of the inverse strategy function, one would need to ``undersmooth'' by choosing an input parameter to slow down the pointwise rate of convergence of the inverse strategy function and reduce its bias. However, one must avoid too much undersmoothing using methods that do not impose monotonicity a priori (e.g.\ \citet{marmer2012quantile,ma2019monotonicity}, and \citet[GPV]{Guerre2000}), because letting the bandwidth go to zero for a fixed sample size would yield an inconsistent estimator for the inverse bid function and also produce an inconsistent estimator of the mean valuation. In contrast, there is no risk of too much undersmoothing using our approach, because our undersmoothed and unsmoothed estimators of objects like the mean valuation are asymptotically equivalent and attain the efficiency bound. This asymptotic efficiency result for both our unsmoothed and undersmoothed estimators therefore provides a large degree of robustness in the choice of bandwidths relative to existing methods.

%
%
More generally, we identify several ways in which the estimator can be tailored to the ultimate target of the empirical analysis. Though it would be possible to first estimate the value density and then obtain, e.g.\ bidder one's expected surplus, there is no benefit of taking this intermediate step.  Indeed, as noted in the previous paragraph, making choices that optimize accuracy of an estimator of ${\alpha}$ or $f_v$ is usually harmful in terms of estimation of the eventual object of interest. As another example, the researcher might select inputs to minimize the integrated mean square error of the estimator for the quantile function of a bidder's valuations, which may be written as ${\alpha}$ evaluated at the quantiles of the win--probability distribution. Because the density of the win--probabilities is often unbounded at the left boundary, the researcher might have to smooth ${\alpha}$ less near the left boundary than away from it in order to ensure integrability of the mean square error. We implement this idea by applying a transformation to the data in conjunction with a kernel--based smoothing method. 

%
%
Our paper relates to recent work on identification in trading mechanisms and estimation of monotone bidding strategies in first--price auctions.
\cite{Larsen2018} operationalize a similar change of variables to prove generic nonparametric identification results in settings where the researcher does not observes the rules of the mechanism and may not directly observe the agent's actions, either, but is willing to assume the data are generated in a Bayes--Nash equilibrium. Thus, their analysis begins one step behind ours in the sense that they estimate the mapping from actions (e.g.~bids) to outcomes (payments and allocations) in a first stage. Not surprisingly, their simulations demonstrate their approach suffers from a large loss of precision compared to estimation strategies that take advantage of prior knowledge of the auction mechanism. We therefore view our respective contributions as complementary advances in identification and estimation of auctions and auction--like mechanisms under shape constraints.

Three recent papers have also considered shape--constrained estimation in first--price auctions. \citet{henderson2012empirical} impose monotonicity on a nonparametric estimator of the inverse bidding strategy---which is equivalent to convexity of the expected payment function---by `tilting' the empirical distribution of the bids, \citet{luo2018integrated} consider an alternative approach that imposes convexity of the integrated quantile function of the bidders' valuations, and \citet{ma2019monotonicity} apply a rearrangement technique to the first step estimator in GPV. The constrained least squares estimator in this paper may be viewed as an extension of \cite{luo2018integrated} to more general auction models with possibly asymmetric bidders, which we achieve by considering the equilibrium expected payment instead of the integrated quantile function. Indeed, both our constrained least squares estimator and the one in \cite{luo2018integrated} can be characterized as the (slope of) greatest convex minorants (GCM), albeit of different functions.  In the case of first--price auctions with two symmetric bidders, the integrated quantile function coincides with the equilibrium expected payment function, and our constrained least squares estimator is numerically equivalent to Luo and Wan's estimator.  More generally, however, the integrated quantile function differs from the expected payment function when there are more than two bidders, and it need not be convex when bidders are asymmetric, because a bidder with a valuation equal to the $\tau$--quantile of its distribution will generally not submit a bid equal to the $\tau$--quantile of its highest competing bid. Thus, the Luo and Wan approach does not apply to asymmetric auctions.  The approach pursued in \citet{ma2019monotonicity} uses the bids instead of the probabilities and hence does not readily extend to other auction mechanisms.

Like the estimator proposed in \citet{luo2018integrated}, neither our constrained least squares estimator nor our nonparametric MLE requires the choice of an input parameter.  Both of our unsmoothed estimators of the equilibrium expenditure function $e$ converge as a process to the same (tight) Gaussian limit process.  The inverse strategy function ${\alpha}$, if the choice variable is the probability of winning, is the derivative of $e$.  Both of our unsmoothed estimators of ${\alpha}$ converge at a $\sqrt[3]{T}$ rate, where $T$ is the number of auctions, and both have a Chernoff limit distribution; this is also true for the estimator in \citet{luo2018integrated}. 
Computation of both of our unsmoothed estimators is simple: the constrained least squares estimator can be computed using an off--the--shelf algorithm and we show that our nonparametric maximum likelihood estimator can be easily computed using a simple pooled adjacent violators algorithm, also.

Although the MLE is asymptotically equivalent to our least--squares alternative, the MLE exhibits finite--sample advantages over the least--squares estimator when the true expected payment function is more convex for large values of $p$, which tends to be the case when bidder one's valuations are relatively strong compared to the maximum of its competitors'. Loosely speaking, the least-squares estimator is biased upward near $p$ equal to one because it is the slope of the GCM of an unconstrained estimator for $e$, with the result that finite--sample noise in the unconstrained estimator for $e$ forces the GCM to ``bow'' outward.\footnote{Illustration of the bowing out issue:
\tikz[baseline=8ex,scale=0.3]{
\draw[thick] (0,5)--(0,0)--(5,0);
\fill[Blue] (0,0) circle(2mm) (1,0.6) circle(2mm) (2,0.75) circle(2mm) (3,2) circle(2mm) (4,2.25) circle(2mm) (5,5) circle(2mm);
\draw (0,0) -- (2,0.75) -- (4,2.25)--(5,5);
\draw[->,Red] (3,4)--(4.4,3.62);
\draw[->,Red] (1.48,2.14)--(2.48,1.14);
\draw (4.5,3.6) node[right]{slope is steep here};
}
} This finite--sample bias is greater when $e$ is more convex. On the other hand, the MLE is less negatively impacted because the MLE for $e$ is not forced to ``bow'' as much as the least--squares estimator. Hence, the MLE can be expected to outperform the least-squares estimator in finite samples in auctions with a small number of symmetric (or approximately symmetric) bidders. 

Our nonparametric maximum likelihood estimator can alternatively be characterized as a two--step estimator, in which the first step is an inverse isotonic regression function estimator that yields bidder one's bid function.  To our knowledge, this is the first direct estimator of the equilibrium bid function itself.\footnote{The estimator that comes closest is \citet{bierens2012semi}, which assumes symmetry and independence, parameterizes the density of valuations and then matches the bid distribution implied by candidate parameter values to the observed bid distribution.  The estimation method is (semi)nonparametric in that the dimension of the parameter vector increases with the sample size, like it is in sieve estimation.}$^,$\footnote{To avoid ambiguity, we refer to the mapping from (to) values to (from) bids as the ``(inverse) bid function'' and the mapping from (to) values to (from) win--probabilities as the ``(inverse) strategy function.''}

Although it is not our primary objective, in the interest of completeness and to facilitate comparison with earlier work, we provide estimators of the quantile function, the distribution function, and the density $f_v$ of valuations.  The quantile function and distribution functions can be estimated using routine operations (such as the delta method) on our estimates of $\alpha$.  As noted by GPV and others, estimating $f_v$ requires nonparametric derivative estimation and the optimal convergence rate is a leisurely $T^{2/7}$.\footnote{To obtain a fourfold improvement requires a data set that's 128 times as large, compared to 32 for typical nonparametric estimators of univariate objects and 16 for parametric estimators.}  We provide estimation results for the derivative ${\alpha}'$ of ${\alpha}$, which indeed converges at the $T^{2/7}$ rate.  There are two ways of estimating $f_v$ using our approach: a two--step procedure in the spirit of GPV and a one--step procedure like \citet{marmer2012quantile}.\footnote{\citet{luo2018integrated} raise the interesting possibility of using a first step unsmoothed estimator as an input to the second stage of GPV, but do not provide asymptotic results.  It is likely that consistency obtains, but the$f_v$ convergence rate and indeed the asymptotic distribution are unknown.}  We do not see any reason to prefer either the one--step or two--step procedure. We provide asymptotic linear expansions of our first step estimators to allow readers to make up their own mind.

Because first--price auction models can be identified when bidders are asymmetric and when only a subset of the bids is observed \citep{Athey2002, Campo2003}, we provide separate results depending on assumptions made about the bids observed by the econometrician. The reason for this flexibility is that, in a first--price auction, bidder one's equilibrium expenditure function $e$ only depends on the distribution of the maximum competitor bid.  We therefore assume that the data are sufficient to obtain an estimate of this distribution and use this (unconstrained) estimator as the starting point for our analysis. If bidders are symmetric and their bids are independent, such an estimator could be $G_T^{n-1}$, where $G_T$ is the empirical distribution of the bids. If bidders are asymmetric then the product of the competitors' marginal empirical bid distributions would be a natural choice. If there is possible dependence among the competitors' bids, either arising from dependence in the competitor's valuations or coordination in their bids, then the empirical distribution of the maximum of the competitors' bids can be used. We show how the asymptotic properties of our constrained estimator improve as we add independence and symmetry assumptions: such improvements can be substantial and depend on the object being estimated.

Our paper addresses many issues and is fairly exhaustive in several dimensions.  Nevertheless, there are several issues that we do not address in the paper.  First, we ignore the potential presence of a (binding) reserve price.  A binding reserve price would affect identification of certain objects,\footnote{For instance, the value distribution below the reserve price.} but for many other objects, allowing for a reserve price would pose a minor, not especially interesting (from an econometric perspective), nuisance.   Further, we do not allow for endogenous entry.  Although endogenous entry can be an important concern in empirical work and raises interesting modeling and identification questions \citep{Levin1994, Li2009, Marmer2013, Gentry2014}, there are many ways of modeling this and it would be beyond the scope of this paper. The same comment applies to possible risk aversion, albeit that risk aversion would likely pose a tougher problem because nonparametric identification of the bidders' utility functions requires an exclusion restriction \citep{Guerre2009}. Finally, there can be auction--level heterogeneity.  Correcting for observed heterogeneity is relatively routine; unobserved heterogeneity might be addressed using methods similar to \citet{Krasnokutskaya2011} or \citet{Roberts2013}.  We leave these questions for future work.

We analyze the performance of our estimators in a fairly extensive simulation study. In general, we find that our estimators perform well and exhibit considerable robustness, both with respect to the design of the simulation study and to the choice of input parameter. However, no clear winner emerges and our various methods differ in systematic ways that are consistent with our asymptotic theory and with intuition. Hence, we do not offer empirical researchers a specific recommendation; rather, we provide a collection of tools, asymptotic results, and general insights that can be applied on a case--by--case basis.

Our paper is organized as follows.  In \cref{sec:model} we describe our model.  \Cref{sec:e and alpha} contains the description of our unsmoothed estimators of the equilibrium expenditure function $e$ and its derivative ${\alpha}$, including a description of their computation.  \Cref{sec:e and alpha} also contains a description of the direct estimate of the bid function.  \Cref{sec:smoothing} introduces and provides results for the smoothed versions of our estimates of ${\alpha}$ and its first derivative ${\alpha}'$, including boundary correction and transformation schemes, plus a description of jackknife estimators.  Results for the estimation of the probability distribution of probabilities under various (a)symmetry and (in)dependence assumptions, can be found in \cref{sec:Fp}.  Then, \cref{sec:objects} contains results on the estimation of several objects of potential interest.  We present our simulation results in \cref{sec:sims}.  Finally, \cref{sec:conclusion} concludes.

\section{First--price auction model with independent private values}
\label{sec:model}

Let $i=1,\dots,n$ index the bidders competing for an object in a first--price, sealed--bid auction. A bidder's value $v_{i}$ is drawn from a distribution $F_{i}$, which takes support on a compact interval in the nonnegative reals.  We assume the seller sets a nonbinding reserve price of zero. We further assume each distribution is absolutely continuous with a density $f_{i}$ that is bounded away from zero on its support.

Each risk--neutral bidder chooses her bid to maximize her expected surplus taking her competitors' strategies as given.  We will take bidder one (1) to be the bidder whose value is to be recovered and use a subscript $c$ to denote her competitors.  Thus, bidder one solves
\begin{equation}
\max_b \cparens[\big]{G_c\parens{b}\parens{v_1 - b}},
\end{equation}
where $G_c\parens{b}$ denotes the probability that bidder one's competitors all bid no more than $b$.\footnote{In principle one can accommodate dependence among bidder one's competitors' bids by treating groups of bidders as individual bidders.  Doing so requires additional assumptions to ensure monotonicity of equilibrium strategies.}

We can equivalently formulate the bidder's problem as a choice of her equilibrium probability of winning:
\begin{equation} \label{eq:profit function of p}
\max_p \cparens{ pv_1 - e_1(p) },
\end{equation}
where $e_1(p) =  Q_c\parens{p} p$ is bidder $i$'s equilibrium expected payment to the seller with $Q_c=G_c^{-1}$ the function.  

The well--known fact that $e_1$ must be convex in $p$ can be seen as a consequence of monotonicity of the equilibrium strategies or incentive compatibility of the direct revelation selling mechanism that implements the Bayes--Nash equilibrium of the first--price auction \citep{Maskin2000}. In any case, the solution to bidder one's problem is illustrated in \cref{fig:menu}. As noted by \citet{Larsen2018} and \citet{Milgrom1982}, bidder one's indifference curves in $(p,e)$-space are represented by straight lines with a slope equal to $v_1$. The optimal expected surplus is therefore attained where $\alpha_1(p) = e'_1\parens{p}=v_1$. 
\begin{figure}[ht]
\begin{center}
\includegraphics{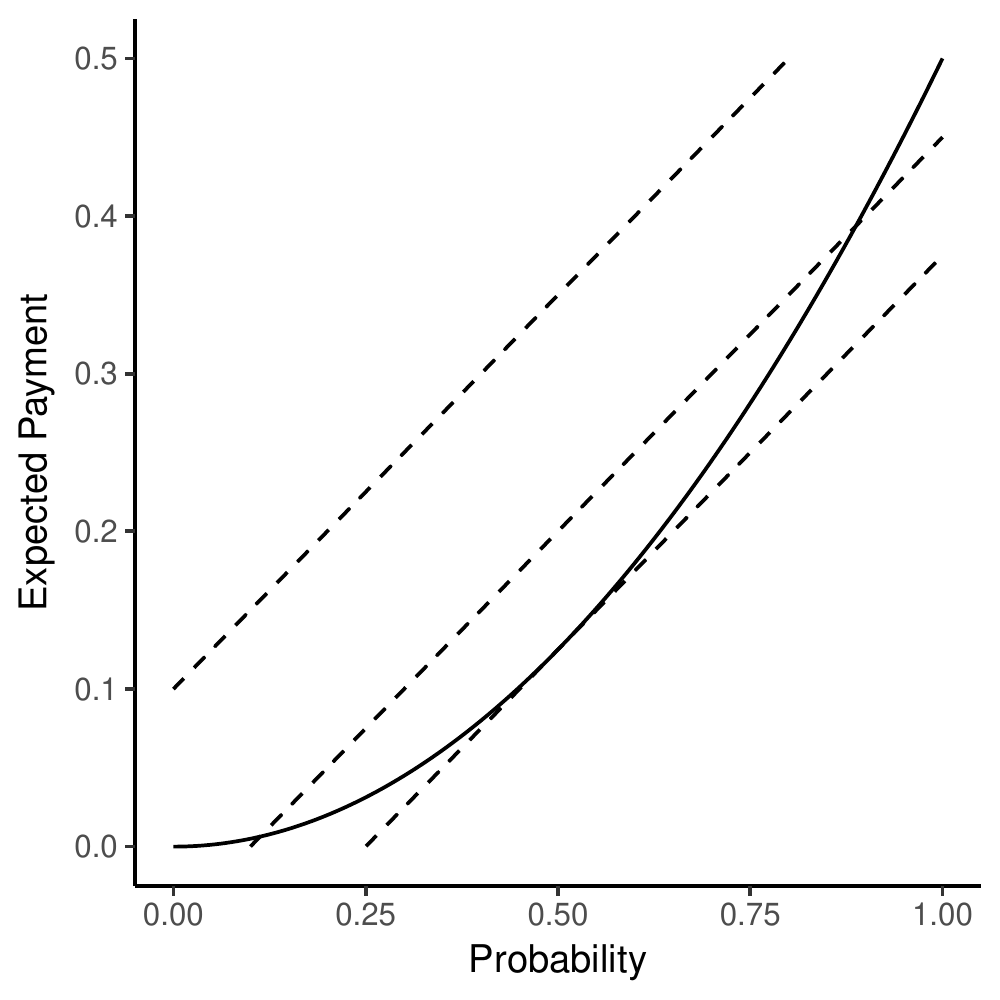}
\caption{A risk--neutral bidder with indifference curves represented by the dashed lines would optimally submit a bid that will win with a probability of 1/2 and expect to pay 1/8 to the seller (unconditional on winning the auction). \label{fig:menu}}
\end{center}
\end{figure}
From here on we drop the subscript from the function $e_1$ and write $e$ to mean $e_1$.

\section{Nonsmooth estimation of $e$ and ${\alpha}$}
\label{sec:e and alpha}
\drop{Let $b_{it}$ denote the bid of bidder $i\in\mathcal{N}$ in auction $t$ for $t=1,\dots,T$, where we have assumed the same set of bidders compete in each auction for ease of exposition.} In order to eliminate conditioning variables in our notation for the competing distribution of bids and equilibrium expected payment function, we assume valuations are independent across bidders, there is no auction--level heterogeneity, and the same set of bidders compete in each auction. 


Under the assumption that bidders' valuations are independent across auctions, each auction is an independent realization of the same game. Therefore, the probability of winning and the expected payment as a function of $b$ can be estimated by $G_{cT}(b)$ and $G_{cT}(b)\,b$, where $G_{cT}$ is a suitable estimate of the distribution function $G_c$ of the maximum of bidder one's competitors' bids.

Though a piecewise linear function $e_T$ whose graph contains
\begin{equation}
\Set[\big]{\parens[\big]{G_{cT}(b), \,b\,G_{cT}(b)}\colon b=b_{1},\dots,b_{T}} \label{eq:unconstrmenu}
\end{equation}
converges to $e$ at a $\sqrt{T}$--rate, it is generally non--convex in finite samples, and the slope of the menu between two nearby points is a poor approximation of its derivative.\footnote{$\sqrt{T} \cparens{e_T\parens{\cdot} - e\parens{\cdot}}$ converges weakly to a Gaussian limit process.} We will show that the greatest convex minorant of $e_T$ can be used to estimate the expected payment function and its derivative in a single step.\footnote{\cite{luo2018integrated} consider a greatest convex minorant estimator of a different function.} Moreover, we show that this estimator can be justified by a least--squares criterion and estimated by isotonic regression.

\subsection{Convexification}

To motivate the least--squares criterion, suppose that a differentiable estimate of the quantile function for bidder one's highest competing bid were available. Multiplying this hypothetical estimator by $p$ would yield a differentiable, though possibly non--convex, estimator, $e_T$. A shape--constrained estimate of the derivative of the expected payment function could then be obtained by solving the following problem
\[
\min_{{\alpha}\in\sA} \parens[\bigg]{ \frac{1}{2}\,\int_0^1 \parens[\big]{{\alpha}\parens{p} - e_T'\parens{p}}^2 \dif p},
\]
where $\sA$ is the set of nondecreasing nonnegative functions defined on $[0,1]$. This least--squares objective can be rewritten as
\[
\frac{1}{2}\,\int_0^1 \parens[\big]{{\alpha}\parens{p}  - e_T'\parens{p}}^2\dif p = \frac{1}{2} \int_0^1 {\alpha}^2\parens{p} \dif p - \int_0^1 {\alpha}(p)\,e'_T(p) \dif p + \frac{1}{2}\int_0^1 e'_T(p)^2 \dif p\,.
\]
The last term does not depend on ${\alpha}$ and may therefore be dropped from the criterion without affecting the shape--constrained estimator. The problem becomes
\begin{equation}
\label{eq:ls}
\min_{{\alpha}\in\sA} \parens[\bigg]{ \frac{1}{2}\int_0^1 {\alpha}^2\parens{p}\dif p - \int_0^1 {\alpha}\parens{p} \dif e_T\parens{p}},
\end{equation}
which can be solved for any $e_T$, differentiable or not, provided that the second integral in \cref{eq:ls} exists.  

In a first-price auction, we use an unconstrained estimate of the empirical quantile function for bidder one's highest competing bid, $Q_{cT}(p)$, and set $e_T(p) = Q_{cT}(p)\,p$ in \cref{eq:ls}. As we noted above, this $e_T$ will generally be non--convex in finite samples and piecewise linear. If $Q_{cT}(p)$ is the empirical quantile function of the highest rival bid, $e_T$ will be discontinuous at $t/T$ for $t = 1,\dots,T$, and the least--squares criterion can be rewritten as
\begin{equation}
\label{eq:lsQT}
\frac{1}{2} \int_0^1 {\alpha}^2(p)\dif p - \sum_{t = 1}^T\int_{\frac{t-1}{T}}^{\frac{t}{T}} {\alpha}(p)\,Q_{cT}(p)\dif p
- \sum_{t=1}^T {\alpha}\parens[\Big]{\frac{t-1}T} \frac{t-1}{T} \cparens[\Big]{Q_{cT}\parens[\Big]{\frac t T} - Q_{cT}\parens[\Big]{\frac{t-1}T}}. 
\end{equation}
where the second integral exists because ${\alpha}$ is bounded and increasing, and $e_T$ is left--continuous. 

Given this representation, we show that the minimizer of \cref{eq:lsQT} over all ${\alpha}\in\sA$ is a right--continuous step--function.
\begin{lem}
	\label{lem:lscharacterization}
	If $e_T$ is piecewise linear in $p$ then the minimizer of the least--squares criterion \eqref{eq:ls} among nondecreasing, nonnegative functions is a right--continuous step--function.  
	\proof
	All proofs can be found in \cref{app:proofs}.\qed
\end{lem}
\begin{figure}[h]\centering
	\begin{tikzpicture}[scale=10]
	\draw[ultra thick] (0,0)--(1,0);
	\foreach \x in {0,0.5,1}
	\draw(\x,0.02)--(\x,-0.02) node[below]{\x};
	\draw(0,0.05) node{$($};
	\draw(0.504,0.05) node{$($};
	\draw(1,0.05) node{$]$};
	\draw(0.496,0.05) node{$]$};
	\draw(0,-0.1) node{$[$};
	\draw(0.504,-0.1) node{$[$};
	\draw(1,-0.1) node{$)$};
	\draw(0.496,-0.1) node{$)$};
	\draw(0.25,0.05) node{$Q_{T1}$};
	\draw(0.75,0.05) node{$Q_{T2}$};
	\draw(0.25,-0.1) node{${\alpha}_{T1}$};
	\draw(0.75,-0.1) node{${\alpha}_{T2}$};
	\end{tikzpicture}
	\caption{\label{fig:alpha}Illustration of the computation of the ${\alpha}_{Tt}$'s for $T=2$.}
\end{figure}
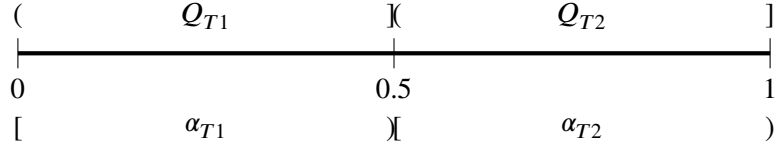
Using \cref{lem:lscharacterization}, and as illustrated in \cref{fig:alpha}, the least--squares problem can be further simplified to
\<\label{eq:ls char}
\min_{{\alpha}_{T1}\leq \dots \leq {\alpha}_{Tt}} \sum_{t=1}^T\parens[\Big]{\frac{1}{2} {\alpha}_{Tt}^2 - Q_{cTt}{\alpha}_{Tt}
- \parens[\big]{Q_{cTt}-Q_{cT,t-1}} \parens{t-1} {\alpha}_{Tt}}\,,
\>
where ${\alpha}_{Tt}= {\alpha}\cparens{\parens{t-1}/T}$ and $Q_{cTt}= Q_{cT}\parens{t/T}$.

 The first term in the summand in \cref{eq:ls char} comes from the integral of ${\alpha}^2$; the second term comes from the integral of ${\alpha}$ with respect to the linear portions of $e_T$; and the third term comes from the integral of ${\alpha}$ at the discontinuities of $Q_{cT}$.


\begin{figure}[h]\centering
\includegraphics{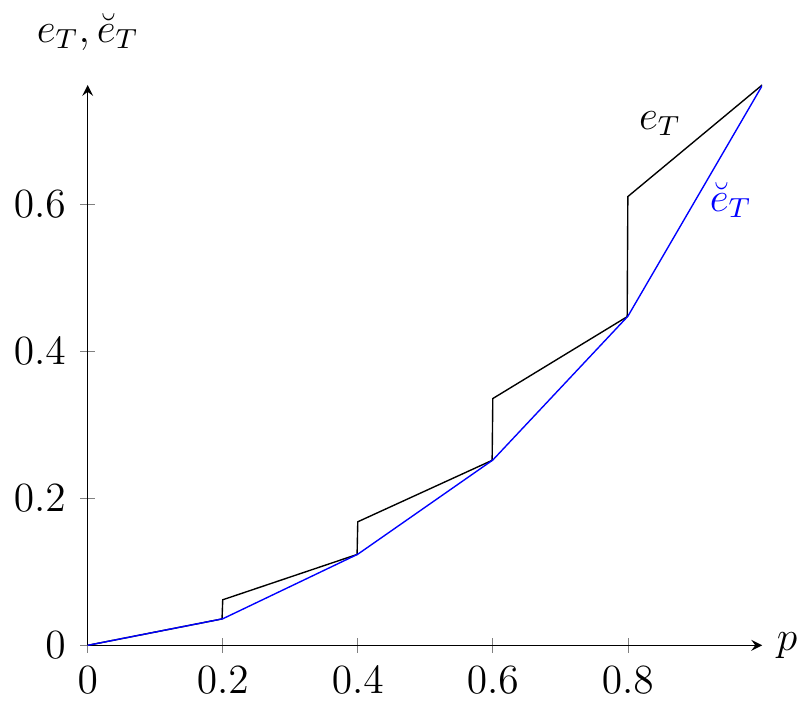}
\caption{\label{fig:convexification}Convexification algorithm illustrated for $T=5$}
\end{figure}

Define ${\alpha}_T\parens{p}= {\alpha}_{T\lceil Tp \rceil}$. We can integrate up ${\alpha}_T$ to obtain a convex estimator $\beT$ of $e$, 
\[
	\beT(p)= \int_0^p \alT(u)\dif u.
\]
As it turns out, $\beT$ is simply the greatest convex minorant of $e_T$, and ${\alpha}_T$ is its left--derivative: see \cref{fig:convexification}.  Note that $\beT$ is both piecewise linear and continuous by construction.

An alternative way to arrive at an estimator ${\alpha}_T^*$ and thence an estimator $\beT^*$ is by defining the problem as an inverse isotonic regression problem.  The idea is to define 
\< \label{eq:ThetaT}
{\Theta}_T\parens{\tilde {\alpha}} = \inf \argmin_p
\sum_{t=1}^T \cparens[\Big]{ e_T\parens[\Big]{\frac{t}{T}} - e_T\parens[\Big]{\frac{t-1}{T}} - \frac{\tilde {\alpha}}{T}}
\one\parens[\Big]{\frac{t-1}{T} \leq p},
\>
The rationale for \cref{eq:ThetaT} is that \cref{eq:ThetaT} essentially imposes monotonicity of the derivative of $e$: it is the natural analog to inverse isotonic regression estimators for the current context.\footnote{An inverse isotonic regression estimator can for a given $m$ be characterized as a minimizer $x$ of $\sum_{i=1}^n \parens{y_i-m} \one\parens{x_i\leq x}$.}  An alternative way of thinking about it is that the population objective function corresponding to \cref{eq:ThetaT} is
\[
\frac1T \sum_{t=1}^T \sparens[\Big]{
T \cparens[\Big]{e\parens[\Big]{\frac tT} - e\parens[\Big]{\frac{t-1}T}} - \tilde {\alpha}
} \one\parens[\Big]{\frac{t-1}T \leq p}
\simeq
\frac1T \sum_{t=1}^T \cparens[\Big]{
	{\alpha}\parens[\Big]{\frac {t-1}T}- \tilde {\alpha}
} \one\parens[\Big]{\frac{t-1}T \leq p},
\]
which is optimized at the value of $p=\parens{t-1}/T$ for which ${\alpha}\cparens{\parens{t-1}/T}$ is the largest value less than $\tilde {\alpha}$.

Returning to \cref{eq:ThetaT}, an estimator ${\alpha}_T^*$ can be defined as
\[
 {\alpha}_T^*\parens{p}= \sup \Set{\tilde {\alpha} \colon {\Theta}_T\parens{\tilde {\alpha}} \leq p}.
\]
There is no a priori reason to prefer ${\alpha}_T^*$ to ${\alpha}_T$ or vice versa, albeit that ${\alpha}_T$ may be easier to compute. In fact, they are numerically equivalent because ${\Theta}_{T}(\tilde \alpha) = \sup\{p:\alpha_{T}(p) < \tilde \alpha\}$.

\drop{The intuition behind this formulation is that for $\Upsilon\parens{\tilde p_1,\tilde p_2}=e\parens{\tilde p_2} - e\parens{\tilde p_1} - \tilde {\alpha} \parens{\tilde p_2-\tilde p_1}$ by \citet[equation (3.3)]{balakrishnan1998order},
\begin{multline} \label{eq:fugly}
 \sum_{t=2}^T \int_0^p\int_{\tilde p_1}^1 \Upsilon\parens{\tilde p_1,\tilde p_2}
f_{p_{t-1},p_t}\parens{\tilde p_1,\tilde p_2} \dif \tilde p_2 \dif \tilde p_1
\\
=
\sum_{t=2}^T\int_0^p\int_{\tilde p_1}^1 
	\Upsilon\parens{\tilde p_1,\tilde p_2}
\frac{T!}{\parens{t-2}!\parens{T-t}!} F_p^{t-2}\parens{\tilde p_1} \cparens{ 1-F_p\parens{\tilde p_2}}^{T-t}
f_p\parens{\tilde p_1}f_p\parens{\tilde p_2}
\dif \tilde p_2 \dif \tilde p_1
\\
= \parens{T-1} T\int_0^p\int_{\tilde p_1}^1 
\Upsilon\parens{\tilde p_1,\tilde p_2}
\cparens{ 1+ F_p\parens{\tilde p_1}-F_p\parens{\tilde p_2}}^{T-2}
f_p\parens{\tilde p_1}f_p\parens{\tilde p_2}
\dif \tilde p_2 \dif \tilde p_1.
\end{multline}
Now, differentiating with respect to $p$ and applying integration by parts yields
\begin{multline}
\parens{T-1}T \int_p^1 \cparens{ e\parens{t}-e\parens{p}-\tilde {\alpha} \parens{t-p}}
\cparens{1+F_p\parens{p}-F_p\parens{t}}^{T-2} f_p\parens{p}f_p\parens{t} \dif t
\\
=
- T  \cparens{ e\parens{1}-e\parens{p}-\tilde {\alpha} \parens{1-p}} f_p\parens{p} F_p^{T-1}\parens{p}\\
	+
	T \int_p^1 \cparens{ {\alpha}\parens{t}-\tilde {\alpha}} f_p\parens{p} \cparens{1+ F_p\parens{p}-F_p\parens{t}}^{T-1} \dif t,
\end{multline}
which for fixed $p<1$ is maximized by a number quickly approaching ${\alpha}\parens{p}$.\joris{Much of this should probably be in an appendix if in the paper at all: not exactly intuitive.  But I don't want to discard it yet.}\karl{I agree that it should not be discarded. An appendix seems appropriate.}
}

If bidders are symmetric then a more efficient unconstrained estimator for the expected payment function is given by $p Q_{T}\parens[\big]{p^{1/(n-1)}}$, where $Q_T$ is the empirical quantile function for the pooled sample of bids $\Set{b_\ell}$ for $\ell = 1,\dots, nT$. In this case, the solution to the least--squares problem in \cref{eq:ls} is found via a weighted isotonic regression of $\cparens[\big]{b_{\parens{\ell}} \ell^{n-1} - b_{\parens{\ell-1}} \parens[\big]{\ell-1}^{n-1}}/\cparens[\big]{\ell^{n-1} - \parens{\ell-1}^{n-1}}$ on $\cparens[\big]{\parens{\ell-1}/{nT}}^{n-1}$ with weights given by $\parens{\ell/nT}^{n-1} - \cparens[\big]{\parens{\ell-1}/nT}^{n-1}$, where $b_{(\ell)}$ denotes the $\ell$--th order statistic. The corresponding constrained estimator for $e$ is the GCM of $p Q_T\parens[\big]{p^{1/(n-1)}}$, as before.

\subsection{Asymptotics for the GCM estimator}

We now develop some asymptotic estimation results for our convex estimator $\beT$.  Before we do so, we will make several assumptions and discuss conditions under which they would hold.
\begin{ass} \label{ass:value distributions}
The private values $v_{t1},\dots,v_{tn}$ in each auction $t$ are independent and drawn from continuous distributions $F_1,\dots,F_n$, respectively.  The distributions have bounded convex supports $[\mathb{v}, \bar{v}]$ and their density functions $f_1,\dots,f_n$ are continuous and nonzero on $(\mathb{v},\bar{v}]$.\footnote{It is standard to model a binding reserve price as an atom at the low end of the value distribution. If there is a binding reserve price $r$, the expected payment function could be redefined as $e(p) = rp$ for $p$ less than $p^*$, the probability that no competitors submit a bid. Although bidder one cannot actually submit a bid so as to win with a positive probability $p<p^*$, a bidder optimizing against this $e$ would choose the corner solution $p=0$ if her valuation is less than $r$ and would be willing to submit a bid of $r$ if her valuation is $r$. Hence, this abuse of notation and terminology is inconsequential. A shape--constrained estimate for $e$ can still be defined as the GCM of $Q_{cT}(p) p$ where $Q_{cT}$ is an appropriate estimate of the quantile function of $\max\cparens{r, B_2,\dots, B_n}$. Essentially, we would treat the reserve price as another competing bid that has a degenerate distribution at $r$. We will not discuss the case of a binding reserve price further in this paper.} There is independence across auctions.\qed  
\end{ass}
\Cref{ass:value distributions} is a standard assumption in the auctions literature.  It is sufficient to ensure the existence of monotone bid functions \citep{Lebrun2006}.  The common support assumption embedded in \cref{ass:value distributions} is unnecessary, but is imposed to make our analysis more wieldy. We note that \cref{ass:value distributions} is stronger than we need: we do not use independence among the competitors' valuations in the proofs of any of our theorems. The assumption can be relaxed provided that bidder one's unique best reply to its competitors' bids is a monotone pure strategy.
\begin{ass} \label{ass:monotone bid strategies}
Bidders are risk neutral and bid according to the (strictly monotonic) Bayes--Nash equilibrium strategies. \qed
\end{ass}
Similarly, \cref{ass:monotone bid strategies} is slightly stronger than we need, because our results only require bidder one's decision problem to be of the form in \cref{eq:profit function of p}. Thus, \cref{ass:value distributions} and \cref{ass:monotone bid strategies} are merely one set of sufficient conditions on the primitives of the model under which our results may be proven.

A consequence of the assumptions made thus far is that $Q_c'$ is continuous and bounded on any closed interval that does not contain zero.  Indeed, the first order condition corresponding to \cref{eq:profit function of p} implies that
\[
 Q_c'\parens{p}= \frac{v-Q_c\parens{p}}{p}.
\]
We now make a high level assumption on the convergence of an estimator of the bid distribution functions and develop conditions under which it is known to hold.
\begin{ass} \label{ass:Gaussian process}
The maximum rival bid distribution can be estimated by $G_{cT}$, for which
$
 \sqrt{T} \cparens{G_{cT}\parens{\cdot}-G_c\parens{\cdot}} \convw \G^*,
$
where $\G^*$ is a Gaussian process with covariance kernel $H^*$ and $\convw$ denotes weak convergence.
\qed
\end{ass}
\Cref{ass:Gaussian process} is a relatively weak assumption, and as previously discussed, is the starting point for our analysis. It is for instance satisfied if we take $G_{cT}$ to be the empirical distribution function of the maximum competitor bid, in which case
\< \label{eq:Hstar simple}
H^*\cparens{Q_c\parens{p},Q_c\parens{p^*}}=\min\parens{p,p^*}-pp^*.
\>
  It would also be satisfied if, instead, we assumed symmetry and took $G_{cT}= G_T^{n-1}$, i.e.\ the empirical bid distribution estimated off all bids raised to the power $n-1$, in which case\footnote{
Note that $\sqrt{T}\parens{G_T - G}$ converges to a Gaussian limit process with covariance kernel
$\sparens{G\cparens{\min\parens{b,b^*}} - G\parens{b}G\parens{b^*}}/n$.  Hence  
$\sqrt{T}\parens{G_T^{n-1}-G^{n-1}}$ converges to a Gaussian limit process with covariance kernel
$G^{n-2}\parens{b} G^{n-2}\parens{b^*} \sparens{G\cparens{\min\parens{b,b^*}} - G\parens{b}G\parens{b^*}}/n$.  Insert $Q_c\parens{p}=Q\parens{p^{1/\parens{n-1}}}$ to get the stated result.
}
 \< \label{eq:Hstar symmetric}
H^*\cparens{Q_c\parens{p},Q_c\parens{p^*}} 
=
\frac{\parens{n-1}^2}{n} \cparens[\big]{
 \min\parens{ p,p^*}^{1/\parens{n-1}} - \parens{pp^*}^{1/\parens{n-1}}
}
\parens{pp^*}^{\parens{n-2}/\parens{n-1}}. 
 \>
 A final example is one in which there is asymmetry plus independence and all bids are observed in which case\footnote{Note that 
     \[ \tag{*}
     \sqrt{T}\parens{G_{cT} - G_c} = \sqrt{T} \parens[\Big]{\prod_{i=2}^n G_{Ti} - \prod_{i=2}^n G_i}
     \simeq \sum_{i=2}^n \sqrt{T}\parens{G_{Ti}-G_i } G_{-i1},  \label{eq:Hstar asymmetric derivation}
 \]
 where $G_{-i1}= \prod_{j\neq i,1}^n G_j$.  The right hand side in \cref{eq:Hstar asymmetric derivation}
 converges weakly to a Gaussian limit process with covariance kernel
 $\sum_{i=2}^n G_{-i1}\parens{b} G_{-i1}\parens{b^*} \sparens{G_i\cparens{\min\parens{b,b^*}}- G_i\parens{b}G_i\parens{b^*}}$, which produces \cref{eq:Hstar asymmetric}, after noting that $G_{-i1}\cparens{Q_c\parens{p}}= p/G_i\cparens{Q_c\parens{p}}$.
}
 \< \label{eq:Hstar asymmetric}
 H^*\cparens{Q_c\parens{p,Q_c\parens{p^*}}} 
 =
 \sum_{i=2}^n \parens[\bigg]{ \frac1{G_i\sparens{Q_c\cparens{\max\parens{p,p^*}}}} - 1} pp^*,
 \>
 where $G_i$ is the bid distribution of bidder $i$.  Note that for $n=2$ \cref{eq:Hstar asymmetric} collapses to \cref{eq:Hstar simple} divided by two since bidder one's bids can also be used in the case of symmetry. We make \cref{ass:Gaussian process} to avoid having to hard--wire a particular set of distributional assumptions into the problem.

\begin{thm}\label{thm:ebreve}
	Let \cref{ass:Gaussian process,ass:value distributions,ass:monotone bid strategies} hold.  Then $\beT$ 
	satisfies 
	\[
	\sqrt{T} \cparens{ \beT\parens{\cdot} - e\parens{\cdot}}\convw
	\G,
	\]
	on $\sparens{0,1}$, 
	where $\G$ is a Gaussian process with covariance kernel
	\[
H\parens{p,p^*}= \zeta\parens{p} \zeta\parens{p^*} H^*\cparens[\big]{ Q_c\parens{p}, Q_c\parens{p^*}}, \label{eq:Hdef}
	\]
	with $\zeta\parens{p} = Q_c'\parens{p} p$ for $p\neq 0$ and ${\zeta}\parens{0}=0$.  
	
    Moreover, for any fixed $0<p<1$, for any closed interval $\sP \subset \parens{0,1}$, if ${\alpha}$ is continuously differentiable then
	$
	\sqrt[3]{T}\max_{p\in \sP} \abs{{\alpha}_T\parens{p}-{\alpha}\parens{p}}= O_p\parens{1}.
	$ 
	\qed
\end{thm}
%
%

It should be noted that weak convergence of quantile processes for distributions with compact support is usually stated on $\parens{0,1}$; see e.g.\ \citet[p308]{vandervaart2000asymptotic}. The reason is that Hadamard--differentiability only obtains on the interior of $\sparens{0,1}$. We prove that that weak convergence of our quantile--related process in fact obtains on $\sparens{0,1}$.  

Cube--root--$T$ convergence of ${\alpha}_T$ is not surprising in view of e.g.\ \citet{kim1990cube}.  Indeed, $\sqrt[3]{T}\cparens{{\alpha}_T\parens{p}-{\alpha}\parens{p}}$ has a Chernoff limit distribution at each fixed $p$.  Further, equations (15) and (16) in \citet{jun2015classical} suggest that 
\< \label{eq:cube root limit result}
 \sqrt[3]{T} \cparens{{\alpha}_T\parens{p}-{\alpha}\parens{p}} \convd
 {\alpha}'\parens{p} \argmax_{t\in \mathbb{R}} \cparens{ \G^\circ\parens{t} - {\alpha}'\parens{p}t^2/2}.
\>
where $\G^\circ$ is a Gaussian process.
A justification and description of the properties of $\G^\circ$ can be found in \cref{app:gcm}.  If \cref{eq:Hstar simple} holds then \cref{eq:cube root limit result} simplifies to
\<\label{eq:cube root limit result simple}
\sqrt[3]{T} \cparens{{\alpha}_T\parens{p}-{\alpha}\parens{p}} \convd \sqrt[3]{4 \zeta^{2}\parens{p} \alpha'\parens{p}}\,\mathbb{C}\,,
\>
where $\mathbb{C}$ is a standard Chernoff--distributed random variable.  \Cref{eq:cube root limit result simple} is also justified in \cref{app:gcm}. 


 Our result for ${\alpha}_T$ in \cref{thm:ebreve} extends the convergence rate result to uniform convergence. Note that this is in contrast to e.g.\ nonparametric kernel regression or density estimation where uniform convergence only obtains at a slower rate.

\subsection{Nonparametric maximum likelihood}
\label{sec:npmle}

\subsubsection{Estimator}

To this point we have relied on the least--squares criterion used to motivate the GCM estimator for $e$ and the isotonic regression estimator for its derivative.  The GCM has the feature that it may be broadly applied in any auction or auction--like setting as long as an appropriate unconstrained estimate $e_T$ is available. In this section, we develop an estimator based upon a nonparametric likelihood criterion that specifically exploits the structure of a first--price auction.  For ease of exposition and notation, we assume that only the maximum competitor bid is used to construct the likelihood, though we note how more data may be used to produce a more efficient estimate in \cref{sec:inv reg mle}. 

We rearrange the familiar formula for a bidder's inverse strategy function
\[
{\alpha}\cparens{G_{c}\parens{b}} = b + \frac{G_c\parens{b}}{g_c\parens{b}},
\]
in order to relate the density of a bidder's highest competing bid to her expected payment function:
\[
g_c\parens{b} = \frac{e\cparens{G_c\parens{b}} / b }{{\alpha}\cparens{G_c\parens{b}} - b}\,.
\]
The loglikelihood of an independent sample of highest competing bids may then be written as
\begin{equation}
\label{eq:npmle}
\sL\parens{\tilde {\alpha},\tilde e}=
\sum_{t=1}^T \cparens[\big]{\log \tilde e_{t} - \log b_{t} - \log \parens{\tilde {\alpha}_{t}-b_{t}}}\,,
\end{equation}
where the $b_t$'s are the maximum competitor bid and the shorthand forms $\tilde e_t$ and $\tilde {\alpha}_t$ are candidate values for $e\cparens{G_c\parens{b_t}}$ and the left--derivative of $e$ evaluated at $G_c\parens{b_t}$, respectively. Before \cref{eq:npmle} may be used as the basis for a nonparametric maximum likelihood estimator, a few remarks are in order.  First, nondecreasing convex real--valued functions defined on $[0,1]$ are continuous on $[0,1)$,\footnote{A nondecreasing convex function defined on a compact interval can jump discontinuously at the right boundary.} which implies that $e$ is uniquely determined on $[0,1)$ by its left-derivative $\tilde {\alpha}$. We will therefore replace $\tilde e$ with a function of $\tilde {\alpha}$ in what follows. Second, for a given $\tilde e$, the implied $G_{c}$ will be a proper distribution function if $\tilde e$ is convex and $\tilde e\parens{1}$ equals the highest order statistic among the rivals' observed bids. Third, because the loglikelihood contribution of $b_{t}$ is increasing in $\tilde e_{t}$ and decreasing in $\tilde {\alpha}_{t}$, the shape--constrained MLE should be piecewise linear in order to minimize the density at values of $b$ between realizations of the competitors' bids while maximizing the density at the observed bids. In particular, kinks in the MLE $\beTml$ occur precisely where $\beTml\parens{p}/p=b_{t}$ for some observed bid $b_{t}$. We can therefore maximize \eqref{eq:npmle} by searching over the space of left--continuous, nondecreasing step functions $\tilde {\alpha}$ defined on the unit interval.

To facilitate the numerical optimization of the maximum likelihood objective in \eqref{eq:npmle}, we introduce the notation $b_{(t)}$ for the $t$-th lowest order statistic and use the fact that $\beTml$ is linear on the interval $[e_{(t-1)}/b_{(t-1)}, e_{(t)}/b_{(t)}]$ to express $e_{\parens{t}}$ in terms of its derivative and $e_{\parens{t-1}}$
\[
e_{(t)} = e_{(t-1)} + {\alpha}_{(t)}\parens[\Big]{\frac{e_{(t)}}{b_{(t)}} - \frac{e_{(t-1)}}{b_{(t-1)}}}
 = e_{(t-1)}\frac{{\alpha}_{(t)}/b_{(t-1)}-1}{{\alpha}_{(t)}/b_{(t)}-1}\,.
 \]
 Combining this recursive relationship with the constraint that $e_{(T)} = b_{(T)}$, we may write $e_{(t)}$ as
 
\begin{equation} \label{eq:e in terms of alpha}
e_{(t)} = b_{(T)}\,\prod_{s = t+1}^{T} \frac{{\alpha}_{(s)}/b_{(s)}-1}{{\alpha}_{(s)}/b_{(s-1)}-1}\,,
\end{equation}
where the product $\prod_{s = t+1}^{T} a_s$ is defined equal to one for $t=T$ for any sequence $\{a_s\}$. Using this expression to replace $e_{(t)}$ in \eqref{eq:npmle}, the loglikelihood becomes
\<
\label{eq:npmle2}
\sL(\tilde {\alpha};b) = 
\sum_{t=1}^{T}  \parens[\Bigg]{\log b_{(T)} + \sum_{s=t+1}^{T} \cparens[\Big]{\log \parens{\tilde {\alpha}_{(s)}/b_{(s)} - 1} - \log \parens{\tilde {\alpha}_{(s)}/b_{(s-1)}-1}} - \log b_{(t)} - \log \parens{\tilde {\alpha}_{(t)}-b_{(t)}}}\,.
\>

By inspection of the above display, the MLE must satisfy ${\alpha}_{\parens{1}} = b_{\parens{1}}$ and ${\alpha}_{\parens{t}}\geq b_{\parens{t}}$. Problematically, this implies an unbounded density at the lower end of the competitor bids' support, which in turn implies that the solution to the maximum likelihood problem is not unique, since the loglikelihood criterion is infinite for any ${\alpha}$ with ${\alpha}_{\parens{1}} = b_{\parens{1}}$ and ${\alpha}_{\parens{2}}>b_2$. Nonetheless, one maximizer of the likelihood distinguishes itself from the rest because, for a fixed ${\alpha}_{\parens{1}}$ and ${\alpha}_{\parens{2}}$ with $b_{\parens{1}} < {\alpha}_{\parens{1}},\, b_{\parens{2}} <{\alpha}_{\parens{2}}$ and ${\alpha}_{\parens{2}}< 2\,b_{\parens{3}} - b_{\parens{2}}$, the solution for $\{{\alpha}_{(t)}\}$ for $t=3,\dots,T$ is unique. Furthermore, this unique solution does not depend on the values of ${\alpha}_{\parens{1}}$ and $\alpha_{\parens{2}}$ because the loglikelihood is additively separable in ${\alpha}_{(t)}$ and the monotonicity constraints on $\alpha_{\parens{t}}$ do not bind for $t=1$, 2, and 3. Thus, we may first maximize $\sL\parens{\tilde {\alpha};b_1,\dots,b_T} - \log \parens{\tilde {\alpha}_{(1)}-b_{(1)}}$ over $\{\tilde{{\alpha}}_{\parens{t}} : t>2\}$. We may then separately define $\breve{{\alpha}}_{T,\parens{1}}^\text{MLE} = b_{\parens{1}}$ and choose any ${\alpha}_{\parens{2}}\in(b_{\parens{2}},\breve{{\alpha}}_{T,\parens{3}}^\text{MLE} ]$. Within this (shrinking) interval, the likelihood contribution of the second--lowest observed competitor bid is strictly decreasing in $\tilde {\alpha}_{\parens{2}}$. In practice, we suggest defining the MLE equal to the boundary value $\breve{{\alpha}}_{T,\parens{2}}^\text{MLE}=b_{\parens{2}}$. \drop{Although this choice results in an indeterminate loglikelihood contribution from the lowest observed competitor bid, one can generally choose an $\epsilon>0$ and define $\breve{{\alpha}}_{T,\parens{2}}^\text{MLE}=b_{\parens{2}}+\epsilon$ such that the estimates of the eventual objects of interest are indiscernible up to machine precision.}

\subsubsection{Pooled--adjacent--violator algorithm (PAVA) for MLE}
By adding and subtracting $\log b_{(s)}$ and $\log b_{(s-1)}$ and canceling terms in the inner summation of \cref{eq:npmle2}, we can rewrite the loglikelihood as
\begin{multline*}
\sum_{t=1}^{T}\sparens[\bigg]{\log b_{(T)} + \\
	\sum_{s={t+1}}^{T}\cparens[\Big]{\log \parens{\tilde {\alpha}_{(s)} - b_{(s)}} -\log b_{(s)}  - \log\parens{\tilde {\alpha}_{(s)}-b_{(s-1)}}+ \log b_{(s-1)}}
 - \log b_{(t)} -\log (\tilde {\alpha}_{(t)}-b_{(t)})}\\
= \sum_{t=1}^{T}\sparens[\bigg]{\sum_{s={t+1}}^{T}\cparens[\Big]{\log \parens{\tilde {\alpha}_{(s)} - b_{(s)}} - \log\parens{\tilde {\alpha}_{(s)}-b_{(s-1)}}} -\log \parens{\tilde {\alpha}_{(t)}-b_{(t)}}}\\
 = \sum_{t=1}^{T}\cparens[\big]{ \parens{t-2} \log\parens{\tilde {\alpha}_{(t)}-b_{(t)}}- \parens{t-1} \log\parens{\tilde {\alpha}_{(t)}-b_{(t-1)}}}\,.
\end{multline*}

The Lagrangian for the isotonic maximum likelihood problem is then\footnote{We omit the constraint $\tilde {\alpha}_{\parens{3}}>\tilde {\alpha}_{\parens{2}}$ from the Lagrangian because this constraint is always slack.}
\begin{equation}
\label{eq:isonpmle}
\max_{\Set{\tilde {\alpha}_{t}}_{t>2}} \sum_{t=1}^{T} \cparens[\big]{\parens{t-2} \log\parens{\tilde {\alpha}_{(t)}-b_{(t)}}-\parens{t-1}\log\parens{\tilde {\alpha}_{(t)}-b_{(t-1)}}} + \lambda_2\parens{\tilde {\alpha}_{\parens{2}}- b_{\parens{2}}} + \sum_{t=4}^T \lambda_{t}\parens{\tilde {\alpha}_{(t)}-\tilde {\alpha}_{(t-1)}}\,.
\end{equation}
This problem can be solved using a pool--adjacent--violators algorithm (PAVA), which divides the large optimization problem into a sequence of at most $T-3$ one-dimensional optimizations. To see this, we observe that the Karush--Kuhn--Tucker (KKT) conditions for this problem are
\begin{equation} \label{eq:npmle Lagrange foc}
\maligned{
\frac{t-2}{\tilde {\alpha}_{(t)} - b_{(t)}} - \frac{t-1}{\tilde {\alpha}_{(t)} - b_{(t-1)}}  + \lambda_{t} - \lambda_{t+1} &= 0,\\
\lambda_{t} &\geq 0,\\
 \tilde {\alpha}_{\parens{t}}-\tilde {\alpha}_{(t-1)} &\geq 0,\\
\lambda_{t}\parens{\tilde {\alpha}_{(t)}-\tilde {\alpha}_{(t-1)}} &= 0.
}
\end{equation}
Let $t_{j}$ be the subsequence of starting points of ``blocks'' for which the nondecreasing constraint binds. By construction, $\tilde {\alpha}_{(t_{j}-1)} < \tilde {\alpha}_{(t_{j})} = \dots = \tilde {\alpha}_{(t_{j+1} -1)} < \tilde {\alpha}_{(t_{j+1})}$. Complementary slackness then implies $\lambda_{t_{j}} = \lambda_{t_{j+1}} = 0$.  Within each block $j$, the value $\tilde {\alpha}$ that satisfies the KKT conditions can then be found by solving for $\tilde {\alpha}$ in 
\begin{multline*}
0 = \sum_{t = t_{j}}^{t_{j+1}-1} \parens[\Big]{\frac{t-2}{\tilde {\alpha}-b_{(t)}} - \frac{t-1}{\tilde {\alpha}-b_{(t-1)}} + \lambda_{t} - \lambda_{t+1}}
 = \sum_{t = t_{j}}^{t_{j+1}-1} \parens[\Big]{\frac{t-2}{\tilde {\alpha}-b_{(t)}} - \frac{t-1}{\tilde {\alpha} - b_{(t-1)}}} \\
= -\frac{t_{j}-1}{\tilde {\alpha}-b_{t_{j}-1}} - \frac{2}{\tilde {\alpha} - b_{(t_{j})}} - \frac{2}{\tilde {\alpha} - b_{(t_{j}+1)}} -\dots - \frac{2}{ \tilde {\alpha}- b_{(t_{j+1}-2)}} + \frac{t_{j+1}-3}{\tilde {\alpha}-b_{(t_{j+1}-1)}}
\end{multline*}

To find the solution to the constrained NPMLE, we initially assign each $\tilde {\alpha}_{t}$ to its own block and set $\Set{\tilde {\alpha}_{t}}$ equal to the unconstrained solution $\tilde {\alpha}_{(t)}= \parens{t-1} b_{(t)} - \parens{t-2}b_{(t-1)}$ for $t>1$ and $\tilde \alpha_{(1)} = b_{(1)}$. This initial guess satisfies the constraint $\alpha_{\parens{t}}\geq b_{\parens{t}}$ but might not produce a monotonic sequence. Beginning with $t = 4$, the PAVA proceeds sequentially by finding the smallest $t$ such that $\tilde {\alpha}_{t}<\tilde {\alpha}_{t-1}$. If such a $t$ exists, we pool $\tilde {\alpha}_{t}$ together with the left adjacent block and recalculate $\tilde {\alpha}$ in the above first--order condition for that block. We then set $\tilde {\alpha}_{(s)}=\tilde {\alpha}$ for all $s$ in the block and repeat until no further violations are found. This algorithm will converge in no more than $T-3$ steps, because exactly one more of the $T-3$ monotonicity constraints are made to bind with equality in each step and no constraints are ever made slack again.

Importantly, every iterate satisfies the dual feasibility KKT condition $\lambda_t\geq0$ because violations of the primal feasibility condition $\tilde{{\alpha}}_{\parens{t}}\geq \tilde {\alpha}_{\parens{t-1}}$ are resolved by imposing $\tilde{{\alpha}}_{\parens{t}}=\tilde{{\alpha}}_{\parens{t-1}}$. Though primal feasibility may also be satisfied, for instance, by setting $\tilde {\alpha}_{\parens{t}}=\tilde{{\alpha}}_{\parens{t+1}}$, these deviations from the PAVA algorithm typically lead to a violation of dual feasibility unless the PAVA algorithm would have pooled these values in a later iteration. Thus, the final iterate of $\Set{\tilde {\alpha}_t}$ will satisfy the KKT conditions.  \Cref{lem:dual active set} formally establishes this claim in an appendix. Moreover, \cref{lem:invexity} demonstrates that the KKT conditions are both necessary and sufficient for the constrained global maximum of the loglikelihood objective. Thus, the algorithm converges to the MLE for $\alpha$. 

\begin{thm} \label{thm:PAVA converges to MLE}
The final iterate of the PAVA described above is the nonparametric MLE for ${\alpha}$. \qed
\end{thm}

We next obtain the NPMLE of the equilibrium payment function by substituting $\balTml$ into equation \cref{eq:e in terms of alpha}. Figure \ref{fig:mle vs gcm} depicts the maximum likelihood estimator for $e$ in comparison with $\beT$ for a sample of five rival bids. In larger samples, the differences in the estimators for $e$ are not visually apparent.
 
As can be seen in \cref{fig:mle vs gcm}, the MLE is invariably above the GCM estimator.  This is no coincidence. The nodes of the GCM are positioned at integer multiples of $1/T$ by construction, whereas the MLE can move the position of the nodes as well as the values at the nodes. The MLE can therefore achieve both convexity and proximity to the original nonconvex estimator without having to duck below the original estimator everywhere.

\grey{This figure can also be used to provide a graphical interpretation of the MLE. In $(p,e)$--space, a bid $b$ is represented by a ray through the origin with a slope of $b$. In \cref{fig:mle vs gcm}, the rays corresponding to the observed highest competing bids are apparent because the graph of the unconstrained estimator for $e$ is composed of segments of these rays. One can show that the shape--constrained MLE for $e$ is the convex, piecewise--linear function that maximizes the sum of the log of the horizontal distances between the maximal point of intersection between $e^{\text{MLE}}_T$ and the rays with slopes equal to the observed bids. Intuitively, this follows because the probability of observing a bid less than or equal to $b$ is represented by the abscissa of the maximal point of intersection between $e$ and the ray with slope $b$. Thus, the log likelihood of the observed bids is maximized by spacing these intersections as evenly as possible (roughly speaking) between zero and one (inclusive) while maintaining convexity. On the other hand, by construction, the points of intersection between the bid rays and the GCM will be more widely spaced for the smallest bids and tightly spaced for the largest bids. One can therefore imagine arriving at the MLE by ``sliding'' the nodes of the GCM back toward the origin along the bid rays, possibly creating new nodes along the way. This explains why the graph of the MLE for $e$ always lies above the graph of the GCM.}


\begin{figure}[ht]
\begin{center}
\includegraphics[width=5in,height=5in]{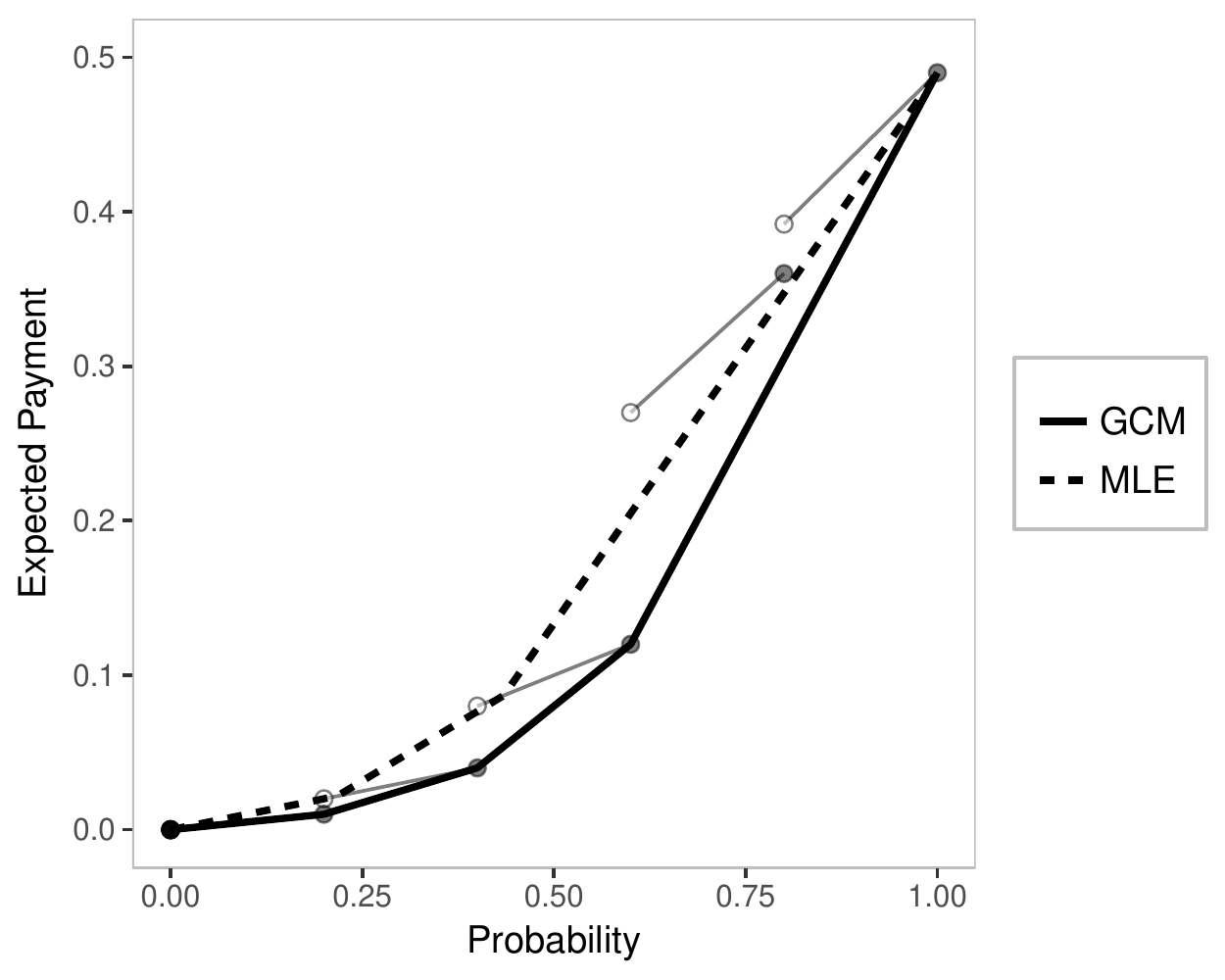}
\caption{An illustrative example of the maximum likelihood estimator (dashed) for the expenditure function compared to the greatest convex minorant $\beT$ (solid) of the unconstrained estimator $e_T$ (grey). \label{fig:mle vs gcm}}
\end{center}
\end{figure}

\subsubsection{Alternative characterization of the MLE}
\label{sec:inv reg mle}

An estimate ${\beta}_T\parens{v}$ of bidder one's bid function at $v$ can  be obtained as the minimizer of
\[
\S_{T}\parens{b,v} =  \sum_{t=1}^{T}\parens[\Big]{\frac{t-2}{v-b_{\parens{t}}} - \frac{t-1}{v-b_{\parens{t-1}}}}\one\parens{b_{(t)}\leq b}\,.
\]
The estimator ${\beta}_T$ is monotonic and it can be inverted to obtain an estimator ${\alpha}_T^\mle$ of ${\alpha}$ at $p=G_c\parens{b}$.  Indeed,
it turns out that both 
\(
 \sqrt[3]{T} \cparens{{\beta}_T\parens{v}-{\beta}\parens{v}} 
\)
and 
\(
 \sqrt[3]{T} \cparens{{\alpha}_T\parens{p}-{\alpha}\parens{p}}
\)
have limiting Chernoff distributions.  Indeed, we have 
\< \label{eq:npmle limit distribution}
 \forall 0<p<1: \sqrt[3]{T} \cparens{{\alpha}_T\parens{p}-{\alpha}\parens{p}} \convd \sqrt[3]{4 {\zeta}^2\parens{p} \cparens[\big]{2Q_c'\parens{p}+Q_c''\parens{p}p}} \C,
\>
where $\C$ is a standard Chernoff distribution.  A sketch of the proof and a derivation of the limit distribution can be found in \cref{app:npmle}.  The limit distribution in \cref{eq:npmle limit distribution} is the same as that in \cref{eq:cube root limit result simple} under \cref{eq:Hstar simple}.

Presumably the limit distribution of the least--squares and maximum likelihood estimators would also coincide when we use all bids, not just the maximum rival bid. In the case of $n$ symmetric bidders, the likelihood of the pooled sample of bids is obtained from $\parens{n-1} g(b) = G\parens{b}/\parens{v-b}$, which becomes $\cparens{e\parens{p}/p}^{1/\parens{n-1}}/\cparens{{\alpha}\parens{p} - Q\parens{p^{1/\parens{n-1}}}}$ after a change of variables. The MLE can then be computed by applying the PAVA to the objective
\[
\sum_{\ell=1}^{nT}\cparens[\Bigg]{ \frac{\ell-n}{n-1} \log\parens[\big]{\alpha_{\parens{\ell}} - b_{\parens{\ell}}} - \frac{\ell-1}{n-1}\log\parens[\big]{\alpha_{\parens{\ell}} - b_{\parens{\ell-1}}}}\,.
\]
The maximum likelihood estimator for the bid function at a fixed $v$ is then given by the minimizer over $b$ of
\[
\sum_{\ell=1}^{nT}\cparens[\Bigg]{ \frac{\ell-n}{\parens{n-1}\parens{v - b_{\parens{\ell}}}} - \frac{\ell-1}{\parens{n-1}\parens{v - b_{\parens{\ell-1}}}}}\one\parens{b_{\parens{\ell}}\leq b}\,.
\]
The maximum likelihood estimator using the full vector of bids in the asymmetric case is considerably more complicated.

\drop{
\subsubsection{Relation to empirical likelihood estimation}
Substituting the identities $p_t = e_t/b_t$ and into the likelihood criterion and constraints, we may express the problem in terms of $p_t$ as
\[
\max_{\Set{p_t}} \log \parens{p_{\parens{t}} - p_{\parens{t-1}}}
\]
subject to $p_{\parens{1}} = 0$ and $p_{\parens{T}}=1$ and an isotonicity constraint on
\[
\frac{p_{\parens{t}} b_{\parens{t}} - p_{\parens{t-1}} b_{\parens{t-1}}}{p_{\parens{t}} - p_{\parens{t-1}}}\,.
\]
Because the difference $p_{\parens{t}} - p_{\parens{t-1}}$ is an estimate of the probability that bidder 1's highest competing bid is in the interval $(b_{\parens{t-1}},b_{\parens{t}}]$, the NPMLE for ${\alpha}$ in a first-price auction is closely related to the distribution of highest competing bids that diverges least from the empirical measure and for which the collection of points \cref{eq:unconstrmenu} lie on the graph of a convex function. Specifically, the measure of divergence is a limit case of the power divergence considered in \cite{henderson2012empirical}, although the isotonicity constraints in our problem differ from the constraints in theirs.

Though standard asymptotic theory for empirical likelihood estimators assumes a finite dimensional parameter of interest and only admits a finite number of equality constraints, the weak convergence of $e_T^{\text{MLE}}$ in first-price auctions\karl{Still need to prove this!} gives hope for an empirical likelihood approach under alternative measures of divergence and/or in other auction formats.\karl{Or is weak convergence of this empirical likelihood problem just a coincidence?} We leave this direction of research for future work.

}

\drop{

\subsubsection{Convergence}

We conclude \cref{sec:npmle} by showing that the limit distribution of our NPMLE estimator of ${\alpha}$ coincides with that of our GCM estimator of ${\alpha}$.\joris{Do we need to say more here?}

\< \label{eq:npmle cube root limit result}
\sqrt[3]{T} \cparens{{\alpha}_T\parens{p}-{\alpha}\parens{p}} \convd
{\alpha}'\parens{p} \argmax_{t\in \mathbb{R}} \cparens{ \G^\circ\parens{t} - {\alpha}'\parens{p}t^2/2}.
\>
where $\G^\circ$ is a Gaussian process with the same properties as the one in \cref{eq:cube root limit result}. If \cref{eq:Hstar simple} holds then \cref{eq:npmle cube root limit result} simplifies to
\<\label{eq:npmle cube root limit result simple}
\sqrt[3]{T} \cparens{{\alpha}_T\parens{p}-{\alpha}\parens{p}} \convd \sqrt[3]{4 \zeta^{2}\parens{p} \alpha'\parens{p}}\,\mathbb{C}\,,
\>
which is the same limit as the one in \cref{eq:cube root limit result simple}.
}
\section{Smoothing, transformations, and boundary correction}
\label{sec:smoothing}
Though not specific to shape--constrained estimation, this paper would be incomplete if it did not also address smoothing because, for objects of interest such as the valuation distribution, smoothing improves the convergence rate of estimators of $\alpha$. Indeed, if $\alpha$ is twice continuously differentiable then the $\sqrt[3]{T}$ convergence rate of the unsmoothed estimators of $\alpha$ can be improved to the standard nonparametric $T^{2/5}$ rate.  In this section, we first introduce our basic smoothing method, which is similar to that in \citet{luo2018integrated}, then develop two important enhancements: boundary correction and transformation.

As noted by \citet{hickman2015replacing}, boundary correction can be important in the estimation of auction models, especially if the objective is to estimate the density of valuations.  The reason is that the bid distribution (in \citet{hickman2015replacing}) or the distribution of win probabilities (here) has compact support and it is well--known that, absent a boundary correction, most nonparametric density estimators are inconsistent at the boundaries.  The situation is more favorable in our case since we know that probabilities vary from zero to one whereas the top of the bid distribution must be estimated, albeit that this can be done super--consistently.  We provide two distinct boundary correction methods, one based on boundary kernels, and one on a boundary correction scheme in the spirit of \citet{hickman2015replacing}.  As expected, both methods yield vast improvements on the performance of our uncorrected estimators near the boundary.  In developing these methods, we have identified an improvement in the choice of the bandwidth sequence recommended in \citet{karunamuni2008some}, which improves the performance of \citet{hickman2015replacing}'s version of the \citet{Guerre2000} estimator substantially.  This improvement is described in a separate paper, \citet{pinkse2019actual}.

Our smoothed estimators for $\alpha$ can be further improved by applying a transformation $\psi$ to the win--probabilities as part of the smoothing method.  Indeed, we show that such transformations $\psi$ can improve the first order asymptotic mean square error, though not the convergence rate, of our smoothed estimators of $\alpha$.  The effect of such transformations on first order asymptotics help explain the feature noted in \citet{ma2019inference} that the asymptotic variance of the \citet{marmer2012quantile} quantile--based--estimator of $f_v$ is often greater than that of the corresponding GPV estimator. These transformation methods are complements, not substitutes, to our boundary correction methods.

\subsection{Smoothing}
The main limitation of our method above is that $\alpha_T$ converges at a $\sqrt[3]{T}$ rate.  This is due to the fact that $\alpha_T$ is discontinuous and hence that $\breve e_T$ is kinky.  To obtain convergence at the typical nonparametric $n^{2/5}$ rate, we can replace $\breve e_T$ with a smoothed version $\hat e_T$, defined by,
\[
\hat e_T\parens{p} =  \frac1h\int_{-\infty}^{\infty} \breve e_T\parens{s}\,
k\parens[\Big]{\frac{s-p}{h}} \dif s,
\]
where $k$ is a twice continuously differentiable kernel with compact support for which $\lazyint{k\parens{s}s^2}=1$,\footnote{$\int k\parens{s} s^2 \dif s=1$ is a normalization.} and $h=h_T$ is a bandwidth such that $\Xi=\lim_{T\to \infty} \sqrt{Th^5} <\infty$.  The restriction on the bandwidth sequence is not necessary for consistency of $\hat e_T$:  unlike kernel--estimators employed by others, $\hat e_T$ is a consistent estimator of $e$ for all bandwidth sequences that converge to zero and the same is true for $\ahT$ defined below.

\drop{This definition of $\hat e_T$ requires that we extend the definition of $\breve e_T$ outside $\sparens{0,1}$, which we do recursively.\joris{Stick in graph}  Later, we propose a another estimator of $e$ based on boundary kernels and a third alternative based on a combination of reflection and transformation. Our new definition of $e$ for any $t>0$ is
\[
\maligned{
	e\parens{1+t} & = 4 e\parens{1} - 6 e\parens{1-t} + 4 e\parens{1-2 t} -e\parens{1- 3 t},\\
	e\parens{-t} & =   4 e\parens{0}  - 6 e\parens{t} + 4 e \parens{2 t} - e\parens{3 t},
}\quad t>0.
\]
We similarly extend ${\alpha}_T$ using

\[
\maligned{
 \alpha_{T}(1+s)& = 6 \alpha_T\parens{1-s} - 8 \alpha_T\parens{1-2s} + 3\alpha_T\parens{1-3s} , \\
 \alpha_{T,1-j}& = 6 \alpha_{T}\parens{s} - 8 \alpha_{T}\parens{2 s} + 3 \alpha_T\parens{3s},
}
\quad j=1,2,\dots.
\]
Then, as before, $\alpha_T\parens{p} = \alpha_{T\lceil Tp\rceil}$  and
\[
 \beT\parens{p}=\int_0^p \alpha_T\parens{t} \dif t,
\]
except that the formula now works for $p\in \mathbb{R}$. Note that $\beT$ is convex and increasing on a compact interval containing $[0,1]$.
\karl{I favor putting all boundary correction discussion in one section. We can move this to that section or cut it from the paper.}
}

This definition of $\hat e_T$ requires modification near the boundaries because $\beT$ is not defined outside $\sparens{0,1}$. We address this issue in \cref{section:boundary correction}. Before providing results for smoothed estimates of $\alpha$ evaluated away from the boundary, we need one further assumption.
\begin{ass} \label{ass:Qc differentiable}
$Q_c$ is thrice continuously differentiable on any closed interval $\sP_Q \subset\halfopen{0,1}$. \qed
\end{ass}
\Cref{ass:Qc differentiable} is essentially equivalent to assuming that $g_c=G_c'$ is twice continuously differentiable, which is implied by continuous differentiability of the value densities \citep{Guerre2000}. Thus, assuming one more continuous derivative in \cref{ass:value distributions} is sufficient for \cref{ass:Qc differentiable}. Assuming that a density is twice continuously differentiable is standard in the nonparametric kernel estimation literature.  The compact subset requirement is needed since $g_c\parens{0}$ may be infinite.

\begin{thm} \label{thm:ehat} Under \cref{ass:Gaussian process,ass:monotone bid strategies,ass:value distributions,ass:Qc differentiable},
$\hat e_T$ is convex, has the same limit properties as $\breve e_T$ on a closed interval $\sP$ contained in $\parens{0,1}$,\joris{Why exclude $p=1$?} and\footnote{The function $\ahT$ converges as a process in an $h$--neighborhood of $p$, but not on any interval of positive length.  Regular kernel density and regression function estimates have the same lack--of--tightness property.  The local tightness result is not in this paper.  A consequence of the lack of tightness is that $\ahT\parens{p}$ and $\ahT\parens{p^*}$ are asymptotically independent for fixed distinct $p,p^*$.}
\item
\[
\forall p \in \sP:
 \sqrt{Th} \cparens[\big]{ \ahT\parens{p} - {\alpha}\parens{p}} \convd
 N\cparens{a''\parens{p}\Xi/2, \sV  },
\]
where
$\ahT\parens{p}=\ehat'\parens{p}$ and
\[
\sV\parens{p} = \lim_{h\to 0} \int_{-\infty}^\infty \int_{-\infty}^\infty 
	\sH_h\parens{p,s,\tilde s}
 k'\parens{s} k'\parens{\tilde s}\dif \tilde s \dif s,	 
\]
where
\<\label{eq:ugly H}
\sH_h\parens{p,s,\tilde s} = \cparens[big]{H\parens{p+sh,p+\tilde s h}
- H\parens{p+sh,p}
- H\parens{p,p+\tilde s h}
+ H\parens{p,p}} / h.
 \>
\qed
\end{thm}
The variance formula in \cref{thm:ebreve} is intimidating, but in many cases it simplifies substantially.  First, as noted following \cref{ass:Gaussian process}, if $G_{cT}$ is taken to be the empirical distribution function of the maximum rival bid then \cref{eq:Hstar simple} holds.%
\begin{lem} \label{lem:V}
	If \cref{eq:Hstar simple} holds  then 
    \<
    H\parens{p,p^*} = \zeta\parens{p}\zeta\parens{p^*} \cparens{ \min\parens{p,p^*}-pp^*},
    \>
     and
    $\sV\parens{p}$ in \cref{thm:ehat} simplifies to
	$
	\sV\parens{p} = \zeta^2\parens{p} \kappa_2, 
	$ 
	where $\kappa_2=\int_{-\infty}^\infty k^2\parens{s} \dif s$.\qed
\end{lem}%
\noindent
Since $k$ is chosen, $\sV$ is easy to estimate.

A second simplification obtains under full symmetry, i.e.\ when \cref{eq:Hstar symmetric} holds.  
\begin{lem} \label{lem:V symmetric}
    If \cref{eq:Hstar symmetric} holds then
    \< \label{eq:H symmetric simplification}
    H\parens{p,p^*} = \frac1n\cparens[\big]{
\min\parens{p,p^*}^{1/\parens{n-1}} - \parens{pp^*}^{1/\parens{n-1}}    
}
pp^* Q'\parens{p^{1/\parens{n-1}}} Q'\parens{p^*\,^{1/\parens{n-1}}}.
    \>
    Further,
    \[
    \lim_{h\to 0}\sH_h\parens{p,s,\tilde s}=
    \frac{p^{n/\parens{n-1}} Q'^2\parens{p^{1/\parens{n-1}}} \abs{\Med\parens{s,\tilde s,0}|}}{n\parens{n-1}},
   \]
    and $\sV$ simplifies to $ p^{n/\parens{n-1}} Q'^2\parens{p^{1/\parens{n-1}}} \kappa_2 / \cparens{n\parens{n-1}}$. \qed
\end{lem}
Finally, we provide a result for the asymmetric IPV case with $n$ bidders.
\begin{lem} \label{lem:V asymmetric}
If \cref{eq:Hstar asymmetric} holds then
\< \label{eq:H asymmetric simplification}    
H\parens{p,p^*} 
=
\zeta\parens{p}\zeta\parens{p^*} p p^* \sum_{i=2}^n \parens[\bigg]{ \frac{1}{G_i\sparens{Q_c\cparens{\max\parens{p,p^*}}}} -1}.
\>    
Further,
\< \label{eq:Hh asymmetric}
 \lim_{h\downarrow 0} \sH_h\parens{p,s,\tilde s} =
 \zeta^2\parens{p} Q_c'\parens{p} \sum_{i=2}^n G_{-i1}^2\cparens{Q_c\parens{p}}  g_i\cparens{Q_c\parens{p}} 
\abs{\Med\parens{s,\tilde s,0}},
\>
where $G_{-i1}$ means the distribution of the maximum bid of all bidders other than $i$ and $1$ with $G_{-i1}=1$ if there are only two bidders.  Finally, $\sV$ equals 
${\kappa}_2 {\zeta}^2\parens{p} Q_c'\parens{p} \sum_{i=2}^n G_{-i1}^2\cparens{Q_c\parens{p}}  g_i\cparens{Q_c\parens{p}}$. \qed
\end{lem}
Note that for $n=2$, the result in \cref{lem:V asymmetric} reduces to that in \cref{lem:V}.  For $n>2$, the function $H$ in \cref{lem:V asymmetric} is generally more favorable, i.e.\ it is more efficient to estimate each rival bid distribution separately than to estimate the distribution of the maximum rival bid using only the maximum rival bids.\footnote{For the case in which rival distributions happen to coincide but this fact is not used in the estimation, $\sV$ in \cref{lem:V asymmetric} reduces to ${\kappa}_2 {\zeta}^2\parens{p} p^{\parens{n-2}/\parens{n-1}}$ which equals $\sV$ in \cref{lem:V} if $n=2$ or $p\in \Set{0,1}$ but is otherwise less.
More generally, note that $\sV$ in \cref{lem:V asymmetric} is
${\kappa}_2{\zeta}^2 \sum_{i=2}^n G_{-i1}^2 g_i / \sum_{i=2}^n G_{-i1} g_i$ which is equal to ${\kappa}_2 {\zeta}^2$ and hence to $\sV$ in \cref{lem:V} if $G_{-i1}=1$, i.e.\ if $n=2$ or $p=1$.
}

A more interesting comparison is that of the formulas for $\sV$ in \cref{lem:V symmetric,lem:V asymmetric} if there is symmetry.  Indeed, the ratio of variances is $\parens{n-1}/n$ in favor of exploiting symmetry.  This result is intuitive since exploiting symmetry means that one can also use the bids of bidder one to estimate $G_c$: one then uses data on $n$ bids per auction instead of $n-1$.

One limitation of \cref{thm:ehat} compared to \cref{thm:ebreve} is that \cref{thm:ehat} does not extend to all of $\sparens{0,1}$.  A second issue is that the bias of $\ahT$ can be large for small values of $p$ as the following example illustrates.
\begin{ex} \label{ex:ehat}
Consider the symmetric case with $F_v$ a standard uniform and $n=3$.  Then $e\parens{p}=Q_c\parens{p} p = 2p^{3/2}/3$ and $e'''\parens{p}=-p^{-3/2}/4 \to -\infty$ as $p\downarrow 0$. \qed
\end{ex}
\drop{As will become apparent in \cref{ex:extensive comparison of {\alpha}} below, the} The GPV estimator also has the unbounded bias at zero problem.

Below, we address each of these limitations.

\subsection{Transformations}
\label{sec:psi}
  Let ${\psi}$ be an increasing function such that for $j=1,2,3$, $e^{\parens{j}}\parens{p} / {\psi}'^j\parens{p}$ and ${\psi}^{\parens{j}}\parens{p} /{\psi}'^j\parens{p}$ are uniformly bounded on $(0,1]$ and for which $\lim_{p \downarrow 0}$ of each of these functions is finite, also.   Then define\joris{Have to figure out what to do for $s<0$}
\begin{equation} \label{eq:ehatpsi}
 \hat {\alpha}_{T{\psi}}\parens{p} 
  = \frac{{\psi}'\parens{p}}h \lazyint{  {\alpha}_T\parens{s} k\parens[\Big]{\frac{{\psi}\parens{p}-{\psi}\parens{s}}h}}. 
\end{equation}
To see how \cref{eq:ehatpsi} solves the exploding bias near zero problem, consider the following.
The reason we needed $e$ to be three times boundedly differentiable in \cref{thm:ebreve} is that its proof contains a second order (bias) expansion of both $e\parens{p+sh}-e\parens{p}$ and ${\alpha}\parens{p+sh}-{\alpha}\parens{p}$: the former for $\beT$, the latter for $\ahT$.  If one uses $\ehatpsi$ then the corresponding expansions become
$e\sparens{{\psi}^{-1}\cparens{{\psi}\parens{p}+sh}} - e\parens{p}$ and 
${\psi}'\parens{p}\parens[\big]{{\alpha}\sparens{{\psi}^{-1}\cparens{{\psi}\parens{p}+sh}} - {\alpha}\parens{p}}$.  This makes all the difference since the second derivative of the first difference with respect to $s$ evaluated at $s=0$ is
\begin{equation} \label{eq:fugly1}
 \frac{{\alpha}'\parens{p}}{{\psi}'^2\parens{p}} - \frac{{\alpha}\parens{p}}{{\psi}'\parens{p}} \frac{{\psi}''\parens{p}}{{\psi}'^2\parens{p}}
\end{equation}
The corresponding expression for the second difference in  the preceding paragraph is
\begin{equation} \label{eq:fugly2}
 \frac{{\alpha}''\parens{p}}{{\psi}'^2\parens{p}}
 -
 3 \frac{{\alpha}'\parens{p}}{{\psi}'\parens{p}} \frac{{\psi}''\parens{p}}{{\psi}'^2\parens{p}}
 -
  {\alpha}\parens{p}\frac{{\psi}'''\parens{p}}{{\psi}'^3\parens{p}}
 +
 3 {\alpha}\parens{p} \frac{{\psi}''^2\parens{p}}{{\psi}'^4\parens{p}}.
\end{equation}
Note that the asymptotic bias in \cref{eq:fugly2} can be made to equal zero by choosing ${\psi}=e$.  Unfortunately, we do not know $e$, so making that choice is infeasible.

Consider \cref{ex:ehatpsi}.
\begin{ex}\label{ex:ehatpsi}
Recall \cref{ex:ehat}.	 If one uses ${\psi}\parens{p}=\log p$ then ${\psi}'\parens{p}=1/p$, ${\psi}''\parens{p}=-1/p^2$, and ${\psi}'''\parens{p}=2/p^3$.  This yields for instance $e'''\parens{p}/{\psi}'^2\parens{p} = - \sqrt{p}$, which is well--behaved near zero. One can verify that the other ratios in \cref{eq:fugly1,eq:fugly2} are equally well--behaved.

Note that our solution does not only work for $n=3$.  Indeed, in the symmetric case with arbitrary $n$, $e\parens{p}= \parens{1-1/n}p^{n / \parens{n-1}}$, such that $a''\parens{p}  / \psi'^2\parens{p} \sim p^{1/\parens{n-1}}$.  

The same goes for the situation in which there are $n-1$ stronger rivals.  It really does not matter since each derivative of $e$ removes a power of $p$ and each negative power of ${\psi}'$ restores one.\joris{Perhaps it may go wrong at $p=1$ if I'm stronger than all my rivals.}
\qed
\end{ex}

The bias formula in \cref{eq:fugly2} is somewhat complicated and a downside of the formula for $\hat {\alpha}_{T{\psi}}$ in \cref{eq:ehatpsi} is that, depending on the choices of $k,{\psi}$, it may require numerical integration.  This is an inconvenience more than a serious problem since ${\alpha}_T$ is piecewise constant.  However, both issues can be addressed by replacing $\hat {\alpha}_{T{\psi}}$ in \cref{eq:ehatpsi} with
\begin{equation} \label{eq:ebarpsi}
    \bar {\alpha}_{T{\psi}}\parens{p}  = \frac1h \lazyint{ {\psi}'\parens{s} {\alpha}_T\parens{s} k\parens[\Big]{\frac{{\psi}\parens{p}-{\psi}\parens{s}}h}},
\end{equation}
which produces the simpler form
\begin{equation} \label{eq:fugly3}
 \frac{{\alpha}''\parens{p}}{{\psi}'^2\parens{p}} - \frac{{\alpha}'\parens{p}{\psi}''\parens{p}}{{\psi}'^3\parens{p}}
 \end{equation}
 in lieu of \cref{eq:fugly2}. The asymptotic bias is zero if one chooses ${\psi}={\alpha}$, which is again infeasible.\footnote{Recall that choosing ${\psi}=e$ was infeasible for $\hat {\alpha}_{T{\psi}}$.}

\subsection{Boundary correction}
\label{section:boundary correction}
When computing $\hat e_T$ at values of $p$ near the boundary or using a kernel with infinite support, the locally weighted average of $\beT(s)$ attempts to put positive weight on values of $\beT$ for which $\beT$ is undefined. If one does not make adjustments to the kernel $k$ or the definition of $\beT$ outside of $\sparens{0,1}$ then ${\alpha}\parens{1}$ will not be consistently estimated, as the following example illustrates for $\bar {\alpha}_{T{\psi}}$.

\newcommand{\bad}{\mathrm{bad}}
\begin{ex} \label{ex:inconsistency at the boundary}
The estimator 
\(
 \bar {\alpha}_{T{\psi}}^\bad\parens{p} = h^{-1} \uint[s]{ {\psi}'\parens{s} {\alpha}_T\parens{s} k\parens[\big]{\parens{{\psi}\parens{p}-{\psi}\parens{s}}/h}}
\)    
is inconsistent at $p=1$.  To see this, note that
\begin{multline*}
\bar {\alpha}_{T{\psi}}^\bad\parens{1}= \frac1h \uint[s]{ {\psi}'\parens{s} {\alpha}_T\parens{s} k\parens[\Big]{\frac{{\psi}\parens{1}-{\psi}\parens{s}}h}}
=\\
\frac1h \uint[s]{ {\psi}'\parens{s} {\alpha}\parens{s} k\parens[\Big]{\frac{{\psi}\parens{1}-{\psi}\parens{s}}h}} + \opone
= {\alpha}\parens{1} \int_{-\infty}^0 k\parens{-s} \dif s+\opone = \frac{{\alpha}\parens{1}}{2} + \opone,
\end{multline*}
by consistency of ${\alpha}_T$ and substitution of $s \ot \cparens{ {\psi}\parens{s}-{\psi}\parens{p}}/h$.  This is the well--known boundary bias problem of nonparametric kernel density estimators. \qed
\end{ex}
If the lower end of the valuation's support is zero then ${\alpha}\parens{0}=0$ and the estimator $\bar {\alpha}_{T{\psi}}^\bad\parens{0}$ \emph{is} a consistent estimator of ${\alpha}\parens{0}=0$.  Note that the problem is true whether ${\psi}\parens{p}=p$ or not.

There are many solutions to this problem.  The traditional approach is to use a `boundary kernel,' i.e.\ a kernel that scales the kernel to make up for the lost mass if a function is estimated near a boundary.  We discuss this possibility in \cref{section:boundary kernels}.  A second possibility is to make use of techniques similar to those espoused in \citet[KZ]{karunamuni2008some} in order to ``make up'' values of $e$ and $\alpha$ beyond $\sparens{0,1}$.  This approach is investigated in \cref{section:kz}.\footnote{\citet{gimenes2019quantile} smooth the quantile function using a local polynomial approach.  The problem studied therein is otherwise unrelated.}

\subsubsection{Boundary kernels}
\label{section:boundary kernels}

The boundary bias problem can be addressed by the use of boundary kernels.  We replace \cref{eq:ehatpsi} with
\begin{equation} \label{eq:ehatpsi2}
 \hat {\alpha}_{T{\psi}}\parens{p} 
 =\frac{{\psi}'\parens{p}}h \uint[s]{
    {\alpha}_T\parens{s} k_{{\psi}h}\parens[\Big]{\frac{{\psi}\parens{p}-{\psi}\parens{s}}h \Bigm\bkdelim p}
}.
\end{equation}
where, for each $p\in[0,1]$, the function $k_{{\psi} h}(\cdot \bkdelim p)$ is a boundary kernel defined now.  Let $\bar {\upsilon}_{\psi} =\cparens{ {\psi}\parens{1}-{\psi}\parens{p}} / h $ and
$\bl {\upsilon}_{\psi}=\cparens{ {\psi}\parens{0}-{\psi}\parens{p}} / h $.  Then we require $k_{{\psi}h}$ to be such that for all $p\in \sparens{0,1}$, $j=0,1,2$, 
\begin{equation} \label{eq:boundary kernel requirements}
 \lim_{h\downarrow 0}\busyint{ s^j  k_{{\psi}h}\parens{-s \bkdelim p} } =|1-j|,  
\end{equation}
where the requirement for $j=2,p\in \Set{0,1}$ is replaced with boundedness.  
 The requirement that the kernel integrate to one is to ensure consistency in view of \cref{ex:inconsistency at the boundary}.  We also want it to integrate to zero if multiplied by $s$ to kill the `$h$ term' in a bias expansion.  

Boundary kernels are easy to construct as \cref{lem:boundary kernel} demonstrates. 
\begin{lem} \label{lem:boundary kernel} Let ${\phi},{\Phi}$ be the standard normal density and distribution functions. Then
$
 k_{{\psi}h}\parens{s \bkdelim p} =\parens{{\omega}_{{\psi}1} - {\omega}_{{\psi}2} s} {\phi}\parens{s}
$
satisfies the requirements in \cref{eq:boundary kernel requirements} for
\[
 {\omega}_{{\psi}2} = \frac{{\Omega}_{{\psi}1}}{{\Omega}_{{\psi}0}^2 + {\Omega}_{{\psi}0} {\Omega}_{{\psi}2}- {\Omega}_{{\psi}1}^2}, 
\qquad
{\omega}_{{\psi}1} = \frac{{\Omega}_{{\psi}0}+{\Omega}_{{\psi}2}}{{\Omega}_{{\psi}0}^2 + {\Omega}_{{\psi}0} {\Omega}_{{\psi}2}- {\Omega}_{{\psi}1}^2}, 
\]
where
${\Omega}_{{\psi}j} = {\Phi}^{\parens{j}}\parens{\bar {\upsilon}_{\psi}} - {\Phi}^{\parens{j}}\parens{\bl {\upsilon}_{\psi}}$. \qed
\end{lem}
The cut--out for $j=2$ and $p\in \Set{0,1}$ in the requirements for the boundary kernel is there because the requirements on $\int s^2 k_{{\psi}h}\parens{-s}$ only affect the `bias' in the asymptotic distribution and because it simplifies the formula for the boundary kernel.  A formula for a boundary kernel that does not require this exception is provided in \cref{lem:complicated boundary kernel} in \cref{app:proofs}. 

\begin{ass}\label{ass:psi}
The transformation ${\psi}$ is thrice continuously differentiable on $\parens{0,1}$ with ${\psi}'$ positive. \qed
\end{ass}

\begin{thm}\label{thm:ehatpsi}
Suppose that $k_{{\psi}h}$ is constructed as in \cref{lem:boundary kernel} and that \cref{ass:Gaussian process,ass:monotone bid strategies,ass:value distributions,ass:Qc differentiable,ass:psi} are satisfied.  Then
\[
\forall p \in \sparens{0,1}:
 \sqrt{Th} \cparens{ \hat {\alpha}_{T{\psi}}\parens{p}-{\alpha}\parens{p}} \convd
 N\cparens[\big]{\sB_{\psi}\parens{p} , \sV_{\psi}\parens{p}},
\]
where for $0<p<1$,
\[
\sB_{\psi}\parens{p}= \text{\emph{expression} \cref{eq:fugly2}} \times \frac{\Xi}{2}
\]
and
\[
\sV_{\psi}\parens{p}=   {\psi}'^2\parens{p}\lim_{h\to 0}\lazyint{ \lazyint[\tilde s]{ {\phi}'\parens{s}{\phi}'\parens{\tilde s} 
%
\sH_h\cparens[\big]{p,s/{\psi}'\parens{p},\tilde s/{\psi}'\parens{p}}
}}.
\]
where $\sH_h$ is as defined in \cref{eq:ugly H}.
If \cref{eq:Hstar simple} holds then we obtain the simpler expression
$
 \sV_{\psi}\parens{p}= {\kappa}_2{\zeta}^2\parens{p} {\psi}'\parens{p}. 
$
For $p \in \Set{0,1}$, $\sB_{\psi},\sV_{\psi}$ are finite. \qed
\end{thm}

We now turn to making $\bar {\alpha}_{T{\psi}}$ boundary--compliant, also.  We use
\begin{equation}\label{eq:alpha psi prime inside}
\bar {\alpha}_{T{\psi}}\parens{p}  =\frac1h \uint[s]{
    {\psi}'\parens{s} {\alpha}_T\parens{s} k_{{\psi}h}\parens[\Big]{\frac{{\psi}\parens{p}-{\psi}\parens{s}}h \Bigm\bkdelim p}},
\end{equation}
which produces the following theorem.
\begin{thm}\label{thm:ebarpsi}
    Suppose that $k_{{\psi}h}$ is constructed as in \cref{lem:boundary kernel}, that \cref{ass:Gaussian process,ass:monotone bid strategies,ass:value distributions,ass:Qc differentiable,ass:psi} are satisfied, and that ${\psi}'$ is bounded.  Then
    \[
    \forall p \in \sparens{0,1}:
    \sqrt{Th} \cparens{ \bar {\alpha}_{T{\psi}}\parens{p}-{\alpha}\parens{p}} \convd
    N\cparens[\big]{\bar\sB_{\psi}\parens{p} , \sV_{\psi}\parens{p}},
    \]
    where for $0<p<1$,
    \[
    \bar\sB_{\psi}\parens{p}= \text{\emph{expression} \cref{eq:fugly3}} \times \frac{\Xi}{2}
    \]
    and
    \[
    \sV_{\psi}\parens{p}= {\psi}'^2\parens{p}\lim_{h\to 0} \lazyint{ \lazyint[\tilde s]{  {\phi}'\parens{s}{\phi}'\parens{\tilde s} \sH_h\cparens[\big]{p,s/{\psi}'\parens{p},\tilde s/{\psi}'\parens{p}}
    }},
    \]
    where $\sH_h$ was defined in \cref{eq:ugly H}.
    For $p \in \Set{0,1}$, $\bar\sB_{\psi},\sV_{\psi}$ are finite. Under \cref{eq:Hstar simple}, $\sV_{\psi}$ simplifies to
     ${\kappa}_2 {\zeta}^2\parens{p}{\psi}'\parens{p}={\zeta}^2\parens{p}{\psi}'\parens{p}/\sqrt{{\pi}}$. \qed
\end{thm}

In both \cref{thm:ehatpsi,thm:ebarpsi} the kernel used was taken to be the kernel constructed in \cref{lem:boundary kernel}.  This is inessential.  Indeed, the results go through with ${\phi}$ replaced with a second--order kernel $k$ if $k_{{\psi}h}$ is chosen as $k_{{\psi}h}\parens{s\bkdelim p} = (\omega_{{\psi}k1} - {\omega}_{{\psi}k2}s)k(s)$ where ${\omega}_{{\psi}k1}$ and ${\omega}_{{\psi}k2}$ 
\[
{\omega}_{{\psi}k1} = \frac{{\Omega}_{{\psi}k2}}{{\Omega}_{{\psi}k0}{\Omega}_{{\psi}k2}- {\Omega}_{{\psi}k1}^2}, \quad
{\omega}_{{\psi}k2} = \frac{-{\omega}_{{\psi}k1}{\Omega}_{{\psi}k1}}{{\Omega}_{{\psi}k2}}\,,
\]
with ${\Omega}_{{\psi}kj} = \int_{\bl {\upsilon}_{\psi}}^{\bar {\upsilon}_{\psi}} u^j k(-u) \dif u$.

One advantage of $\bar {\alpha}_{T{\psi}}$ over $\hat {\alpha}_{T{\psi}}$ is computation, as the following lemma demonstrates.
\begin{lem} \label{lem:computation alpha bar}
	The formula for $\bar {\alpha}_{T{\psi}}$ simplifies to
	\[
	\bar {\alpha}_{T{\psi}}\parens{p}=\sum_{j=1}^T {\alpha}_{Tj}  {\Lambda}_{{\psi}j}\parens{p},
	\]
	where ${\Lambda}_{{\psi}j}\parens{p}=K_{{\psi}h}\cparens{ - {\upsilon}_{j-1}\parens{p} \bkdelim p}
	-
	K_{{\psi}h}\cparens{ - {\upsilon}_j\parens{p}\bkdelim p}$,
	with $K_{{\psi}h}= \int k_{{\psi}h}$ and ${\upsilon}_j\parens{p}= \cparens{ {\psi}\parens{j/T}- {\psi}\parens{p}} / h$.      For $k_{{\psi}h}$ as constructed in \cref{lem:boundary kernel},
	\(
	K_{{\psi}h}\parens{s} = {\omega}_{{\psi}1} {\Phi}\parens{s} + {\omega}_{{\psi}2} {\phi}\parens{s}. 
	\) \qed
\end{lem}

\subsubsection{Another boundary correction}
\label{section:kz}

A second way of implementing boundary corrections is to create artificial values of ${\alpha}_T\parens{p}$ for $p$ outside $\sparens{0,1}$.  Our approach is loosely motivated by the KZ method for kernel estimators, but it is a bit cleaner because of our specific circumstances: we are trying to smooth out an existing estimator which means that we already have values of ${\alpha}_T\parens{p}$ between zero and one.

Here, we restrict $k$ to be the Epanechnikov kernel $\tikz \draw[smooth,Red,domain=-1:1,scale=0.35] plot({\x},{(0.75-0.75*\x*\x)});$ which is a quadratic on $\sparens{-1,1}$; indeed it is $3\parens{1-x^2}/4$.\footnote{Earlier, we had taken $\int k\parens{s}s^2 \dif s$ to equal one, which is not true for an Epanechnikov kernel.  We adjust the asymptotic bias expression accordingly.} Consequently, any boundary correction procedure will be immaterial if the distance between ${\psi}\parens{p}$ and ${\psi}\parens{1},{\psi}\parens{0}$ exceeds $h$.  We focus on correcting estimates near the upper bound, $p=1$.  Impose the scale and location normalizations ${\psi}\parens{1}=0$ and ${\psi}'\parens{1}=1$.

Define 
\begin{equation} \label{eq:alpha extension}
 {\alpha}\parens{1+s} = {\alpha}\sparens[\big]{1 - {\rho}\cparens[\big]{{\psi}\parens{1+s}}} {\rho}'\cparens[\big]{{\psi}\parens{1+s}}, \quad s>0,
\end{equation}
where
\(
{\rho}\parens{s} = s + d s^2 + \cparens{d^2 - {\psi}''\parens{1} d/6}s^3, 
\)
with $d={\alpha}'\parens{1}/{\alpha}\parens{1}$.  Then it is straightforward but unpleasant to verify that the thus extended version of ${\alpha}$ is twice continuously differentiable at one.  We extend ${\alpha}_T$ analogously to \cref{eq:alpha extension} using a suitable estimator $\hat d$ in lieu of $d$, defining $\hat {\rho}$ to be like ${\rho}$ but with $\hat d$ replacing $d$.  We can then obtain a smoothed estimate of ${\alpha}$ by defining
\begin{equation}\label{eq:alpha R}
 \ahTR\parens{p} = \frac1h \int_{-\infty}^\infty {\alpha}_T\parens{s} {\psi}'\parens{s} k\parens[\Big]{\frac{{\psi}\parens{p}-{\psi}\parens{s}}h} \dif s,
\end{equation}
where the superscript $R$ stands for `reflection.'
As noted, away from the boundary, the behavior of $\ahTR$ is no different than that of the estimator without boundary bias correction.  So we only analyze its behavior in an $h$--neighborhood of the boundary, as formulated in \cref{thm:boundary correction reflection}.

\begin{thm} \label{thm:boundary correction reflection}
	Let 
\begin{enumerate*}[label=\itshape{(\roman*)}]
	\item \cref{ass:Gaussian process,ass:monotone bid strategies,ass:value distributions,ass:Qc differentiable} be satisfied;
	\item $k$ be the Epanechnikov kernel;
	\item ${\psi}$ be twice continuously differentiable at 1 with ${\psi}\parens{1}=0$ and ${\psi}'\parens{1}=1$;
	\item $\hat d$ converge to $d$ at a rate no slower than $\sqrt[5]{T}$.
\end{enumerate*}
Then,
\< \label{eq:boundary correction statement}
\sqrt{Th} \cparens{ \ahTR\parens{1-th}-{\alpha}\parens{1-th}} 
+
\frac{\sqrt{Th^3}}8 {\alpha}\parens{1} \parens{1-t}^3 \parens{t+3}\parens{\hat d - d}
\convd
N\parens[\big]{ \sB^R\parens{t}, \sV^R\parens{t}},
\>
where
\(
\sB^R\parens{t} = \cparens{{\alpha}''\parens{1} - {\alpha}'\parens{1}{\psi}''\parens{1}} {\Xi}/10
\)
and
\[
\sV^R\parens{t} = \lim_{h\to 0} \int_0^{1+t} \int_0^{1+t} 
\diam k_t\parens{s} \diam k_t\parens{\tilde s} 
\sH_h\parens{1,-s,-\tilde s}
\dif s \dif \tilde s,
\]
where $\diam k_t\parens{s} = \parens{3/2}\cparens{\parens{t-s} \one\parens{1-t \leq s\leq 1+t}-2s \one\parens{0\leq s\leq 1-t}}$ and $\sH_h$ was defined in \cref{eq:ugly H}.
If \cref{eq:Hstar simple} holds then we obtain the simpler expression
\(
\sV^R\parens{t} = 3 {\zeta}^2\parens{1}
\cparens[\big]{ 2-t^2 \parens{t^3-5t+5}}/5. 
\)\joris{I like the fact that the variance at 1 is twice that at $1-h$ but I'd expected it to be four times since it uses half the data.}\qed
\end{thm}
As noted, the conditions on ${\psi}$ are normalizations: without them ${\psi}\parens{1},{\psi}'\parens{1}$ would pop up in various places.  Our assumption of the Epanechnikov kernel is not essential but the proofs do make use of the fact that the kernel has bounded support.  Moreover, the polynomial portion of the second term in \cref{eq:boundary correction statement} would be more complicated. 

Since $\hat d$ is essentially a nonparametric kernel derivative estimator, achieving a $T^{1/5}$ rate is feasible under \cref{ass:Qc differentiable}.\footnote{If the function whose derivative is estimated is twice differentiable then it is well--known that the bias is $O\parens{h_d}$ and the variance $O\parens{1/Th_d^3}$, where $h_d$ is the bandwidth used for the estimation of $d$.  Here, ${\alpha}$ is the function whose derivative is to be estimated, which is twice differentiable under \cref{ass:Qc differentiable} since ${\alpha}''=Q_c'''p + 3 Q_c''$.}  If one assumes $Q_c$ to have one more derivative at $1$ then ${\alpha}$ is thrice differentiable at 1, which would imply that picking a bandwidth $h_d$ for $\hat d$ that converges faster than $T^{-1/10}$ and slower than $T^{-1/5}$ would make the second term in \cref{eq:boundary correction statement} disappear: $h_d \sim T^{-1/7}$ would be optimal.  

So, here we advocate picking a bandwidth for $\hat d$ which tends to zero more slowly than $T^{-1/5}$ whereas KZ advocates making the bandwidth go to zero faster than $T^{-1/5}$.  In a separate note \citep{pinkse2019actual} we show that there is a bug in both \citet{karunamuni2005boundary} and KZ and that there one needs to assume the existence of an extra derivative and choose a bandwidth that converges more slowly in order to obtain their claimed results.

Near the left boundary, we apply an analogous reflection method based upon
\[
{\alpha}\parens{s} = {\alpha}\sparens[\bigg]{{\rho}_0\parens[\bigg]{\frac{{\psi}\parens{0}-{\psi}\parens{s}}{{\psi}'\parens{0}}}}{\rho}_0'\parens[\bigg]{\frac{{\psi}\parens{0}-{\psi}\parens{s}}{{\psi}'\parens{0}}}\,,
\]
where ${\rho}_0(s) = s - d_0 s^2 + \cparens[\big]{d_0^2 - d_0 {\psi}''\parens{0}/\sparens{6{\psi}'\parens{0}}}s^3$ and $d_0 = {\alpha}'\parens{0}/{\alpha}\parens{0}$. The formula is messier simply because we had already normalized the location and scale of ${\psi}$ at $p=1$ to simplify the expressions near the right boundary.

\subsubsection{Preserving monotonicity}
One caveat to our boundary kernel estimators and `reflection' procedure is that they can undo monotonicity near the boundaries in finite samples, although for different reasons. The boundary kernels are nonpositive near the boundary and are therefore capable of producing nonmonotonic estimates when ${\alpha}_T$ is relatively flat near the boundary. On the other hand, the transformation--and--reflection procedure in \cref{eq:alpha R} continuously extends ${\alpha}$ and its first two derivatives such that $\alpha(1+s)$ is generally decreasing in $s$ for large enough $s>0$. Indeed, this is inevitable when ${\alpha}'$ is close to zero and ${\alpha}''$ is negative. In any case, we may easily remedy this by redefining the smoothed estimator for ${\alpha}$ as the ``cumulative maximum'' of the objects defined in \cref{eq:ehatpsi2}, \cref{eq:alpha psi prime inside}, and \cref{eq:alpha R}, for example $\bar \alpha_{T\psi}(p) = \max\cparens[\big]{\sparens{{\psi}'\parens{p}/h} \uint[s]{{\alpha}_T\parens{s} k_{{\psi}h}\cparens{\parens{{\psi}\parens{p}-{\psi}\parens{s}}/h \bkdelim p}},\ \sup_{q<p}\hat {\alpha}_{T{\psi}}(q)}$.

Alternatively, in the case of the transformation-and-reflection procedure, we may apply this monotonization device to the definition of the extended ${\alpha}_{T}$. The kernel--smoothed estimator of the resulting monotonic function will then be increasing on $[0,1]$ because $k$ is a nonnegative kernel. Such a procedure will continuously extend ${\alpha}'$ and ${\alpha}''$ at one, but may introduce a discontinuity in ${\alpha}''$ at a point $p>1$ for which ${\alpha}'(p)=0$. We tolerate this discontinuity, however, because $d>0$ and a finite ${\alpha}''$ imply that the discontinuity is at a location bounded away from one. As a result, \cref{thm:boundary correction reflection} does not require any modification.

\subsection{Derivative estimators}
\label{section:derivatives}

As we will see in \cref{sec:objects}, the density of the value distribution depends on ${\alpha}'$, not on ${\alpha}$ itself.  Although the primary objective in our paper concerns estimation of derived objects like the bidder surplus, we include results for the value density in the interest of completeness.  For that purpose, we derive some results for an estimator of ${\alpha}'$, both away from and near the boundary.

The first result, \cref{thm:derivatives away from the boundary}, is for the case in which we are trying to estimate ${\alpha}'$ away from the boundary, whereas \cref{thm:derivatives near the boundary} applies to a neighborhood of the (upper) boundary.  

\begin{thm} \label{thm:derivatives away from the boundary}
	Let 
\begin{enumerate*}[label=\itshape{(\roman*)}]
    \item \cref{ass:Gaussian process,ass:monotone bid strategies,ass:value distributions} be satisfied;
    \item $Q_c$ be four times continuously differentiable on any compact subset of $\parens{0,1}$;
    \item $k$ be the Epanechnikov kernel;
    \item $\lim_{T\to\infty} \sqrt{Th^7} ={\Xi}_d<\infty$.
\end{enumerate*}	
Then, at any fixed $0<p<1$,
\[
 \sqrt{Th^3} \cparens{ \bar {\alpha}_{T{\psi}}'\parens{p}- {\alpha}'\parens{p}}
  \convd
 N\parens[\big]{ \sB^R\parens{p}, \sV^R\parens{p}},
\]
with
\[
\sB^R\parens{p} = {\Xi}_d
\frac{
{\alpha}'''{\psi}'^2 - 3 {\alpha}''{\psi}''{\psi}' - {\alpha}'{\psi}'''{\psi}' + 3 {\alpha}'{\psi}''^2
}{10{\psi}'^4},
\]
where all ${\alpha},{\psi}$'s are evaluated at $p$, and
\[
 \sV^R\parens{p}= \frac94 {\psi}'^4\parens{p} \lim_{h\downarrow 0}\int_{-1}^1 \int_{-1}^1 \sH_h\cparens[\big]{p,s/{\psi}'\parens{p},\tilde s/{\psi}'\parens{p}} \dif s \dif \tilde s.
\]	
If \cref{eq:Hstar simple} holds then $\sV^R\parens{p}$ simplifies to
\(
\parens{3/2} {\psi}'^3\parens{p} {\zeta}^2\parens{p}.  
\)
If \cref{eq:Hstar symmetric} holds then $\sV^R\parens{p}$ simplifies to
\< \label{eq:variance of derivative estimator in symmetric case}
 \frac{3 {\psi}'^3\parens{p}}{2n\parens{n-1}} p^{n/\parens{n-1}} Q'^2\parens{p^{1/\parens{n-1}}}. 
\>
Simplifying expressions for $F_p,{\alpha}$ and their first three derivatives in the symmetric case can be found in \cref{lem:alpha symmetric} in \cref{app:derivatives}. \qed
\end{thm}

Observe that the optimal convergence rate is the same as that for nonparametric kernel derivative estimators, namely $T^{2/7}$ for $h \sim T^{-1/7}$, as expected.  Note further that, like before, the scale of ${\psi}$ and the bandwidth $h$ are interchangeable.  Again, the optimal yet infeasible choice of ${\psi}$ in terms of the asymptotic bias is ${\psi} \propto {\alpha}$.

\begin{thm}\label{thm:derivatives near the boundary}
    Let
\begin{enumerate*}[label=\itshape{(\roman*)}]
    \item \cref{ass:Gaussian process,ass:monotone bid strategies,ass:value distributions} be satisfied;
    \item $Q_c$ be four times continuously differentiable on any compact subset of $\parens{0,1}$;
    \item $k$ be the Epanechnikov kernel;
    \item ${\psi}$ be thrice continuously differentiable at 1 with ${\psi}\parens{1}=0$ and ${\psi}'\parens{1}=1$;
    \item $\hat d - d = O_p\parens{T^{-2/5}}$;
    \item $\lim_{T\to\infty} \sqrt{Th^7} ={\Xi}_d<\infty$.
\end{enumerate*}	
Then for any $0\leq t\leq 1$,
\[
\sqrt{Th^3} \cparens{ \bar {\alpha}_{T{\psi}}^R\,\!'\parens{1-th}- {\alpha}'\parens{1-th}}
- 
\sqrt{\frac Th} \frac{{\alpha}\parens{1}}2 \parens{1-t}^3 \parens{\hat d-d}
\convd
N\parens[\big]{ \sB^{Rd}\parens{t}, \sV^{Rd}\parens{t}},
\]
where
\begin{multline*}
\sB^{Rd}\parens{t}
=
\frac{{\Xi}_d}{80}\parens[\Big]{
	8 \cparens[\big]{{\alpha}'''_\uparrow\parens{1}+3{\alpha}'\parens{1}{\psi}''^2\parens{1} - {\alpha}'\parens{1}{\psi}'''\parens{1} - 3{\alpha}''\parens{1}{\psi}''\parens{1}  }
	+
	\\
	\parens{4+t}\parens{1-t}^4 \cparens{{\alpha}'''_\uparrow\parens{1}-{\alpha}'''_\downarrow\parens{1}}}
\end{multline*}
with $	{\alpha}_\uparrow'''$, ${\alpha}_\downarrow'''$ denoting left and right derivatives,
and
\[
\sV^{Rd}\parens{t} =
\frac94\lim_{h\downarrow 0}
\int_{1-t}^{1+t} \int_{1-t}^{1+t}
\sH_h\parens{1,-s,-\tilde s}
 \dif \tilde s \dif s.
\]
If \cref{eq:Hstar simple} holds then the asymptotic variance simplifies to
\[\tag*{\qed}
\sV^{Rd}\parens{t} =
3 {\zeta}^2\parens{1} 
t^2\parens{3-t}. 
\]
\end{thm}

Although the asymptotic variances are formulated differently, the asymptotic distributions in the two theorems coincide if one takes $t=1$ in \cref{thm:derivatives near the boundary}.  Indeed, if $t=1$ then the correction via $\hat d$ becomes immaterial since there is no boundary bias concern then.  Note that if $h \sim T^{-1/7}$ then the convergence rate is still $T^{2/7}$ irrespective of the value of $t$. 

A perhaps puzzling finding is that the asymptotic variance is zero if $t=0$.  However, note that this is not the asymptotic variance of $\hat {\alpha}_T^{R}\,\!'\parens{1}$ itself.  Indeed, the (variation in the) asymptotic distribution of $\hat {\alpha}_T^{R}\,\!'\parens{1}$ is then determined by the estimation of $d$.  To get the asymptotic distribution of $\hat {\alpha}_T^{R}\,\!'\parens{1}$ itself requires us to commit to a specific estimator of $\hat d$ and derive the joint distribution.  This is neither difficult nor interesting.

\subsection{Jackknife estimators}

\Cref{thm:ebreve,lem:V} motivate still more estimators of ${\alpha}$.  Note that ${\alpha}\parens{p}= Q_c'\parens{p} p + Q_c\parens{p} = \zeta\parens{p}+Q_c\parens{p}$.  Since $Q_c$ can be estimated at a rate of $\sqrt{T}$, its estimation is of secondary concern.  But $\zeta\parens{p}$ enters the variance formulas in \cref{thm:ebreve,lem:V}.  

We will assume for the purpose of this discussion that $G_c$ is estimated using the empirical distribution function of the maximum rival bid, such that $H\parens{p,p^*}= \zeta\parens{p}\zeta\parens{p^*}\cparens{\min\parens{p,p^*}-pp^*}$ and the conditions of \cref{lem:V} are satisfied.

We present two versions, one based on \cref{thm:ebreve} and one on \cref{lem:V}:\joris{Don't know whether we need to undersmooth for the second formula.}
\[
\maligned{
\ahJ\parens{p} & = \sqrt{\frac{ \parens{T-1} \sum_{t=1}^T
	\cparens{\beT\parens{1} - \beT\parens{p}
		-\beTmt\parens{1} + \beTmt\parens{p}}^2}{p \parens{1-p}}}
	+ \hat Q_c\parens{p},
	\\
\ahJs\parens{p} & =
\sqrt{\frac{h\parens{T-1} \sum_{t=1}^T 
		\cparens{ \ahT\parens{p}-\ahTmt\parens{p}}^2}{\kappa_2}}
	+ \hat Q_c\parens{p},
}
\]
where the $-t$ subscripts denote leave--one--out estimators, i.e.\ the identical estimator without using observation $t$.  Note that $\ahJ$ is only defined on $0<p<1$ albeit that it can be defined to equal zero at zero and one.  This is precisely the reason for having
$\beT\parens{1}-\beTmt\parens{1}$ in the numerator even though it could be left out without affecting the result for fixed $0<p<1$.\footnote{$\sqrt{T}\cparens{\beT\parens{1}-e\parens{1}}=\opone$.}

We inserted a generic estimator $\hat Q_c$ into the definitions of $\ahJ,\ahJs$.  Its form is largely immaterial, but natural choices would be respectively $\beT\parens{p}/p$ and $\ehat\parens{p}/p$  for $p>0$ and zero for $p=0$.

There are three downsides to the use of these jackknife estimators.  The first issue is that in their current incarnation it is assumed that $H^*$ has a specific form.  But the formulas can be generalized or derived for other forms of $H^*$.
Second, the jackknife estimators are costlier to compute since each estimator has to be computed $T+1$ times.  This may be of little practical relevance since computation of $\ahT$ is fast.
Finally, the jackknife estimators are not guaranteed to be monotonic.  This is a property they share with other estimators, including GPV, and which can be addressed by the use of a monotonization procedure, which is not difficult but admittedly cumbersome.\footnote{See \cite{ma2019monotonicity} for a monotonization procedure of the GPV estimator.}  We do not study the asymptotic properties of jackknife estimators in this paper.


\section{Estimation of $F_{p}$}
\label{sec:Fp}
We now turn to the much simpler problem of estimating the distribution of a bidder's equilibrium win--probabilities.

\subsection{Symmetric Bidders}
In a symmetric equilibrium with $n$ bidders, the probability that a bidder with a valuation of $v$ wins is simply given by the probability that all other bidders have a valuation less than $v$. Accordingly, the distribution of a bidder's optimally chosen win--probabilities is 
\[
F_{p}(p) = p^{1/\parens{n-1}}\,,
\]
No estimation is necessary if $n$ is known because the distribution of equilibrium win--probabilities does not depend on the unknown distribution $F_{v}$.

We can accommodate endogenous entry as long as the screening value, i.e.~the lowest valuation $v^*$ for which a bidder is willing to participate, is observed. We would simply define $F_{p}(p)$ to equal zero for all $p < \alpha^{-1}(v^{*})$. For example, in a first--price auction with a reserve price $r>\mathb{v}$, $v^{*}=r$ and
\[
F_{p}(p) = \maligned{
p^{1/\parens{n-1}},\quad p\geq r^{n-1}\\
0, \quad p < r^{n-1}
}\,.
\]
We will use this fact when we discuss estimation of counterfactual expected revenues for the seller in \cref{sec:objects}.

\subsection{Asymmetric bidders}
In a high--bid auction\footnote{A high--bid auction is one in which the highest bidder wins with probability one.} with bidders whose valuation distributions are not identically distributed, the equilibrium distribution of win--probabilities for bidder is $F_{p}(p) = G\cparens{Q_{c}(p)}$. The distribution $F_{p}$ can then be estimated in a straightforward fashion as $F_{pT}(p) = G_{T}\cparens{Q_{cT}(p)}$, where $G_{T}$ and $Q_{cT}$ are the empirical distribution of bidder 1's bid and an estimate of the quantile function of its highest competing bid. The weak convergence of this process on $(0,1)$ is closely related to the extensively studied ``copula process'' and the fact that the marginal bid densities are strictly positive on their compact support.

\begin{thm}\label{thm:Fp}
	\(
	\sqrt{T}\parens{F_{pT}-F_p} \convw \G_p\,,
	\)
	where 
	$\G_p$ is a Gaussian process with covariance kernel\\
\( 
F_p\cparens{\min\parens{p,p^*}} - F_p\parens{p} F_p\parens{p^*}
+
f_p\parens{p} f_p\parens{p^*} H^*\cparens{Q_c\parens{p},Q_c\parens{p^*}}.
\) \qed
\end{thm}
Recall that $  H^*\cparens{Q_c\parens{p},Q_c\parens{p^*}}$ can be as simple as $\min\parens{p,p^*}-pp^*$ in case only the maximum competitor bid is used: see \cref{eq:Hstar simple}.

\subsection{Minimum relative entropy}
In a particular application, the true marginal distributions of valuations might not differ substantially, even when the econometrician is unwilling to impose bidder symmetry in the estimation. 
Thus, a minimum relative entropy estimator for the distribution of win--probabilities may be an attractive alternative. Define $\hat f_{pT}$ as the minimizer of
\[
\int_{0}^{1} f_{p}\parens{p} \log\parens[\Big]{ f_{p}\parens{p} p^{\frac{n-2}{n-1}}}\dif p\,
\qquad
\text{subject to} \qquad
\int_{0}^{1} \delta_{T}\parens{p} f_{p}\parens{p} \dif p = \int_0^1 \delta_{T}\parens{p}\dif F_{pT}\parens{p},
\]
where $\Set{\delta_{T}}$ is a user--specified sequence of functions.\footnote{Elsewhere, we use $k$ to denote kernel and $h$ to denote bandwidth.  In view of the similar meaning and limited scope for confusion we duplicate notation here to make better use of other symbols.}  A natural choice would be $\delta_T(p) = \sparens{1,p,p^{2},\dots, p^{\iota_{T}}}'$ for some growing sequence of natural numbers $\iota_{T}$. The solution to this problem is
\(
f_{p}\parens{p} = \exp\cparens{\mu'\delta_{T}\parens{p}}p^{\parens{2-n}/\parens{n-1}}\,,
\)
where $\mu$ solves
\(
\int_{0}^{1} \delta_{T}\parens{p} \exp\cparens{\mu'\iota_{T}\parens{p}}p^{\parens{2-n}/\parens{n-1}}\dif p =  \int_0^1 \delta_{T}\parens{p}\dif F_{pT}\parens{p}\,.
\)
Given our choice of $\delta_{T}$, the estimate $\hat{f}_{pT}$ is the nearest density (in the sense of Kullback--Leibler divergence) to the symmetric case that matches the first $\iota_{T}$ sample moments of $p$. Since this yields something similar to a sieve estimator, we do not provide asymptotic results here and refer to \citet{chen2007large} for details of such estimators.

\section{Derived objects}
\label{sec:objects}

Applied researchers typically are not directly interested in the private values that rationalize a particular sample of bids, but may estimate these so--called pseudo values in order to construct other estimates. For instance, the sample of pseudo values may be used to obtain estimates of the private value distribution.  

The same comment applies to the density of the private value distribution: because the marginal value distributions are the primitives of the model, i.e.~any counterfactual outcomes or other objects of interest may be computed using the private value distribution, estimating the (density of the) pseudo values at an optimal rate is considered a goal itself in a good chunk of the literature. This intermediate step may be unnecessary or undesirable when the ultimate target of estimation can be written in terms of higher level objects or when the distribution of equilibrium win--probabilities is known.  Below are some examples.

In each case, the object of interest may be expressed as $\theta\parens{{\alpha}, F_p}$, where ${\theta}$ is a known function, and we estimate the object by plugging in some combination of estimates of ${\alpha}$ and $F_p$. There are two overarching themes in the following discussion. First, the asymptotic derivations are greatly simplified by the fact that $F_p$ is known in any symmetric equilibrium, and we can expect significant improvements in finite--sample (and often also asymptotic) performance when we plug in the true $F_p$ as opposed to an estimated distribution and pool bids across bidders to more accurately estimate the rival bid distributions. Second, we may expect the plug--in estimator for ${\theta}$ to be $\sqrt{T}$--consistent and asymptotically unbiased when ${\theta}$ takes the form ${\theta}\parens{{\alpha}, F_p} = \int {\theta}_1(\alpha) \dif F_p$ for an appropriately differentiable function ${\theta}_1$, as is often the case when integrating over nonparametrically estimated objects.\footnote{By `asymptotically unbiased' we mean that the limit distribution has mean zero.}

We now turn to a discussion of individual objects to be estimated.  Although not the primary objective in our exercise, we briefly discuss how to estimate the value distribution function, quantiles, and density function in \cref{section:value distribution}.  We then turn to some objects of greater interest to us, namely the bidder surplus, the mean of the value distribution, profit as a function of the number of bidders, and profit as a function of a hypothetical reserve price.

\subsection{Value distribution}
\label{section:value distribution}

There are different attributes of the value distribution that can be estimated.  The easiest object to recover is the quantile function.  Note that
since $v = {\alpha}\parens{p} = {\alpha}\cparens{G_c\parens{b}}$,
\[
Q_v\parens{{\tau}}= {\alpha}\cparens{Q_p\parens{{\tau}}}={\alpha}\sparens{G_c\cparens{Q_b\parens{{\tau}}}}, \qquad {\tau} \in \sparens{0,1},
\]
which simplifies to ${\alpha}\parens{{\tau}^{n-1}}$ in the symmetric case.\footnote{In the symmetric case, a direct estimator of the quantile function like the one proposed in \cite{gimenes2019quantile} may be preferable.}   The functions $G_c,Q_b$ can be estimated $\sqrt{T}$--con\-sis\-tent\-ly, but not so for ${\alpha}$ as our results thus far have shown.  So even though we are estimating quantiles, namely quantiles of the value distribution, these quantiles cannot be estimated at the parametric rate because the values are not observed.  Indeed, the limit distribution of an estimator $\hat Q_v$ of $Q_v$ is simply the limit distribution of whatever estimator of ${\alpha}$ is used evaluated at $p=G_c\cparens{Q_b\parens{{\tau}}}$.
Likewise, the value distribution function is simply
\[
F_v\parens{v} = F_p\cparens{{\alpha}^{-1}\parens{v}}. 
\]
With symmetric bidders, $F_p= p^{1/\parens{n-1}}$.  In the case of asymmetry,
$F_p$ can be estimated $\sqrt{T}$--consistently, such that the limit distribution is by the delta method given by $\parens{f_p/{\alpha}'}\cparens{{\alpha}^{-1}\parens{v}}=f_v\parens{v}$ times the limit distribution of the estimator of ${\alpha}$.  Note that the delta method is only valid for $v \neq 0,\bar v$, which is of little consequence since we already know the values of $F_v\parens{0},F_v\parens{1}$, albeit that uniformity arguments would suggest that the implied asymptotic distribution would not reflect the finite sample performance near 0 and 1, either, although the convergence rate is still $T^{2/5}$ for the same reason that $\beT$ converges at the $\sqrt{T}$ rate on the entire interval $\sparens{0,1}$: see the comments in the paragraph following  \cref{thm:ebreve}.

There are two ways to estimate the value density: one--step and two--step.  With the two--step estimator, one first generates valuation estimates by doing e.g.\ $\hat v_t = \bar {\alpha}_{T{\psi}}\cparens{\hat G_{cT}\parens{b_{t1}}}$ and then plugs those estimates into a nonparametric kernel density estimator.  The two--step estimator is analogous to GPV except that our first step is different.  It can be shown\footnote{Derivation not provided here.} that both the first step in GPV and our smoothed estimates of ${\alpha}$ permit asymptotic linear expansions of estimator minus expectation at $b=Q_c\parens{p}$,
\< \label{eq:linear expansions}
\maligned{
-\cparens[\bigg]{\frac1{Th} \sum_{t=1}^T \frac{G_c\parens{b}}{g_c^2\parens{b}} k\parens[\Big]{\frac{b_{ct}-b}h} -
\text{its expectation}} & \qquad &\text{ (GPV)},
\\
-\cparens[\bigg]{\frac1{Th} \sum_{t=1}^T \frac{G_c\parens{b}}{g_c\parens{b}} k\parens[\Big]{\frac{G_c\parens{b_{ct}}-p}h} -
	\text{its expectation}} & \qquad &\text{ (ours)}, 
\\
-\cparens[\bigg]{\frac1{Th} \sum_{t=1}^T {\psi}'\cparens{G_c\parens{b}}\frac{G_c\parens{b}}{g_c\parens{b}} k\parens[\Big]{\frac{{\psi}\cparens{G_c\parens{b_{ct}}}-{\psi}\parens{p}}h} -
	\text{its expectation}} & \qquad &\text{ (ours with ${\psi}$)}, 
}
\>
The first two formulas in \cref{eq:linear expansions} are similar, but note the different arguments to the kernel and the fact that one denominator has a square on $g_c$ and the other one does not.  The formula with ${\psi}$ simplifies to the one without for ${\psi}\parens{p}=p$ and to the GPV expansion for ${\psi}\parens{p}=Q_c\parens{p}$.  However, the bias of our estimator with ${\psi}=Q_c$ does not coincide with that for the first step GPV bias: either can be greater.

We only provide asymptotics for the one--step estimator.  For the one--step estimator, note that the value density function is
\[
f_v\parens{v} = \parens{f_p/{\alpha}'}\cparens{{\alpha}^{-1}\parens{v}},
\]
and hence requires an estimate of ${\alpha}'$, which we provided in \cref{section:derivatives}   In the symmetric case, $f_p\parens{p}=p^{\parens{2-n}/\parens{n-1}}/\parens{n-1}$ and one would need to use an estimate of $G_c$ (and hence $Q_c$) that fully exploits symmetry.   With asymmetric bidders it also requires an estimate of $f_p$, but density estimates converge faster than do their derivatives so the estimate of ${\alpha}'$ determines the asymptotic distribution of $\hat f_v\parens{v}$, which is $- \parens{f_p/{\alpha}'^2}\cparens{{\alpha}^{-1}\parens{v}}$ times the limit distribution of one's estimate of ${\alpha}'$, again by the delta method.  From \cref{eq:variance of derivative estimator in symmetric case} it follows that the bias and variance of our estimator of the value density in the symmetric case are given by
\[
\sB_{f}^\symm\parens{p}= - \frac{f_p\parens{p}}{{\alpha}'^2\parens{p}}\sB^R\parens{p},
\]
and
\begin{multline*}
\sV_f^\symm\parens{p} = \frac{3\parens{n-1}^5 {\psi}'^3\parens{p}}{2n^5} \frac{p^{\frac{3n-4}{n-1}} Q'^2\parens{p^\frac{1}{n-1}}}
{\cparens[\big]{ Q'\parens{p^{\frac{1}{n-1}}} + p^{1/\parens{n-1}}Q''\parens{p^{\frac{1}{n-1}}}/n }^4}
\\
=
\frac{3\parens{n-1}^5 {\psi}'^3\cparens{G^{n-1}\parens{b}}}{2n^5}
\frac{ G^{3n-4}\parens{b} g^{10}\parens{b}}
{
\cparens{  g^2\parens{b} -G\parens{b}  g'\parens{b}/n }^4
}.
\end{multline*}
For ${\psi}\parens{p}=Q_c\parens{p}$ (or indeed a suitable estimate thereof) our variance coincides with that of \citet[MS]{marmer2012quantile}, theorem 2.  \citet{ma2019inference} note that the variance of the MS estimator exceeds that of GPV for the same choice of kernel and bandwidth if one undersmooths, i.e.\ if one chooses a bandwidth which makes the bias disappear faster than the variance.  We recommend against undersmoothing for the purpose of estimating ${\alpha}$ and note that ${\psi}=Q_c$ is not optimal.\footnote{For the purpose of inference undersmoothing makes sense but for estimation it is better to choose a bandwidth that converges at the optimal rate since it results in a better convergence rate of the estimator than if one undersmooths. 
	 Second, \cref{eq:linear expansions} suggests that the observation in \citet{ma2019inference} is due to the use of a one--step instead of a two--step estimator.  Finally, \citet{ma2019inference} do not employ transformations like ${\psi}$, which can yield a smaller variance.  Indeed, for the \emph{infeasible} choice ${\psi}\parens{p}=c {\alpha}\parens{p}$ for $c>0$ one obtains a bias of zero and a variance equal to
	\[
	\frac{c^3K_1 G^2\parens{b} g\parens{b}}{n^2\parens{n-1}\cparens{ g^2\parens{b} - g'\parens{b}G\parens{b}/n}},
	\]
	which can be made small by choosing $c$ small.  Thus, any gains one obtains from doing a two--step procedure can be obtained by making a different choice of ${\psi}$ and kernel or bandwidth.
}

\drop{
Since \citet{ma2019inference} is about inference rather than estimation, undersmoothing there may be justified.  However, undersmoothing results in a slower convergence rate.  For comparing the estimators themselves it is better to take their respective biases into consideration and compare the mean square errors of estimators using the optimal bandwidth: a comparison on the basis of variances alone is meaningless. Indeed, for the \emph{infeasible} choice ${\psi}\parens{p}=c {\alpha}\parens{p}$ for $c>0$ one obtains a bias of zero and a variance equal to
\[
   \frac{c^3K_1 G^2\parens{b} g\parens{b}}{n^2\parens{n-1}\cparens{ g^2\parens{b} - g'\parens{b}G\parens{b}/n}},
\]
which can be made small by choosing $c$ small.\footnote{Unfortunately, there is no feasible alternative for choosing ${\psi}\parens{p}=c {\alpha}\parens{p}$ unless one assumes the existence of additional derivatives, in which case one can obtain a better convergence rate with a different estimator: these are well--known issues in the nonparametric kernel estimation literature, which apply here also even though our estimator is not a kernel estimator.}  Thus, the comparison is meaningless.\thought{kernels need to have exponentially declining tails: stick assumption in section 3}

A final observation to make here is that the difference in variances observed by \citet{ma2019inference} appears to be due to the use of a two--step estimator instead of a one--step estimator, not of the difference in the method to estimate valuations.  One could simply plug in $\hat v_t = \bar {\alpha}_{T{\psi}}\parens{t/T}$
}

\thoughtinline{The above results should be checked rigorously.}

\drop{
Finally, 

instead of what is described above we could use a two--step procedure and simply estimate the value density $f$ using first--step estimates .  

 We prefer the above formulation because it is easier to analyze the asymptotic behavior of an estimator of ${\alpha}'$ than a density estimator from a pseudo--sample. Since our primary objective is the estimation of other objects and allowing for asymmetry, we do not pursue this comparison further.
}


\subsection{Bid function}

Note that the bid function at $v$ is simply
\(
Q_c\cparens{{\alpha}^{-1}\parens{v}},
\)
and that $Q_c$ can be estimated $\sqrt{T}$--consistently.  Hence the limit distribution of our bid function estimate is simply
$Q_c'/{\alpha}'$ times the limit distribution of the estimate of ${\alpha}$ used.  Since the bid function estimate uses an estimate of the inverse of ${\alpha}$, the estimate of ${\alpha}$ had better be monotonic: this is yet another advantage of imposing monotonicity from the outset.

\subsection{Bidder surplus}

We now turn our attention to estimation of the bidder's surplus. The surplus for bidder one is given by
\begin{equation} \label{eq:BS}
\BS = \Expc{ \parens{V_1-B_1} \one\parens{B_c \leq B_1}}=  \uint{A\parens{p}f_p \parens{p}},
\end{equation}
where $A\parens{p}={\alpha}\parens{p}p - e\parens{p} = Q_c'\parens{p} p^2$.

There are two important and separate cases.  First, in the case of symmetry $F_p$ is known to be $p^{1/\parens{n-1}}$ and does not need to be estimated.  If $F_p$ is unknown then it can be replaced with the empirical distribution function.

Regardless, one would expect $\sqrt{T}$--consistency despite the presence of nonparametric objects in the definition of $\BS$.  This is a common theme in the semiparametric econometrics literature \citep[see e.g.][]{robinson1988root,powell1989semiparametric}.  Even though nonparametric estimators, other than e.g.\ the empirical distribution function, typically converge at a rate slower than $\sqrt{T}$, integrating them often restores the parametric $\sqrt{T}$ rate.  The reason is that integrating is like averaging and hence reduces the variance, which opens up the possibility of undersmoothing to make the bias vanish at a rate faster than $\sqrt{T}$.  Note that if the unsmoothed estimator ${\alpha}_T$ is used, no adjustment of smoothing parameters is needed at all since no smoothing is conducted in the first place.  It does not appear to matter for the asymptotic distribution of our estimator of $\BS$ whether or not a smoothed estimator of $\BS$ is used, as long as it is undersmoothed.  Symmetry matters a lot, however.

In \cref{section:bs symmetry} we discuss the symmetric case and in \cref{section:bs asymmetry} the asymmetric case.



%
%

\subsubsection{Symmetry}
\label{section:bs symmetry}

With symmetry the situation simplifies in that then 
\begin{equation} \label{eq:fpp symmetric}
\fpp = \frac{1}{n-1} p^{\parens{2-n}/\parens{n-1}}.
\end{equation}
This simplifies the asymptotic theory since integration by parts and \cref{eq:fpp symmetric} turns \cref{eq:BS} into
\[
\bs = \frac{e\parens{1}}{n-1} - \uint{
    e\parens{p} \cparens{ f_p'\parens{p} p + 2 f_p\parens{p}} 
} 
= \frac{e\parens{1}}{n-1} - \frac n{n-1} \uint{
    e\parens{p} p^{\frac{2-n}{n-1}}},
\]
which can be estimated by
\[
\bshat^\symm = \frac{\beT\parens{1}}{n-1} - \frac{n}{\parens{n-1}^2} \uint{
    \beT\parens{p} p^{\frac{2-n}{n-1}}}.
\]
The asymptotic theory for $\bshat^\symm$ is trivial in view of \cref{thm:ebreve}.
\begin{thm} \label{thm:bs symmetric not smooth}
    Under the assumptions of \cref{thm:ebreve},
    \[
    \sqrt{T} \parens{\bshat^\symm - \bs} \convd
    N\parens{ 0, \sV_{\bs}^\symm},
    \]
    where
    \[
    \sV_{\bs}^\symm = \frac{n^2}{\parens{n-1}^4}
    \int_0^1 \int_0^1 H\parens{p,p^*} \, \parens{pp^*}^{\frac{2-n}{n-1}} \dif p \dif p^*. \tag*{\qed}
    \]
\end{thm}

It should be pointed out that if symmetry is fully exploited then the function $H$ is different, also.  Indeed, from \cref{eq:Hstar symmetric} it follows that then
\begin{multline*}
 \sV_{\bs}^\symm = \frac{n}{\parens{n-1}^4} \int_0^1 \int_0^1
 Q'\parens[\big]{p^{\frac1{n-1}}} Q'\parens[\big]{p^{*\frac1{n-1}}} \parens{p p^*}^{\frac1{n-1}} \cparens{\min\parens{p,p^*}^{1/\parens{n-1}} - \parens{pp^*}^{1/\parens{n-1}}} \dif p^* \dif p
 \\
 =
 \frac n{\parens{n-1}^2}
 \int_0^1 \int_0^1
 Q'\parens{p} Q'\parens{p^*} \parens{pp^*}^{n-1} \cparens{\min\parens{p,p^*}-pp^*} \dif p^* \dif p,
\end{multline*}
which equals
\< \label{eq:bs symm asvar expressed in terms of G}
\frac{n}{\parens{n-1}^2} \int_0^{\bar b} \int_0^{\bar b}
G^{n-1}\parens{b} G^{n-1}\parens{b^*}
  \sparens{G\cparens{\min\parens{b,b^*}} - G\parens{b} G\parens{b^*}} \dif b \dif b^*,
\>
where we provide \cref{eq:bs symm asvar expressed in terms of G} if readers would like to compare it to a future GPV--based estimator.

\subsubsection{Asymmetry}
\label{section:bs asymmetry}

A natural generic estimator of $\BS$ in the absence of a symmetry assumption is
\< \label{eq:bshat asymmetric}
 \bshat = \uint[F_{pT}\parens{p}]{ \cparens{{\alpha}_T\parens{p}p - \beT\parens{p}}}.
\>
Naturally, ${\alpha}_T$ can be replaced with a smoothed version in which case it would be advisable to replace $\beT$ with the estimator of $e$ corresponding to the smoothed estimate of ${\alpha}$, also.

\begin{thm}\label{thm:bs asymmetric not smooth}
    Under the assumptions of \cref{thm:ebreve}, if $G,G_c$ are estimated using different data and $G_T$ is the empirical distribution function of bids of bidder one then
    \[
    \sqrt{T}\parens{\BShat-\BS} \convd N\parens{0,\sV_\BS^a},
    \]	
 where 
 \< \label{eq:bs asymmetric variance}
 \sV_\BS^a = 
 \uint{
     \uint[p^*]{
         \sparens[\Big]{
             {\Gamma}_1\parens{p} {\Gamma}_1\parens{p^*}  H^*\cparens{Q_c\parens{p},Q_c\parens{p^*}}
             +
             {\Gamma}_2\parens{p} {\Gamma}_2\parens{p^*} H_1^*\cparens{Q_c\parens{p},Q_c\parens{p^*}}}
 }},
 \>
 with
\(
H_1^*\parens{q,q^*}  = G\cparens{\min\parens{q,q^*}} - G\parens{q}G\parens{q^*}
\),
\(
{\Gamma}_2\parens{p}  = {\alpha}'\parens{p}p 
\),
and
\(
{\Gamma}_1\parens{p} = Q_c''\parens{p} p^2 f_p\parens{p} + Q_c'\parens{p} \cparens{ p^2 f_p'\parens{p} + 4pf_p\parens{p}}
\). 
Under \cref{eq:Hstar simple}, the asymptotic variance becomes 
\< \label{eq:bs efficiency bound}
 \Var \frac{G_c^2\parens{b}}{g_c\parens{b}} + \Var\parens[\bigg]{ \frac{G_c^2\parens{b_c}g\parens{b_c}}{g_c^2\parens{b_c}} + 2 \int_0^{b_c} \frac{G_c\parens{t} g\parens{t}}{g_c\parens{t}} \dif t},
\> 
which is the semiparametric efficiency bound for estimators of $\BS$ which only use bids and maximum rival bids for estimation.
\qed
\end{thm}

The asymptotic variance in \cref{eq:bs asymmetric variance} is intimidating but it simplifies considerably in an important special case, as \cref{eq:bs efficiency bound} illustrates.  The fact that our estimator achieves the semiparametric efficiency bound should come as no surprise since our estimator is asymptotically linear and imposing shape restrictions is well--known not to help in reducing the asymptotic variance in many cases.\footnote{See \citet[page 106]{newey1990semiparametric} for a discussion and \citet{tripathi2000local} for results on the semiparametric efficiency bound subject to shape restrictions in the partially linear model of \citet{robinson1988root}.}

Note that the semiparametric efficiency bound is defined only relative to the amount of information available.  For instance, if one uses all bids instead of only the maximum rival bid then \cref{eq:bs asymmetric variance} is less than \cref{eq:bs efficiency bound} but still achieves the semiparametric efficiency bound.  We do not show this.
Nevertheless, it is reassuring that no other regular estimator exists with a smaller asymptotic variance under the same assumptions.

It is not immediately obvious that the variance in \cref{thm:bs asymmetric not smooth} is worse than that in \cref{thm:bs symmetric not smooth}, albeit that the fact that a more efficient estimate of $G_c$ can be used should tip the balance.  We provide a comparison in the least favorable case for symmetry, namely that of two bidders.\footnote{With more than two bidders, the gain in efficiency of estimating $G_c$ is greater.}
\begin{ex} \label{ex:bs variance comparison}
    Suppose that $n=2$ bidders are symmetric and $F_v\parens{v}=v^{1/{\gamma}}$ for some ${\gamma}>0$.  Then $\bar b = 1/\parens{1+{\gamma}}$, $G\parens{b}=G_c\parens{b}=\cparens{\parens{1+{\gamma}}b}^{1/{\gamma}}$, $g\cparens{Q_c\parens{p}}=\parens{1+{\gamma}} p^{1-{\gamma}}/{\gamma}$. $g'\cparens{Q_c\parens{p}}=\parens{1+{\gamma}}^2\parens{1-{\gamma}} p^{1-2{\gamma}}/{\gamma}^2$, $Q\parens{p}=Q_c\parens{p}=p^{\gamma}/\parens{1+{\gamma}}$,
    $Q'=Q_c'={\gamma} p^{{\gamma}-1}/\parens{1+{\gamma}}$,
    $Q''=Q_c''={\gamma}\parens{{\gamma}-1}p^{{\gamma}-2}/\parens{1+{\gamma}}$, and $e\parens{p}=p^{{\gamma}+1}/\parens{1+{\gamma}}$.  Thus, from \cref{eq:bs symm asvar expressed in terms of G} it follows using some tedious calculus that
    \[
     \sV_\BS^\symm =  \frac{2{\gamma}^2}{\parens{1+{\gamma}}^2\parens{2+{\gamma}}^2\parens{3+2{\gamma}}},
    \]
    which equals $1/90$ for a uniform value distribution.  To obtain $\sV_\BS^a$ note that ${\Gamma}_2\parens{p}= {\gamma}p^{\gamma}$, ${\Gamma}_1\parens{p}={\gamma}\parens{3+{\gamma}}p^{\gamma}/\parens{1+{\gamma}}$, and $H^*\cparens{Q_c\parens{p},Q_c\parens{p^*}}=H_1\cparens{Q_c\parens{p},Q_c\parens{p^*}}=\min\parens{p,p^*} - pp^*$ which (after some tedious calculus) yields
    \[
    \sV_\BS^a = 
    \frac{2 {\gamma}^2 \parens{5+4{\gamma} + {\gamma}^2}}{\parens{1+{\gamma}}^2 \parens{2+{\gamma}}^2\parens{3+2{\gamma}}},
    \]
    which equals $1/9$ in the uniform $F_v$ case.
    The variance in the asymmetric case is $5+4{\gamma}+{\gamma}^2$ times as large as in the symmetric case.  Since assuming symmetry speeds up convergence of $\beT$ \emph{and} obviates the need to estimate $F_p$, it was clear that the ratio would exceed two.  But in the uniform $F_v$ case the factor is ten!\footnote{That's not a factorial.} \qed
\end{ex}

The conclusion from \cref{ex:bs variance comparison} must be that symmetry should be imposed whenever reasonable. If the researcher is not willing to assume symmetry and pool bids in estimating $e_T$, the minimum relative entropy estimator for $f_p$ may close part of the gap between $\sV_\BS^a$ and $\sV_\BS^\symm$.

As it turns out, smoothing does not improve the asymptotic distribution.  Too much smoothing can introduce an asymptotic bias and slow down convergence.  The most important consideration is that the implied estimator of $e$ converges (after norming and scaling) to the same Gaussian limit process as $\beT$ which for the smoothed estimator simply requires that $h \to 0$ fast enough as $T\to\infty$.  We state the theorem for $\bar {\alpha}_{T{\psi}}$ but the result applies with minor modifications to any estimators which satisfy the aforementioned desiderata.
\begin{thm} \label{thm:bs asymmetric smoothed}  Suppose that the assumptions of \cref{thm:ebarpsi} are satisfied.  Then the same limit distribution obtains if one replaces ${\alpha}_T,\beT$ in \cref{thm:bs asymmetric not smooth} with $\bar {\alpha}_{T{\psi}},\bar e_{T{\psi}}$ and chooses a bandwidth $h$ which tends to zero faster than $T^{-1/4}$.  \qed
\end{thm}

Note that the bandwidth can tend to zero arbitrarily fast since a bandwidth of zero simply takes us back to the unsmoothed estimator.  This is in sharp contrast to other approaches, e.g.\ one based on the estimator of the inverse bid function in GPV, where taking the bandwidth to zero \emph{before} taking the sample size to infinity would blow up the asymptotic variance: letting $h \downarrow 0$ with GPV does not produce a consistent estimator of $g_c,g,{\alpha}$.  Consequently, it is not clear a priori that using a second order kernel and undersmoothing GPV would produce a consistent estimator of $\BS$, let alone a $\sqrt{T}$--consistent estimator.  We have no theoretical results on this, though our simulation results suggest that letting the bandwidth go to zero in a GPV--based estimator of $\BS$ would break $\sqrt{T}$--consistency.

\subsection{Mean of the value distribution}

The mean of the value distribution of bidder one is
\begin{equation}
\label{eq:MV}
 \MV= \int_0^{\bar v} v f_v\parens{v} \dif v = \uint{ {\alpha}\parens{p} \fpp}
\end{equation}
There are several ways of estimating $\MV$. For instance, one can estimate it using estimates of the bid distributions directly since
\[
\MV = \int_{0}^{\bar b} \parens[\Big]{b + \frac{G_{c}\parens{b}}{g_{c}\parens{b}}} \dif G_{1}\parens{b}\,,
\]
which would be most natural if one used the GPV machinery. However, we will present results for
\[
\MVhat = \int_{0}^{1}\alpha_{T}\parens{p} \dif F_{pT}\parens{p}\,.
\]
The asymptotics for $\MVhat$ are similar to those for $\BShat$ and the proof is therefore mercifully short.

%

\begin{thm} \label{thm:mean v symmetric}
    Under the assumptions of \cref{thm:ebreve},\joris{perhaps we need a condition about the behavior of $H$ near zero.}
\[
\sqrt{T} \uint[F_p\parens{p}]{ \cparens{{\alpha}_T\parens{p}-{\alpha}\parens{p}}} \convd N\parens{0,\sV_\MV^\symm},
\]	
where
\[
 \sV_\MV^\symm = \frac{\parens{n-2}^2}{\parens{n-1}^4} \uint{
\uint[p^*]{
H\parens{p,p^*} \parens{pp^*}^{\frac{3-2n}{n-1}} 
} 
},
\]
which under \cref{eq:Hstar symmetric} simplifies to	
\[
 \frac{\parens{n-2}^2}{\parens{n-1}^2n} \int_0^1 \int_0^1
 Q'\parens{p} Q'\parens{p^*} \cparens{ \min\parens{p,p^*}-pp^*} \dif p^* \dif p,  
\]
which can alternatively be expressed as
\(
\sparens{\parens{n-2}^2 /\cparens{\parens{n-1}^2n}} \int_0^{\bar b} \int_0^{\bar b}
\sparens{G\cparens{\min\parens{b,b^*}} - G\parens{b} G\parens{b^*}} \dif b^* \dif b. 
\)
\qed
\end{thm}

Note that the asymptotic variance in \cref{thm:mean v symmetric} equals zero if $n=2$.  This is intuitive since then the mean of the value distribution is simply $\bar b$, which can be estimated super--consistently.    Naturally, some of these properties evaporate once we examine the asymmetric case.

\begin{thm}\label{thm:mean v}
 Under the assumptions of \cref{thm:ebreve},\joris{do we need more?}
\[
 \sqrt{T} \parens{\MVhat-\MV} \convd N\parens{0,\sV_\MV^a}, 
\]
where $\sV_\MV^a$ is like $\sV_\BS^a$ with ${\Gamma}_1,{\Gamma}_2$ divided by $p$. \qed
\end{thm}

Note that the symmetry assumption is also easy to exploit without using our machinery, because $G_c(b)/g_c(b)= G(b)/\cparens{(n-1)g(b)}$ and 
\[
\MV = \int_0^{\bar b} b \dif G(b) + \int_0^{\bar b} \frac{G(b)}{n-1} \dif b\ =
\frac{\bar b + \parens{n-2} \Exp b}{n-1}. 
\]
Since the upper bound of the bid distribution can be estimated at a rate faster than $\sqrt{T}$, estimation of $\MV$ by replacing $\bar b$ and $\Exp b$ with their sample counterparts would work, also.  So for the purpose of estimating the mean of the value distribution in the symmetric case, our methodology is probably overkill.

%

\subsection{Profit}

Estimating the seller's profit is a trivial exercise (if the seller's valuation is zero as we assume throughout) since profit is simply the sum of the winning bids.  A more interesting object is profit as a function of a hypothetical reserve price $r$, $\PR\parens{r}$, or number of bidders $\PR^*\parens{n}$.

In the asymmetric case, this is a complicated endeavor.  Indeed, using the machinery described earlier in the paper we can recover the value distributions for each bidder.  However, there is generally no analytical solution for the bid function in the asymmetric case like there is in the symmetric case. We must therefore numerically solve for the counterfactual equilibrium bid distributions in order to compute the counterfactual revenue. Because this method does not depart from the existing literature,  we limit our discussion to the symmetric case, where we do have suggestions for how to exploit the symmetry assumption in the counterfactual analysis.  Since the theoretical results here are similar to those obtained earlier in terms of method of proof, we state the results in the text instead of enunciating them.

\subsubsection{Counterfactual number of bidders}
We have repeatedly made use of the fact that $F_{p}$ is a known function of $n$, the number of bidders.  $F_p$ is still a known function of any counterfactual number of bidders $m$, possibly different from $n$. In addition, the counterfactual expected payment function is a known function of the factual expected payment function if the distribution of valuations is held fixed. Specifically, the equilibrium $\alpha$ for a given number of bidders $n$ satisfies $\alpha\parens{\tau^{n-1};n} = Q_{v}(\tau)$ for all $n,{\tau}$,
which implies that for ${\xi}={\xi}_{mn}=\parens{n-1}/\parens{m-1}$, $\alpha(p;m) = \alpha\parens{p^{\xi};n}$ for all $p\in[0,1]$ and $n,m\geq 2$.\footnote{Note that the bid distribution changes with $n$ but the value distribution remains constant.
Indeed,	
$Q_v\parens{{\tau}} = Q\parens{{\tau};n} + {\tau} Q'\parens{{\tau};n}/ \parens{n-1}$ defines $Q\parens{\cdot;n}$ as an implicit function of $Q_v$, indeed
$Q\parens{{\tau};n}= \int_0^1 Q_v\parens{t^{1/\parens{n-1}}{\tau}} \dif t$.
}

The counterfactual expected payment function is then
\(
e(p;m) = \int_{0}^{p}\alpha\parens{t^{\xi}} \dif t
\)
and the expected revenue is given by
\begin{multline*}
\PR^*\parens{m} = \int_{0}^{1}e(p;m) \dif F_{p}(p;m) 
=
\int_0^1 \int_0^p {\alpha}\parens{t^{\xi}} \dif t \dif F_p\parens{p;m}
=
\int_0^1 {\alpha}\parens{p^{\xi}} \cparens{1 - F_p\parens{p;m}} \dif p 
=
\\
\uint{ \parens[\Big]{\frac{{\chi}_2+1}{{\xi}} p^{{\chi}_2} - \frac{{\chi}_1+1}{{\xi}} p^{{\chi}_1}} e\parens{p}},
\end{multline*}
where ${\chi}_j= \parens{m-2n+j}/\parens{n-1}$.

Following the previous examples, a plug--in estimator for $\PR^*\parens{m}$ in which we substitute an estimate of $e$ and the known $F_{p}(\cdot;m)$ converges at a $\sqrt{T}$--rate.  Indeed, the limit distribution is a mean zero normal with variance
\[
 \uint{ \uint[p^*]{
  H\parens{p,p^*}	\parens[\Big]{\frac{{\chi}_2+1}{{\xi}} p^{{\chi}_2} - \frac{{\chi}_1+1}{{\xi}} p^{{\chi}_1}}\parens[\Big]{\frac{{\chi}_2+1}{{\xi}} {p^*}^{{\chi}_2} - \frac{{\chi}_1+1}{{\xi}} {p^*}^{{\chi}_1}}
 }
}.
\]

\subsubsection{Counterfactual reserve prices}

By (1) in \citet{jun2019information}, we have
\[
\PR\parens{r} = \bar v - r F_v^n\parens{r} + \int_r^{\bar v} \cparens[\big]{F_v^n\parens{v} - n F_v^{n-1}\parens{v}} \dif v.
\]
Using the substitution $p = F_v \parens{v}$ the problem then entails finding $p^*=F_v^{n-1}\parens{r}$ for which
\< \label{eq:profit condition}
\PR\cparens{{\alpha}\parens{p^*}} = n\parens[\bigg]{ p^*{\alpha}\parens{p^*}\parens[\big]{1-p^{*1/(n-1)}} + \frac1{n-1} \int_{p^*}^1\int_{p^*}^p {\alpha}(u) \dif u \; p^{\frac{2-n}{n-1}} \dif p}
\>

It should be apparent from our earlier discussion that since ${\alpha}$ is estimated at a slower--than--parametric rate, the first right hand side term in \cref{eq:profit condition} is estimated at a rate less than $\sqrt{T}$ but the second right hand side term in \cref{eq:profit condition} can be estimated at the typical parametric rate.  The asymptotic distribution is hence determined by the estimation of $F_v\parens{r}$, which was already discussed in \cref{section:value distribution}. The choice of bandwidth should therefore be made with an eye toward the precise value or range of values of counterfactual reserve prices under consideration.

\drop{\grey{Importantly, we do not need to estimate the entire distribution of private valuations in order to estimate $\PR\cparens{{\alpha}\parens{p^*}}$ because the estimated counterfactual equilibrium revenue only depends on the primitives of the model through $p^*$ and $\alpha$. Thus, the optimal sequence of bandwidths would minimize the MSE of an estimate of $\alpha(p^*)$ as opposed to an asymptotic integrated criterion. In fact, if we pose the question in terms of the probability of no sale occurring instead of reserve prices, we could estimate the seller's expected revenue in an auction that does not result in a sale, say 5\% of the time, without estimating $F_v$ at any point. To do so, we would simply substitute $r^*=Q_v\parens{0.05^{1/n}}$ in \eqref{eq:profit condition}. Indeed, when auctioneers are directly interested in the trade-off between the probability of a no-sale and their expected revenue, this formulation of the question is especially appropriate.  For example, the US Forest Service's directive is to provide a steady supply of timber at competitive prices, which suggests that precisely this trade-off is central to their mission.}\joris{I did not fully understand this para..... needs attention}}


\drop{
\section{Comparison of the estimators}

\subsection{Asymptotic comparison}
\jorisinline{This bit is not specific to boundary corrections, so should be rewritten accordingly.  Not sure where best to stick this part.}
\karlinline{Took it out of metrics and stuck it here. Makes sense to me to have analytical comparison next to simulation evidence.}
\jorisinline{Should this be in the paper at all?  The main reason to have it in would be to make a comparison of estimators of the private values themselves, which would make sense.  In that case, perhaps we should rewrite it and stick it into section 6.}

The first question to ask is how our estimator of ${\alpha}$ compares to the corresponding (first stage) GPV estimator.  In other words, how well can we recover the private values using either method. The first stage of the GPV estimator for a given value $b$ is
\[
b + \frac{ \hat G_c\parens{b}}{\hat g_c\parens{b}}.
\]
To make the estimators comparable, we compare the performance of our estimator to that of
\[
Q_c\parens{p} + \frac{\hat G_c\cparens{Q_c\parens{p}}}{\hat g_C\cparens[\big]{Q_c\parens{p}}}.
\]
\Cref{ex:extensive comparison of {\alpha}} compares the asymptotic properties of both estimators for the special case of power distributions.  Although there is in fact symmetry in the design of the example, the example does not use such `knowledge.'  We focus on an asymptotic comparison away from the boundaries so as not to confound our comparison with differences in boundary correction schemes, though we do not mean to suggest that differences in behavior at the boundary are irrelevant.

\grey{The fact that a boundary kernel is being used is immaterial in the comparison since it is an asymptotic comparison away from the boundaries.   A limitation of \cref{ex:extensive comparison of {\alpha}}, then, is the fact that it ignores the potentially relevant differences in behavior at the boundaries.}

\begin{ex} \label{ex:extensive comparison of {\alpha}}
	Suppose that bidders are symmetric and $F_v\parens{v}=v^{\delta}$ for some ${\delta}>0$.  Then $G_c\parens{b}= \cparens{\parens{1+{\gamma}}b}^{1/{\gamma}}$, $e\parens{p}=p^{1+{\gamma}}/\parens{1+{\gamma}}$ for ${\gamma}=1/\cparens{\parens{n-1}{\delta}}$, such that $e'\parens{p}=p^{\gamma}$, $e''\parens{p}={\gamma} p^{{\gamma}-1}$, $e'''\parens{p}={\gamma} \parens{{\gamma}-1}p^{{\gamma}-2}$, and $g_c\parens{b}= \parens{1+{\gamma}}^{1/{\gamma}} b^{1/{\gamma}-1}/{\gamma}$, $g_c'\parens{b} = \parens{1+{\gamma}}^{1/{\gamma}}b^{1/{\gamma}-2} \parens{1/{\gamma}-1}/{\gamma}$, 
	$g_c''\parens{b}=\parens{1+{\gamma}}^{1/{\gamma}}b^{1/{\gamma}-3} \parens{1/{\gamma}-1}\parens{1/{\gamma}-2}/{\gamma}$.
	
	Suppose that ${\psi}'\parens{p}=p^{-{\lambda}}/c$ for some $0\leq {\lambda}\leq 1$.\footnote{The multiplicative constant $c$ doesn't matter since it can be absorbed into the bandwidth, but is included here to facilitate the comparison between estimators.}
	At fixed $0<p<1$, ${\zeta}\parens{p}={\gamma} p^{\gamma}/\parens{{\gamma}+1}$ such that the variance expressions for $\hat {\alpha}_{T{\psi}}$ and the GPV estimators are (without multiplying by $\sqrt{Th}$)
	\[
	\sV_{\psi}^*= \frac{1}{Th\sqrt{{\pi}}} \times \frac{1}{c}\parens[\Big]{\frac {\gamma}{1+{\gamma}}}^2 p^{2{\gamma}-{\lambda}}, \qquad
	\sV_\GPV^*=\frac{{\kappa}_2}{Th} \times \parens[\Big]{\frac {\gamma}{1+{\gamma}}}^3 p^{3{\gamma}-1},
	\]
	respectively. 
	The bias is\jorisr{should we compare this to that for $\hat {\alpha}$, also?}\karl{I think this could be helpful as a guide for when the extra computation might be helpful.}
	\[
	\bar\sB_{\psi}^*= \frac{h^2}{2} \times c^2{\gamma}\parens{{\gamma}+{\lambda}-1}p^{{\gamma}+2{\lambda}-2},
	\quad
	\sB_\GPV^* = -\frac{h^2}{2} \times \frac{\parens{1-{\gamma}^2}\parens{1-2{\gamma}}}{{\gamma}p^{\gamma}}.
	\]
	First, note that if a normal kernel is used in GPV then ${\kappa}_2=1/\sqrt{{\pi}}$, so that difference is immaterial.  Further, note that if we make the asymptotic variances equal by choosing ${\lambda}=1-{\gamma}$ and $c=\parens{1+{\gamma}}/{\gamma}$ then the asymptotic bias of our estimator is zero, but the GPV estimator is biased downward.  Unfortunately, we know neither the value distribution nor its parameters so choosing ${\psi}$ to make the bias zero would require a two--step procedure or a higher order kernel.
	
	Next, suppose that no transformation is taken, i.e.\ ${\lambda}=0$.  For $c= \parens{1+{\gamma}} p^{1-{\gamma}}/{\gamma}$, the variances are the same, but the bias of our estimator is larger.
	
	In view of the presence of $n-1$ in the denominator of the definition of ${\gamma}$, ${\lambda}=1$ is a more reasonable choice.  If ${\gamma}$ is close to zero then both the bias and variance of our estimator are close to zero, also.  The variance of the GPV estimator is all in this case, also, but the bias is large.
	\older{
		The two biases are equal if ${\gamma}=1$, e.g.\ if values are uniformly distributed and $n=2$.  However, ${\gamma}$ will typically be less than one, so for small $p$, the GPV bias will be larger.  For ${\gamma}<0.27568$, in this example the bias is less for our estimator than for GPV for all values of $p$, but that would require the indicated choice of ${\psi}$. 
		
		But that choice of ${\psi}$ is suboptimal from a mean square error perspective. Indeed, it can be shown that if ${\psi}\parens{p}=c p^{3/7}$ in the above example --- irrespective of the choice of ${\gamma}$ --- then the mean square error is proportional to $p^{2{\gamma}-3/7}$, which suggests that for ${\gamma}<3/14$ uniformity of the $T^{2/5}$ convergence rate will be absent as $p\to 0$.  The GPV estimator has this undesirable property for all values of $0<{\gamma}<1$ as the above bias formula demonstrates.
		
		For both estimators, we have another lever to use: the bandwidth.  Indeed, if we pick ${\psi}\parens{p}=\log p$ and $h=h\parens{p}\propto p^{-3/5}$ then the mean square error for our estimator will be proportional to $p^{2{\gamma}-2/5}$, which improves the range of bad ${\gamma}$ values a bit to ${\gamma} < 1/5$.  For GPV, choosing $h \propto p^{\parens{4{\gamma}-1}/3}$ will limit the range of `poorly behaved' ${\gamma}$ values to ${\gamma}<4/13$, but even that requires one to know ${\gamma}$. 
	}
	\qed
\end{ex}

The main take--away from \cref{ex:extensive comparison of {\alpha}} is that there is no natural ordering of the two estimators, i.e.\ there can be instances in which either estimator has better properties.

\thoughtinline{Not sure about all the inline todo--notes; we need to think about those.}

\jorisinline{Stick in a numerical example that compares the two estimators.  I'm thinking of something analogous to \texttt{alpha-performance.pdf} but with GPV in there, also.}

\jorisinline{Would it be worthwhile to work out results for estimated ${\psi}$?}

}

\section{Monte Carlo simulations}
\label{sec:sims}
We provide a simulation study to compare the performance of our estimators. Our goal is not to crown a winner but to highlight systematic ways in which various methodological choices impact the bias and mean squared error of the estimator.



\subsection{Simulation parameters}
We parameterize bidder one's maximum competitor bid distribution as $G_{c}\parens{b} \propto \parens{\theta/b  + \gamma - \theta}^{-1/\theta}$ for $b\in\sparens{0,\bar b}$ with $\bar b = 2 \big/ \parens[\big]{1+\theta + \sqrt{4 \gamma +(\theta-1)^{2}}}$ and $\gamma,\,\theta>0$. Bidder one's inverse bid function is then $\beta^{-1}(b) = (1+\theta)b + (\gamma-\theta)b^{2}$.  Note that $\beta^{-1}$ is strictly increasing on $[0,\bar b]$, which implies convexity of the expected payment function $e(p) = \theta c p^{\theta + 1} \big/ \cparens[\big]{1 + c (\gamma-\theta)p^{\theta}}$, where $c$ is a constant that depends on $\gamma$ and $\theta$. 

The maximum competitor bid distribution is chosen such that the support of bidder one's valuations is $\sparens{0,1}$ regardless of the values of $\gamma$ and $\theta$. We can then fix bidder one's valuation distribution and independently vary the maximum competitor bid distribution to achieve various shapes of the inverse strategy function $\alpha$ and competitor bid density, which are the respective targets of estimators based on our approach or estimators based on the inverse bid function (IBF) like GPV.\joris{Didn't we commit to calling it the inverse strategy function somewhere else?  Also, sentence rewritten.}\karl{we don't mention IBF until next para} \Cref{fig:alphas,fig:competitor dens} plot these functions. If $\gamma=\theta$ then the competitor bid distribution is a power distribution and the inverse bid function is simply linear. As $\theta$ approaches zero, the competitor bid distribution approaches a truncated Fr\'echet distribution and the inverse bid function a convex quadratic.  

\begin{figure}[ht]
	\centering
   \includegraphics[width = .7\textwidth]{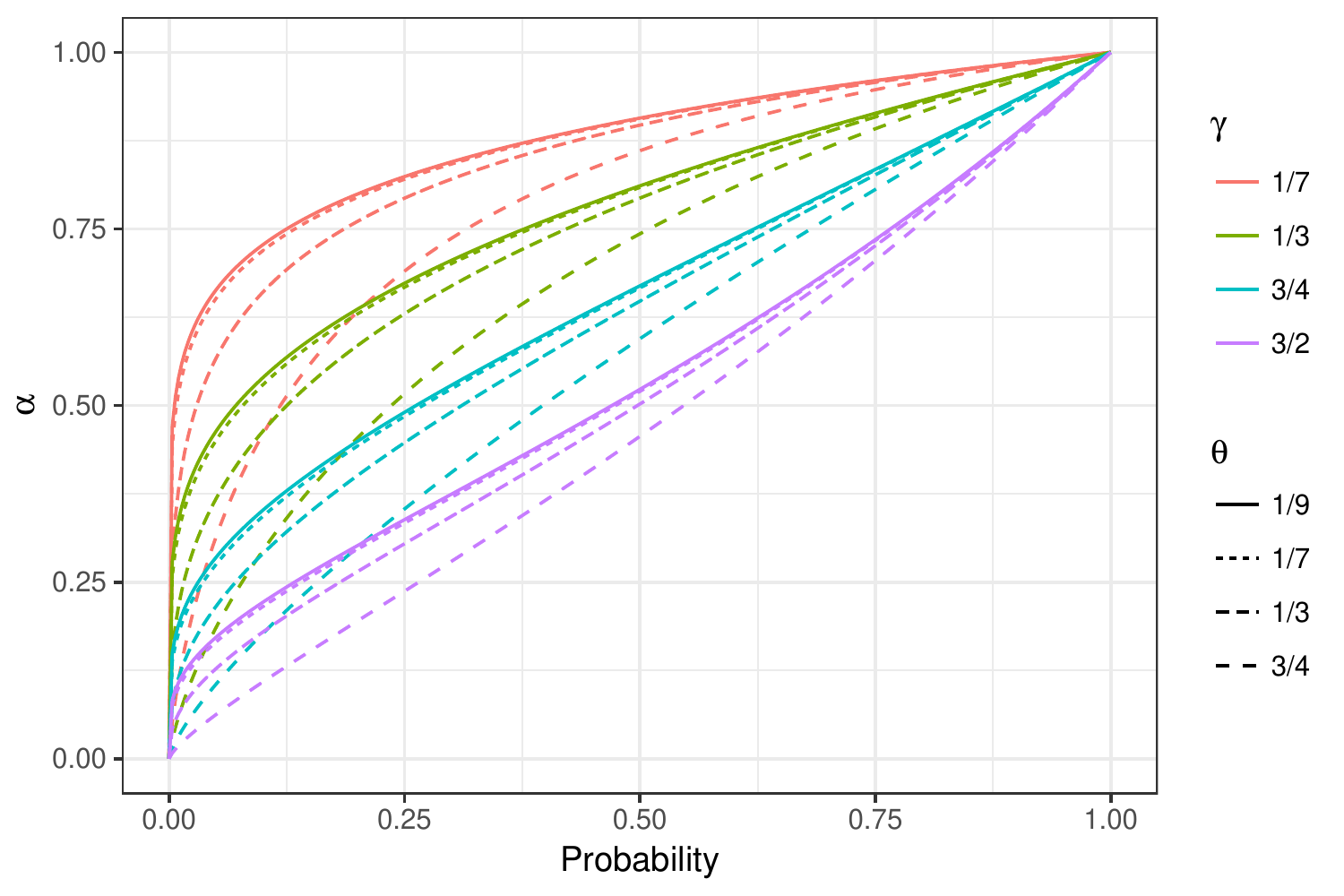}
	\caption{The inverse strategy function for various parameter values.	\label{fig:alphas}
	}
\end{figure}

\begin{figure}[ht]
	\centering
   \includegraphics[width=.7\textwidth]{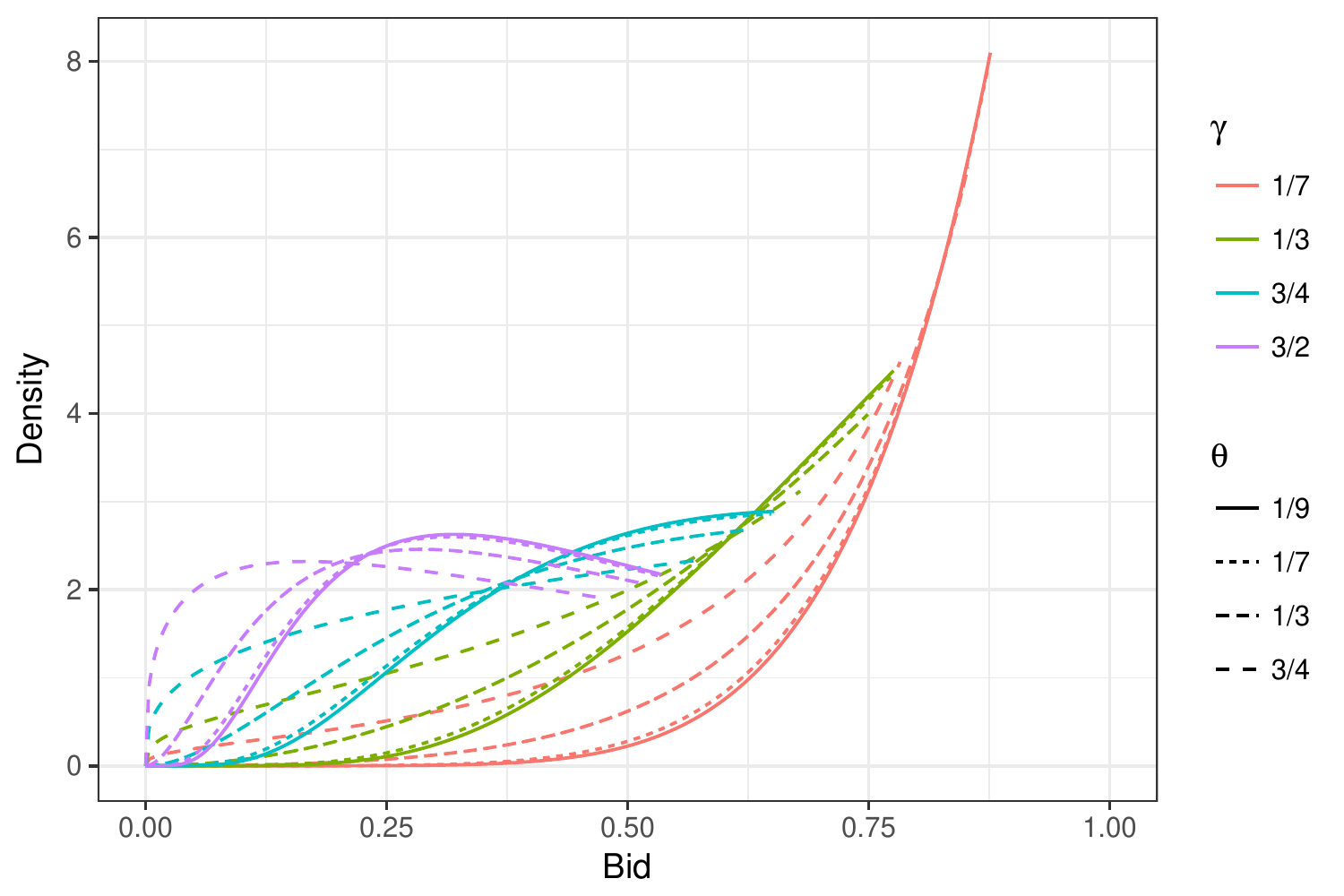}
   \caption{The competitor bid density for various parameter values.\label{fig:competitor dens}}
\end{figure} 

For every combination of $\gamma = 3/2,\, 3/4, \,1/3,\, 1/7$ and $\theta=3/4, 1/3, 1/7, 1/9$, we draw independent and identically distributed samples of $T=100$, 250, and 500 maximum competitor bids. Thus, $T$ represents the number of auctions as well as the number of bids used to estimate bidder one's expected payment function and inverse strategy. We then independently sample $T$ draws of bidder one's valuations according to a power distribution $F_{v_1}(v) = v^{3/2}$, compute her optimal bid, and apply our various methodologies to estimate several objects of interest. We compare our estimators with an estimator based on an approach similar to GPV in which only the independent sample of highest competitor bids are used to estimate the inverse bid function. This estimator is labeled ``IBF'' to indicate the estimates of the various objects were constructed from a nonparametric estimate of the inverse bid function. The IBF estimator does not perform any boundary correction or trimming, hence cannot be expected to perform well near the boundary.  To be clear, this is not a critique of \citet{Guerre2000}: the goal in their paper is to estimate the inverse bid function and valuation density at an optimal rate in the interior of the support of the valuations. Here, we use the estimator for the inverse bid function as an input into other objects of interest. A fairer comparison with our boundary--corrected estimators can be found in ``IBF--BC.'' This estimator uses a boundary correction routine similar to \citet{hickman2015replacing} and is hence better suited for estimating objects that require integration over the entire support of the bids or valuations.

All simulations employ an Epanechnikov kernel and a rule--of--thumb bandwidth sequence, multiplied by an additional scaling factor of $1/5, 1/2, 1,$ or $3/2$ in order to explore sensitivity to the choice of bandwidth. We use a Gaussian reference distribution for choosing bandwidths for methods that use nonparametric kernel density estimators and $\alpha(p) = \bar v p^\gamma$ as our reference function for bandwidths using our procedure. Specifically, we use the sample mean and variance of ${\alpha}_T\cparens{G_{cT}(b_{1})}$ to estimate the parameters $\bar v$ and ${\gamma}$ in the parametric reference model, then choose the bandwidth that would minimize the mean integrated squared error of the estimator for ${\alpha}$ under the reference model. This optimal bandwidth also depends on $\psi$.\footnote{For some choices of ${\psi}$, the squared error is not integrable on $[0,1]$. In these cases, our rule of thumb minimizes the integrated squared error on $[0.05,1]$.}\ 

We consider five different choices of ${\psi}$. The first is the identity transformation ${\psi}_1(p) = p$ and the second is the infeasible zero--bias transformation ${\psi}_2(p) = {\alpha}(p)$. The next transformation ${\psi}_3(p) = \log(p)$ ensures that the asymptotic bias is vanishingly small for $p$ close to zero, though the asymptotic variance can be large. The transformation ${\psi}_4=\sqrt{p}$ minimizes the MISE in ${\alpha}$ if ${\alpha}$ is a power function with exponent greater than $1/2$.\footnote{If ${\alpha}$ is a power function and the exponent is less than $1/2$, the optimal ${\psi}$ would be the infeasible choice ${\psi}_2$.}  Finally, ${\psi}_5(p) = \sqrt[5]{p}$ balances the integrated asymptotic bias and variance of the estimator for ${\alpha}$ when ${\alpha}$ is a power function, regardless of the exponent.

For ${\psi}_2$, the rule of thumb suggests that an infinite bandwidth would minimize the integrated MSE because the first--order bias is always zero. A better rule of thumb would suggest a bandwidth sequence on the order of $T^{-1/7}$, resulting in a faster rate of convergence than the other estimators. For the sake of comparison, we do not take this route and instead use the same bandwidth as we do for ${\psi}_5$. For the undersmoothed estimates, we simply multiply the rule--of--thumb bandwidths by $T^{-2/15}$ so that the sequences are on the order of $T^{-1/3}$.

For the boundary corrected estimates that use reflection---IBF--BC and $\bar {\alpha}_{T}^R$---the auxiliary bandwidth is proportional to the main bandwidth. In \cite{pinkse2019actual}, we show that choosing bandwidths converging at a rate of $T^{-1/7}$ would be optimal for both estimators if ${\alpha}$ (equivalently $g_c$) has three continuous derivatives near the boundary. In this paper, we do not choose a bandwidth sequence to capitalize on this extra smoothness because doing so would put the reflection methods at an advantage relative to the boundary kernel estimators. Unlike the reflection methods, our boundary kernel method does not involve any auxiliary input parameters that could be modified to take advantage of this extra smoothness.  That said, we note that if the researcher is willing to strengthen the smoothness assumption, the reflection method or a boundary kernel method that takes advantage of the extra smoothness could be more attractive in practice precisely for this reason.\footnote{In fact, if the target of the estimation were the valuation density, the researcher might assume three derivatives of $g_c$, anyway, in order to attain the typical $T^{2/7}$ rate of convergence.}

All simulations use a thousand replications.

\subsection{Simulation results: $\sqrt{T}$--consistent estimators}
We first review the simulation results for the integrated objects MV and BS. \Cref{fig:medals BS MV} illustrates the relative root mean squared error (RMSE) of our unsmoothed and smoothed, boundary--kernel--based estimators along with the IBF estimators. The bandwidths are chosen proportional to $T^{-1/3}$ so that the resulting $\sqrt{T}$--consistent estimator is asymptotically unbiased. For lack of a better rule, the constant of proportionality in the bandwidth sequence is simply the rule--of--thumb constant multiplied by our additional scale factor. The various estimators are arranged in columns, and each row represents a different combination of the target object, bandwidth scaling factor, ${\gamma}$, ${\theta}$, and $T$. The value in each cell is colored to reflect the value of the RMSE divided by the minimum RMSE across the columns. The lightest green indicates the best performing estimator, while the darkest purple indicates the RMSE was at least three times as large as that of the best performing estimator. The color scale is top--coded because some estimators performed extremely poorly.

\begin{figure}[htp]
	\includegraphics[angle=90, height = \textheight]{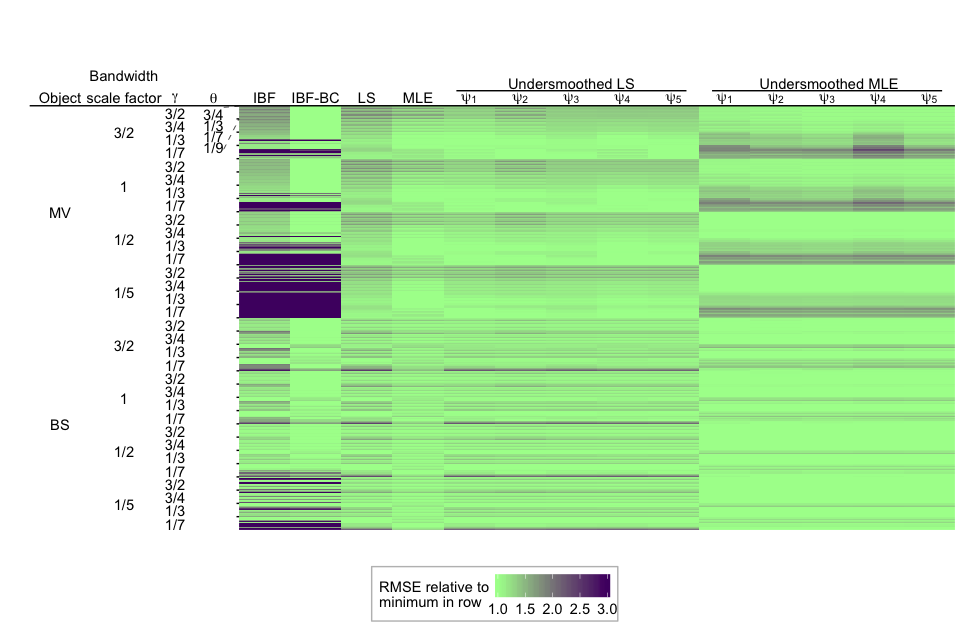}
	\caption{Relative performance of $\sqrt{T}$--consistent estimators. $T$ is the most frequently repeating parameter (in decreasing order) along the vertical axis. The first row reflects estimates of the mean valuation when the rule-of-thumb bandwidths are scaled by 3/2, ${\gamma}=3/2$, ${\theta}= 3/4$, and $T=500$.	\label{fig:medals BS MV}
	}
\end{figure}

The unsmoothed MLE consistently performs well across a variety of parameter values, while the unsmoothed isotonic regression estimator (LS) has difficulty for some parameter values because it suffers from finite--sample bias for values of $p$ close to one. Intuitively, this bias arises from the fact that the GCM, by definition, must lie below the estimate of the true expected payment function, which leads to an upward bias in its slope near $p=1$. This bias is more pronounced when ${\gamma}$ and ${\theta}$ are both relatively large, because the true expected payment is more convex near the right boundary. In contrast, the unsmoothed MLE is not as badly biased when ${\gamma}$ and ${\theta}$ are large, because the graph of the MLE for $e$ does not have to lie below the unconstrained estimator. The finite sample bias in estimates of ${\alpha}$ for large values of $p$ more negatively affects the relative performance in estimating the bidder's expected surplus because the values of $\alpha(p)$ for large $p$ are weighted relatively more in the integral formula for BS than for MV.  We expect these differences in the unsmoothed estimators to vanish as $T$ increases because all our estimators in \cref{fig:medals BS MV} are asymptotically equivalent.

When ${\gamma}$ is small, the undersmoothed, boundary--corrected IBF estimator for BS and MV appears to under--perform in small samples, but otherwise has a relatively small RMSE. As expected, however, the relative performance of the IBF approach is sensitive to the scale of the bandwidth sequence. We cannot conclude from \cref{fig:medals BS MV} that our approach is robust to the choice of bandwidth, however, because it does not compare the relative performance of different bandwidth scaling factors. \Cref{tbl: BS RMSE T = 500} makes this comparison in the estimation of the bidder's surplus using $T=500$ auctions. The results demonstrate that the asymptotic behavior of our undersmoothed estimators is fairly similar across bandwidth sequences, whereas the IBF--BC estimator can be the best performing for some parameter values or worst performing estimator depending on the choice of bandwidth.  We consider this robust performance, and indeed not having to choose an input parameter, a valuable characteristic of our approach.

\begin{table}[!htbp] \centering 
	\caption{RMSE of estimators for the bidder's expected surplus using $\psi_5(p) = p^{1/5}$ and $T =$ 500 auctions relative to the minimum RMSE across all combinations of estimators and bandwidth scaling factors. Relative values are multiplied by 1000.} 
	\label{tbl: BS RMSE T = 500} 
	\footnotesize 
	\begin{tabular}{@{\extracolsep{-5pt}} cccccccccc} 
		\\[-1.8ex]\hline 
		\hline \\[-1.8ex] 
		& & \multicolumn{8}{c}{$(\gamma,\theta)$}\\
		\cline{3-10}
		& Bandwidth scaling factor & $(1/7,1/9)$ & $(1/7,1/7)$ & $(1/7,1/3)$ & $(1/7,3/4)$ & $(1/3,1/9)$ & $(1/3,1/7)$ & $(1/3,1/3)$ & $(1/3,3/4)$ \\ 
		\hline \\[-1.8ex] 
		IBF--BC & $0.2$ & $1671$ & $3596$ & $1741$ & $1000$ & $1153$ & $1103$ & $1848$ & $1008$ \\ 
		IBF--BC & $0.5$ & $1306$ & $1510$ & $1026$ & $1037$ & $1036$ & $1041$ & $1000$ & $1067$ \\ 
		IBF--BC & $1$ & $1244$ & $1050$ & $1041$ & $1050$ & $1025$ & $1026$ & $1008$ & $1084$ \\ 
		IBF--BC & $1.5$ & $1212$ & $1000$ & $1049$ & $1059$ & $1020$ & $1020$ & $1018$ & $1093$ \\ 
		LS & $0$ & $1465$ & $1262$ & $1000$ & $1010$ & $1104$ & $1115$ & $1326$ & $1000$ \\ 
		MLE & $0$ & $1148$ & $1070$ & $1059$ & $1059$ & $1041$ & $1044$ & $1094$ & $1064$ \\ 
		Sm. LS & $0.2$ & $1434$ & $1201$ & $1008$ & $1014$ & $1093$ & $1103$ & $1264$ & $1007$ \\ 
		Sm. LS & $0.5$ & $1399$ & $1177$ & $1016$ & $1018$ & $1085$ & $1095$ & $1208$ & $1014$ \\ 
		Sm. LS & $1$ & $1386$ & $1128$ & $1024$ & $1024$ & $1088$ & $1097$ & $1237$ & $1019$ \\ 
		Sm. LS & $1.5$ & $1369$ & $1107$ & $1029$ & $1032$ & $1088$ & $1097$ & $1229$ & $1022$ \\ 
		Sm. MLE & $0.2$ & $1053$ & $1128$ & $1078$ & $1074$ & $1008$ & $1007$ & $1044$ & $1097$ \\ 
		Sm. MLE & $0.5$ & $1022$ & $1167$ & $1086$ & $1078$ & $1000$ & $1000$ & $1035$ & $1103$ \\ 
		Sm. MLE & $1$ & $1014$ & $1184$ & $1093$ & $1084$ & $1005$ & $1002$ & $1074$ & $1107$ \\ 
		Sm. MLE & $1.5$ & $1000$ & $1209$ & $1097$ & $1091$ & $1005$ & $1003$ & $1071$ & $1109$ \\ 
		\hline \\[-1.8ex] 
		\multicolumn{2}{c}{Min.  value} & $ 5.43\cdot 10^{-3}$ & $2.1\cdot 10^{-3}$ & $4.57\cdot 10^{-2}$ & $9.25\cdot 10^{-2}$ & $4.35\cdot 10^{-2}$ & $4\cdot 10^{-2}$ & $6.12\cdot 10^{-3}$ & $7.44\cdot 10^{-2} $\\
		\hline\\[-1.8ex]
		\cline{3-10}
		& & $(3/4, 1/9)$ & $(3/4, 1/7)$ & $(3/4, 1/3)$ & $(3/4, 3/4)$ & $(3/2, 1/9)$ & $(3/2, 1/7)$ & $(3/2, 1/3)$ & $(3/2, 3/4)$\\
		\cline{3-10}\\[-1.8ex]
		\hline \\[-1.8ex] 
		IBF--BC & $0.2$ & $1075$ & $1062$ & $1083$ & $1316$ & $1074$ & $1074$ & $1080$ & $1140$ \\ 
		IBF--BC & $0.5$ & $1014$ & $1016$ & $1021$ & $1056$ & $1021$ & $1021$ & $1024$ & $1045$ \\ 
		IBF--BC & $1$ & $1003$ & $1003$ & $1004$ & $1006$ & $1005$ & $1005$ & $1005$ & $1009$ \\ 
		IBF--BC & $1.5$ & $1000$ & $1000$ & $1000$ & $1000$ & $1000$ & $1000$ & $1000$ & $1000$ \\ 
		LS & $0$ & $1075$ & $1077$ & $1104$ & $1382$ & $1086$ & $1087$ & $1094$ & $1150$ \\ 
		MLE & $0$ & $1042$ & $1043$ & $1059$ & $1208$ & $1063$ & $1064$ & $1069$ & $1112$ \\ 
		Sm. LS & $0.2$ & $1069$ & $1070$ & $1096$ & $1327$ & $1079$ & $1080$ & $1087$ & $1140$ \\ 
		Sm. LS & $0.5$ & $1064$ & $1066$ & $1091$ & $1308$ & $1077$ & $1076$ & $1084$ & $1137$ \\ 
		Sm. LS & $1$ & $1066$ & $1067$ & $1090$ & $1299$ & $1077$ & $1076$ & $1082$ & $1131$ \\ 
		Sm. LS & $1.5$ & $1067$ & $1069$ & $1092$ & $1311$ & $1077$ & $1077$ & $1083$ & $1132$ \\ 
		Sm. MLE & $0.2$ & $1012$ & $1012$ & $1017$ & $1061$ & $1028$ & $1028$ & $1031$ & $1051$ \\ 
		Sm. MLE & $0.5$ & $1009$ & $1010$ & $1014$ & $1064$ & $1026$ & $1025$ & $1029$ & $1049$ \\ 
		Sm. MLE & $1$ & $1012$ & $1012$ & $1015$ & $1074$ & $1028$ & $1027$ & $1029$ & $1047$ \\ 
		Sm. MLE & $1.5$ & $1014$ & $1014$ & $1017$ & $1086$ & $1029$ & $1028$ & $1031$ & $1049$ \\ 
		\hline \\[-1.8ex] 
		\multicolumn{2}{c}{Min.  value} & $ 1.07\cdot 10^{-1}$ & $1.07\cdot 10^{-1}$ & $8.88\cdot 10^{-2}$ & $1.45\cdot 10^{-2}$ & $1.77\cdot 10^{-1}$ & $1.78\cdot 10^{-1}$ & $1.76\cdot 10^{-1}$ & $1.28\cdot 10^{-1} $\\
		\hline\\[-1.8ex]
	\end{tabular}
\end{table}

The IBF--BC estimator for $\BS$ performs particularly well when $\gamma$ equals $\theta$, in which case the maximum competitor bid is a power distribution and the inverse bid function is linear. Even without undersmoothing, the asymptotic bias in the bidder's expected surplus would be zero when $\gamma = \theta = 1$ or $\gamma= \theta = 1/2$ and is fairly small at intermediate values. This fact helps explain why the RMSE for BS using LS and MLE relative to IBF--BC is larger, for example, in the $(1/3, 1/3)$ column compared to the columns on either side.

\Cref{tbl: BS RMSE T = 500} represents less than 6\% of the information contained in \cref{fig:medals BS MV}. Many more tables (available online \href{http://personal.psu.edu/kes380/files/shape_supp.pdf}{here}) provide further quantitative comparisons of the RMSE, as well as the bias of these and other estimators discussed below.

\subsection{Simulation results: distribution and quantile functions}
The next set of results compare the root mean integrated squared error (RMISE) of the value distribution and quantile function. Based upon the results in the first two columns, the IBF approach without boundary correction is (unsurprisingly) dominated by the boundary--corrected estimator. Next, comparing the second column with the columns to the right, we find that our estimators are again more robust to the choice of bandwidth and tend to outperform the boundary--corrected IBF approach when $\gamma$ and $\theta$ are small, i.e.~the highest competing bid is stronger. Although all bandwidth sequences are proportional to $T^{-1/5}$, the finite--sample behavior of our smoothed estimators is less (negatively) impacted by a small bandwidth. This finding is related to the fact that our estimator for $\alpha$ is consistent even when the bandwidth tends to zero.  Again, robustness is a virtue.

Comparing the rows for which either $\gamma$ or $\theta$ is small, we also find that our transformation method significantly reduces the RMISE in the estimate of the quantile function. In particular, the IBF--BC and the ``no transformation'' ($\psi_{1}$) estimators have greater RMISE compared to the estimators that employ the transformations $\psi_{3}$, $\psi_{4}$, or $\psi_{5}$. Looking across all columns, the smoothed least--squares estimator in conjunction with the transformations $\psi_{3}$ or $\psi_{5}$ appears to consistently perform best or near the best when the highest competing bid is relatively stronger  ($\gamma$ and $\theta$ are small), while the the IBF--BC estimator and the smoothed MLE estimator with $\psi_{5}$ or $\psi_{3}$ perform better in terms of RMISE when the highest competing bid is relatively weak.  We note that in auctions with three or more bidders one would expect the highest competing bid to be relatively strong and hence the smoothed least squares estimator to outperform the alternatives.

The differences between the estimators of the value distribution are less striking. This is due in part to the fact that some of the differences in the estimates of the quantile function are driven by the behavior near the left boundary. Using the delta method, one can show that the asymptotic distribution of the estimators for $F_{v}$ are scaled by the value density at $v$. Under our simulation design, $f_{v}$ approaches zero for small $v$. The relative differences in the estimators are dampened as a result. This also explains the fact that the log--transformation $\psi_{3}$ tends to be the best in terms of RMISE($F_{v}$) for small values of $\gamma$ and $\theta$. Under these parameters, $\alpha''$ diverges as $p$ approaches zero, which produces a large bias for small values of $p$ absent a transformation. The log--transformation ensures the asymptotic bias vanishes at the low end, while the fact that $f_{v}(v)$ is small mitigates the detrimental effects on the asymptotic variance. The transformation $\psi_{5}$ yields similar results, though it does not reduce the bias and increase the variance by as much.


\begin{figure}[htp]
	\includegraphics[angle=90, height = 0.95\textheight]{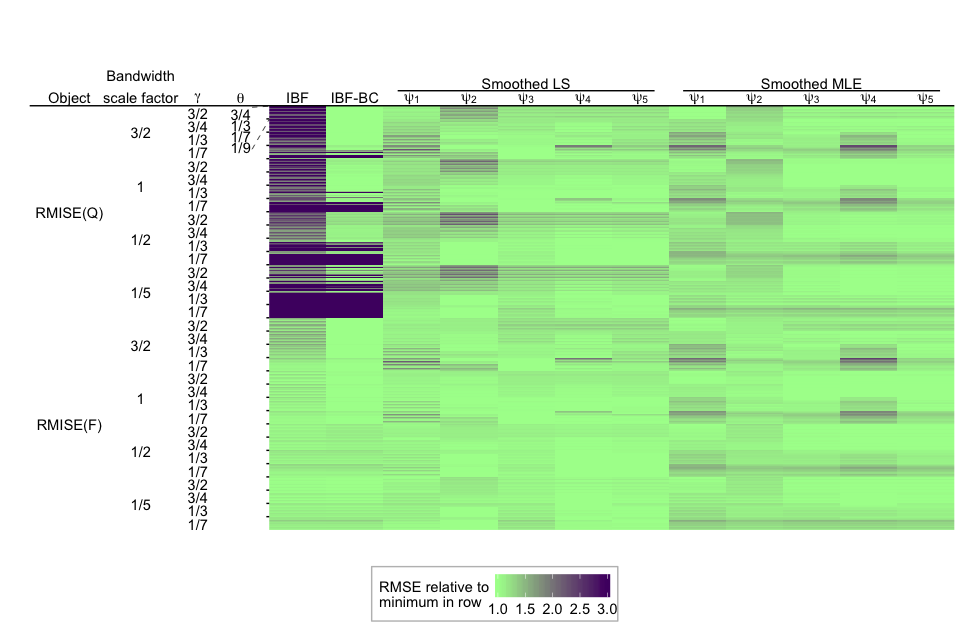}
	
		\caption{Relative performance of estimators of value distribution and quantile function. $T$ is the most frequently repeating parameter (in decreasing order) along the vertical axis. The first row reflects estimates of the mean valuation when the rule-of-thumb bandwidths are scaled by 3/2, ${\gamma}=3/2$, ${\theta}= 3/4$, and $T=500$.	\label{fig:medals F Q}
		}
\end{figure}

\subsection{Simulation results: density function}
We now compare the various methods of estimation of the density of bidder one's valuations. For the indirect methods, a bandwidth proportional to $T^{-1/5}$ is used in both the first and second steps. For the direct method, ${\alpha}'$ is estimated using a bandwidth proportional to $T^{-1/7}$. We compare the direct method using an estimate of $f_p$ as well as the true $f_p$, which the econometrician would know under the assumption the data were generated in a symmetric equilibrium.\footnote{When ${\gamma} = {\theta} = 1/3$, the data could be generated by a three-bidder auction with $F_{v}(v) = v^{3/2}$. These simulated data might also be generated in a two-bidder auction in which bidder one's competitor's valuations are distributed according to $F_{v_2}(v)=(4 v/ 5)^3$. When only the highest competitor bid is used, our approach does not depend on which of these models is correct except when we consider that $f_p$ is known a priori in the former case but not in the latter.}\ The boundary-corrected kernel density estimate $f_p$ is obtained from the sample of $p_t = G_{cT}(b_{1t})$ using a bandwidth proportional to $T^{-1/5}$. Note, however, that the true density $f_p$ is unbounded near $p=0$ in a symmetric equilibrium with more than two bidders, which may result in poor performance of the density estimate for small values of $p>0$. Thus, even though the pointwise rate of convergence of our estimate of $f_p$ is faster than the rate of convergence of our estimate of ${\alpha}'$, we would expect this estimator to perform poorly in finite samples. Indeed, the simulation results indicate that the direct method combined with the true $f_p$ compares favorably with the indirect estimates, but the direct method combined with an estimate of $f_p$ can be relatively poor. In such cases, better results might be achieved by estimating the density of an appropriate transformation of $p$ and using a change of variable formula to recover an estimate of $f_p$. Alternatively, the minimum relative entropy estimator in \cref{sec:Fp} could be used.
\begin{figure}[htp]
	\includegraphics[angle=90, height = \textheight]{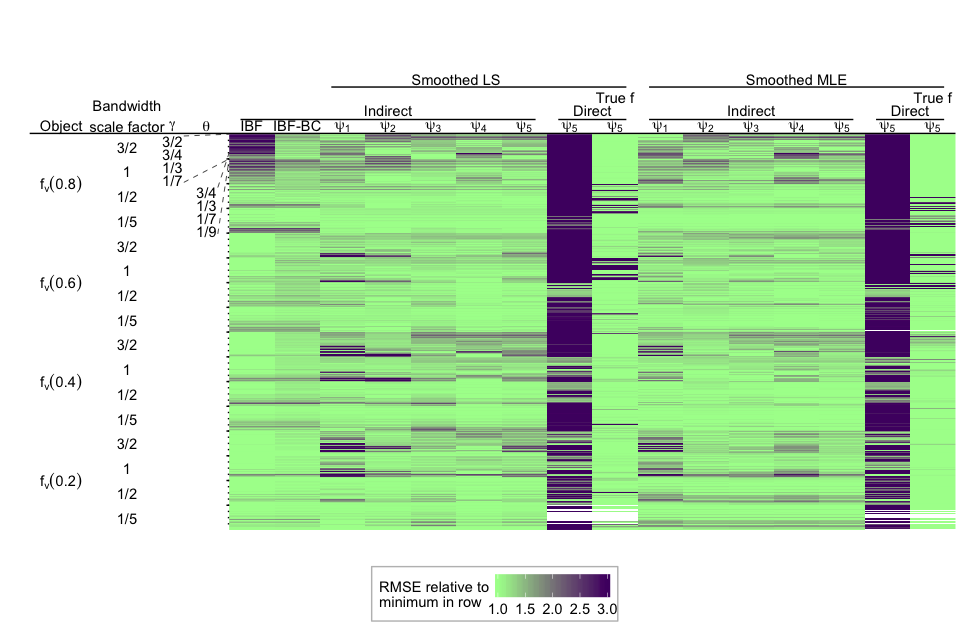}
	\caption{Relative performance of estimators of value density. $T$ is the most frequently repeating parameter (in decreasing order) along the vertical axis. The first row reflects estimates of the mean valuation when the rule-of-thumb bandwidths are scaled by 3/2, ${\gamma}=3/2$, ${\theta}= 3/4$, and $T=500$.	\label{fig:medals f}
	}
\end{figure}

\subsection{Simulation results: boundary correction methods}
\Cref{fig:medals alpha1} illustrates the simulation results for estimates of the inverse strategy function at the right boundary. The reflection--based boundary correction methods tend to perform better when the target of smoothing---$\alpha$ or $g_{c}$---is relatively flat and linear near the boundary. In this case, we would expect the error in the estimate of the auxiliary parameter $\hat d$ to be relatively small. On the other hand, the boundary kernel method tends to perform better when $\alpha'$ is large near the boundary, which would be more likely to happen if the number of bidders is small.

\begin{figure}[htp]
	\includegraphics[angle=90,height=0.95\textheight]{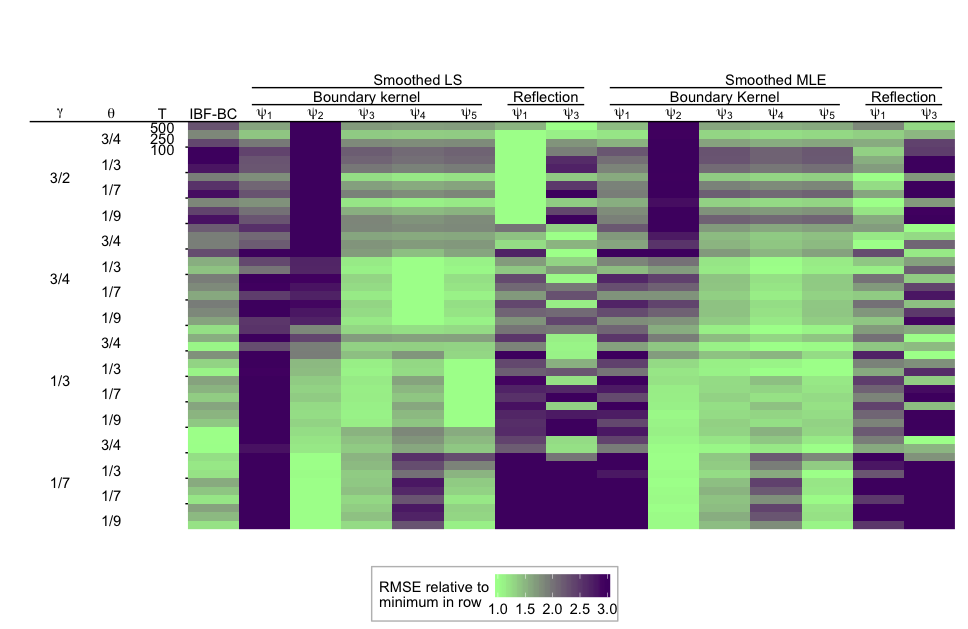}
	\caption{Relative performance of boundary correction methods.	  \label{fig:medals alpha1}}
\end{figure}

\drop{\grey{
Notice that rather than specify the distribution of each bidder's valuation, we structure our simulations similarly to our above analysis in the sense that we focus on the estimation of the parameters of interest for a fixed bidder, taking the distribution of the highest order statistic among the competitors' bids as given. In this way, we can easily vary the strength of bidder one's competition while holding the distribution of bidder one's valuations fixed. On the other hand, this approach obfuscates the sets of primitives that could generate these data. To clarify the issue, we observe that these data could be generated by a two-bidder auction in which bidder 2's valuations are distributed according to $F_{v,2}(v) = \left(\frac{1+\gamma_{c}}{1+\gamma_{1}}v\right)^{1/\gamma_{c}}$. At the same time, these data could also be generated by an $n$-bidder auction in which each of bidder one's competitors' valuations are distributed according to $F_{v,i}(v) = \left(\theta\,v\right)^{\frac{1}{(n-1)\gamma_{c}}}$ for $i = 2,\dots, n$, where $\theta = \frac{(1+\gamma_{c})((n-1)\gamma_{c} +(n-2)\gamma_{1})}{(1+\gamma_{1})(n-1)\gamma_{c}+ (n-2)\gamma_{1}}$. In particular, $\gamma_{c} = \gamma_{1}/(n-1)$ corresponds to a symmetric $n$-bidder auction.}}


\section{Conclusion}
\label{sec:conclusion}
This paper reformulates the empirical analysis of auction models as an isotonic estimation problem by treating the probability of winning as the choice variable in the bidders' decision problem. The nonparametric least--squares and nonparametric maximum likelihood estimators for a bidder's inverse strategy function are shown to converge at the optimal nonparametric rate. As a complementary set of results, we prove the asymptotic behavior of two boundary correction methods that can be combined with transformation to better control the bias--variance tradeoff in the kernel smoothed versions of our estimators. While these smoothing methods are important when estimating some objects of potential interest to the researcher, smoothing is not necessary for others. We prove that using our unsmoothed estimator as an input to a simple plug--in estimator of parameters such as the bidder's expected surplus achieves the semiparametric efficiency bound.

Though the results in this paper can guide several important methodological choices in empirical research on auctions, our theorems are silent regarding several extensions to the baseline model that have become standard in the empirical auction literature. Namely, we do not address the possibility of affiliation among the bidders' valuations, unobserved auction--level heterogeneity, or risk aversion. We leave these considerations for future work.

\bibliographystyle{apa}

\bibliography{shape.bib}

\appendix

\section{Proofs}
\label{app:proofs}

\subsection{GCM}
\label{app:gcm}

The proofs are arranged in the order in which their corresponding results are introduced in the text.

\begin{proof}[Proof of \cref{lem:lscharacterization}]
	Because $\alpha$ has to be nondecreasing, the optimizer must be constant between $(t-1)/T$ and $t/T$. The minimizer is right--continuous with possible discontinuities at $t/T$, because the objective is minimized if the ``jumps'' in $\alpha$ coincide with discontinuities in the derivative of $e_T$. In other words, the value of $\alpha$ should be smaller anywhere to the left of the discontinuity in order to minimize the first integral in \eqref{eq:ls}. Any jumps in $\alpha$ should be ``timed'' to take advantage of the negative contribution to the least--squares criterion that comes from the discontinuities in the derivative of $e_T$.
\end{proof}

\begin{proof}[Proof of \cref{thm:ebreve}]
We first establish the results for $\beT$, then uniform consistency of $\alpha_T$.  Convexity of $\beT$ follows by construction since $\alpha_T$ is restricted to be monotonic.  

Trivially extending the arguments in \citep[lemma 21.4 and the dicussion at the top of p308]{vandervaart2000asymptotic} about the empirical quantile process to $\hat e_T$, $\sqrt{T} \cparens{e_T\parens{\cdot}- e\parens{\cdot}}$ has the asserted limit process on $\parens{0,1}$. We now extend this to $\sparens{0,1}$.  Note that $e_T\parens{0}- e\parens{0}=0$ and 
$\sqrt{T}\parens{e_T\parens{1}-e\parens{1}}=\opone$, so we have convergence of finite marginals and tightness on $\parens{0,1}$.  We thus only need to extend tightness to $\sparens{0,1}$. 

We show the argument at one, where the argument at zero follows analogously.  Let $p_T=1-1/T$ and $\Delta_T\parens{p}=\sqrt{T}\parens{e_T\parens{p}-e\parens{p}}$.  Then for any sequence $\delta_T=o\parens{1}$ by the triangle inequality,
\begin{equation} \label{eq:ebreve1}
 \sup_{1>p>1-\delta_T} 
\abs[\big]{ \Delta_T\parens{1} - \Delta_T\parens{p}}	 
\leq
\sup_{1>p>1-\delta_T} \abs[\big]{\Delta_T\parens{p_T}-\Delta_T\parens{p}}
+
\abs[\big]{\Delta_T\parens{p_T} - \Delta_T\parens{1}}.
\end{equation}
The first right hand side term in \cref{eq:ebreve1} is \opone\ by tightness on $\parens{0,1}$.  The second right hand side term in \cref{eq:ebreve1} is \opone\ since the second highest order statistic converges at rate $T$.

 By \citet{carolan2001marginal}, the limit process is identical for the greatest convex minorant provided that $e$ is \emph{strictly} convex, which is implied by \cref{ass:monotone bid strategies} since $\alpha=e'$ is the inverse bid function composed with $Q_c$.

Finally, uniform convergence of $\alpha_T$.  Let $t_T=1/\sqrt[3]{t}$.  Then, by the monotonicity of $\alpha$ and $\alpha_T$,
\begin{multline}\notag
\max_{t_T \leq p \leq 1-t_T} \cparens{\alpha_T\parens{p}-\alpha\parens{p}} 
\leq\\
\max_{t_T \leq p \leq 1-t_T} \frac{ \beT\parens{p+t_T} -\beT\parens{p} - e\parens{p+t_T}+ e\parens{p}}{t_T}
+
\max_{p\in\sP} \parens[\Big]{\frac{e\parens{p+t_T}-e\parens{p}}{t_T} - \alpha\parens{p}}
\end{multline}
The first right hand side term is $O_p\parens{\sqrt{1/t_T T}}=O_p\parens{t_T}$.  The second right hand side term is bounded above by
$\alpha\parens{p+t_T}-\alpha\parens{p}=O_p\parens{t_T}$, also.  The minimum can be dealt with analogously.


\end{proof}

\newcommand{\tp}{\tilde p}%

\begin{just*}[of \cref{eq:cube root limit result,eq:cube root limit result simple}]\notag
We provide a sketch of the proof and a derivation of the limit distribution.   A full proof would be more careful, especially about issues pertaining to uniformity.  However, there is nothing special about the present scenario and a full rigorous proof would be lengthy but routine.

Our justification follows two steps.  In the first step, we derive a limit result for the inverse problem, i.e.\ the estimation of $\alpha^{-1}$.  In the second step we then apply equations (15) and (16) of \citet{jun2015classical} to obtain the limit distribution of ${\alpha}_T$ itself.

	We first establish asymptotics for the `inverse isotonic regression'--type estimator and then take its inverse to obtain asymptotics for ${\alpha}_T$.
Note that for ${\xi}={\alpha}\parens{p}$,
\[
 \argmin_{\tp} \cparens{ e\parens{\tp}- {\xi}\tp} = {\alpha}^{-1}\parens{{\xi}}=p.
\]
Let $\hat p$ be its sample equivalent, such that
\begin{multline*}
\sqrt[3]{T}\parens{\hat p -p} = \sqrt[3]{T}\cparens[\big]{\argmin_{\tp} \cparens{\beT\parens{\tp} - {\xi}\tp} - p} 
=
\sqrt[3]{T}\cparens[\big]{\argmin_{\tp} \cparens{\beT\parens{\tp} - \beT\parens{p} -  {\xi}\parens{\tp-p}} - p} 
=
\\
\sqrt[3]{T} \cparens[\Big]{\argmin_{\tp} \parens[\Big]{
\cparens{ \beT\parens{\tp} - \beT\parens{p} - e\parens{\tp} + e\parens{p}}
+
\cparens{ e\parens{\tp}- e\parens{p} - {\xi}\parens{\tp-p}}
}
-p}
\simeq \\
\sqrt[3]{T} \cparens[\Big]{\argmin_{\tp} \parens[\Big]{
        \cparens{ \beT\parens{\tp} - \beT\parens{p} - e\parens{\tp} + e\parens{p}}
        +
        {\alpha}'\parens{p} \parens{\tp-p}^2/2
    }
    -p} = 
\\
\argmin_t \parens[\Big]{
    T^{2/3}\cparens{ \beT\parens[\big]{p+t/\sqrt[3]{T}} - \beT\parens{p} - e\parens[\big]{p+t/\sqrt[3]{T}} + e\parens{p}}
    +
    {\alpha}'\parens{p} t^2/2
}\\
\convd
\argmin_t \parens[\big]{ \G^\circ\parens{t} + {\alpha}'\parens{p}t^2/2}
\sim
\argmax_t \parens[\big]{ \G^\circ\parens{t} - {\alpha}'\parens{p}t^2/2}.
\end{multline*}
where $\G^\circ$ is a Gaussian process with covariance kernel\footnote{This fact can be most easily seen by thinking in terms of a Bahadur representation.}
\begin{multline}
\lim_{T\to\infty}
\frac{1}{\sqrt[3]{T}}\cparens[\bigg]{
 H\parens[\Big]{Q_c\parens[\big]{p+t/\sqrt[3]{T}},Q_c\parens[\big]{p+s/\sqrt[3]{T}}}
 -
 H \parens[\Big]{Q_c\parens[\big]{p+t/\sqrt[3]{T}},Q_c\parens{p}}
 \\
 -
H \parens[\Big]{Q_c\parens[\big]{p},Q_c\parens[\big]{p+s/\sqrt[3]{T}}}
+
H\parens[\Big]{Q_c\parens{p},Q_c\parens{p}}
},
\end{multline}
which under \cref{eq:Hstar simple} simplifies to 
\(
{\zeta}^2\parens{p} \abs{  \Med\parens{s,t,0}}:
\)
in other words, $\G^\circ$ is then ${\zeta}\parens{p}$ times a standard two--sided Brownian motion $\G^B$, such that then by a change of variables,
\[
\argmax_t \parens[\big]{ \G^\circ\parens{t} - {\alpha}'\parens{p}t^2/2} \sim
\cparens[\big]{2 \zeta\parens{p}/\alpha'\parens{p}}^{2/3}\,\argmax_{t} \parens[\big]{\G^{B}\parens{t}-t^{2}}
\]
From equations (15) and (16) in \citet{jun2015classical} it then follows that
\[
 \sqrt[3]{T} \cparens{ {\alpha}_T\parens{p}-{\alpha}\parens{p}} \convd 
 \sqrt[3]{4\zeta^{2}\parens{p}\alpha'\parens{p}}\argmax_{t} \parens[\big]{ \G^B\parens{t} - t^2}. \qedhere
\]
\end{just*}

\subsection{NPMLE}
\label{app:npmle}

\begin{lem}\label{lem:pseudoconcave}
	The derivative of $\sL_{j}(\alpha) = \sum_{s=t_{j}}^{t_{j+1}-1} \cparens[\big]{(s-2) \log\parens{\alpha-b_{(s)}} - (s-1) \log\parens{\alpha -b_{\parens{s-1}}}}$ with respect to $\alpha$ is zero exactly once on $(b_{t_{j+1}-1}, \infty)$ and crosses zero from above.
	\proof
	
	Multiplying the stationarity condition $\sL_{j}'(\alpha) = 0$ by $\alpha - b_{\parens{t_{j+1}-1}}$ and collecting terms, we find that $\alpha$ is a stationary point if and only if
	\[
	(t_{j}-1)y_{t_{j}-1}(\alpha) + 2 \sum_{s=t_{j}}^{t_{j+1}-2} y_{s}(\alpha) = t_{j+1}-3,\quad\quad \text{where } 
	y_{s}(\alpha) = \frac{\alpha - b_{\parens{t_{j+1}-1}}}{\alpha-b_{\parens{s}}}.
	\]
	For all $s$, $y_{s}\parens{b_{\parens{t_{j+1}-1}}} = 0$ and $y_{s}(\alpha)$ is continuous and increasing in $\alpha$. 
	The left side of the equation is equal to zero at $b_{\parens{t_{j+1}-1}}$ and approaches $t_{j} - 1+2(t_{j+1}-1-t_{j}) = 2 t_{j+1} - t_{j} - 3 > t_{j+1} - 3$ as $\alpha$ increases. There exists an $\alpha$ that solves the equation above by the intermediate value theorem, and the solution is unique because the left side is strictly monotonic in $\alpha$.
	
	Finally, $\sL_{j}'$ crosses zero from above because $\sL_{j}'$ diverges to positive infinity as $\alpha$ approaches $b_{{\parens{t_{j+1}-1}}}$. \qed
\end{lem}

\begin{lem}\label{lem:mle is an average}
	If $\tilde {\alpha}_{(t_{j})}^{(k-1)}$ and $\tilde {\alpha}_{(t_{j+1})}^{(k-1)}$ are the values of $\tilde {\alpha}$ in the two blocks that are pooled together in the $k$-th step, then the new value is $\tilde {\alpha}_{(t_{j})}^{(k)}$ between $\tilde {\alpha}_{(t_{j})}^{(k-1)}$ and $\tilde {\alpha}_{(t_{j+1})}^{(k-1)}$.
	\proof
	Without loss of generality, assume $\tilde {\alpha}_{(t_{j+1})}^{(k-1)}<\tilde {\alpha}_{(t_{j})}^{(k-1)}$.The zero of $\sL_{j}' + \sL_{j+1}'$ must be greater than $\tilde {\alpha}_{(t_{j+1})}^{(k-1)}$, because both $\sL_{j}'$ and  $\sL_{j+1}'$ are positive to the left of $\tilde {\alpha}_{(t_{j+1})}^{(k-1)}$ by \cref{lem:pseudoconcave}. On the other hand, the zero of $\sL_{j}'+\sL_{j+1}'$ must be less than $\tilde {\alpha}_{(t_{j})}^{(k-1)}$ because both derivatives are negative to the right of $\tilde {\alpha}_{(t_{j})}^{(k-1)}$ by \cref{lem:pseudoconcave}. \qed
\end{lem}

\begin{lem}\label{lem:dual active set}
	PAVA for the NPMLE is a dual active set method, i.e.~PAVA satisfies stationarity, complementary slackness, and dual feasibility at every step of the algorithm, but does not satisfy primal feasibility until the final iterate.
	\proof
	The PAVA algorithm clearly satisfies the stationarity and complementary slackness conditions at every step, and satisfies primal feasibility at the last step (primal feasibility is the stopping criterion). It remains to show that the Lagrange multipliers are nonnegative at every step. 
	
	Let 
	\[
	\tilde {\lambda}_t^{(k)} = \sum_{s=3}^{t-1} \pd{\sL(\tilde {\alpha})}{\tilde {\alpha}_{\parens{s}}}= \sum_{s=3}^{t-1} \parens[\bigg]{\frac{s-2}{\tilde {\alpha}_{\parens{s}}^{(k)}-b_{\parens{s}}}-\frac{s-1}{\tilde {\alpha}_{\parens{s}}^{(k)}-b_{\parens{s-1}}}}
	\]
	denote the Lagrange multipliers implied by the stationarity conditions, where the superscripts $(k)$ indicate the value of the variable after the $k$-th step of the algorithm. Initially, $\tilde {\lambda}_t^{(0)}=0$ for all $t$. 
	
	We will proceed by induction on $k$. Let $t_{j}$ and $t_{j+1}$ be the starting points of the adjacent blocks pooled together in the $k$-th step for some $k>0$. Assume $\tilde {\lambda}_t^{(j)}\geq0$ for all $t$ and $j < k$.
	
	Suppose by way of contradiction that a negative Lagrange multiplier is introduced in the $k$-th step. The negative multiplier must apply to one of the active constraints in the two most recently merged blocks of constraints, because the Lagrange multipliers on constraints outside these two blocks are unaffected: the multipliers are all zero for the slack constraints on the singleton blocks to the right, and clearly $\tilde \lambda_{t}^{(k)}$ is unaffected for all $t\leq t_{j}$.
	
	A negative multiplier on one of the constraints in the most recently merged blocks implies that there exists a constraint within this chain of equalities that can be slackened and increase the loglikelihood. We will show that this leads to a contradiction because slackening any one of the constraints and moving in the direction that would increase the loglikelihood will necessarily violate primal feasibility.
	
	Suppose we slacken the constraint $\tilde \alpha_{(s)} \geq \tilde\alpha_{(s-1)}$. There are two cases to consider. First, suppose the slackened constraint belongs to the left pre-merged block, i.e.~$s$ is such that $t_{j}\leq s < t_{j+1}$, and let $\tilde {\alpha}_{\parens{t_j}}^{\prime(k)}$ denote the new solution to the stationarity condition in the sub-block to the left of $s$. Let $\tilde {\alpha}_{\parens{s}}^{\prime(k-1)}$ denote the solution for $\tilde {\alpha}$ in the block beginning with $s$ and ending $t_{j+1}-1$. Then $\tilde \alpha_{\parens{t_j}}^{\prime(k)}$ must be greater than the value of $\tilde {\alpha}_{\parens{s}}^{\prime (k-1)}$ in the right sub-block, otherwise relaxing this constraint would have been feasible and improved the loglikelihood in an earlier iterate, thereby contradicting our assumption that the $k$-th step is the first that introduces a negative Lagrange multiplier. In addition, $\tilde {\alpha}_{\parens{t_j}}^{\prime(k)}> \tilde {\alpha}_{(t_j)}^{(k-1)}> \tilde {\alpha}_{\parens{s}}^{\prime(k-1)}$ by \cref{lem:mle is an average}.  Finally, we invoke \cref{lem:mle is an average} again to conclude that the value of $\tilde {\alpha}_{\parens{t_{j+1}}}^{\prime(k)}$ is less than $\tilde {\alpha}_{\parens{t_j}}^{\prime (k)}$. Hence, none of the constraints in the left block can be removed while maintaining primal feasibility.
	
	On the other hand, we may suppose the objective would by improved by making one of the constraints in the right block slack. By a similar argument, this too would violate primal feasibility. Therefore, none of the constraints in the merged block can be removed without violating primal feasibility or decreasing the objective. Therefore, none of the Lagrange multipliers are negative after the $k$-th iterate. By induction, there are no negative Lagrange multipliers in any step of the algorithm. \qed
\end{lem}

\begin{defn}[Invex function]\label{def:invex}
	A function $f: S \subseteq \Re^n \to \Re$ is invex at $u\in S$ if there exists a $\Re^n$-valued function ${\eta}$ such that $f(x) - f(u) \geq {\eta}(x) \cdot \nabla f(u)$ for all $x\in S$, where $\nabla f$ denotes the gradient vector of $f$. Such a function ${\eta}$ is known as an invexity kernel.
	
	A function $g: S \subseteq \Re^n \to \Re$ is type I invex at $u\in S$ if there exists a $\Re^n$-valued function ${\eta}$ such that $ - g(u) \geq {\eta}(x) \cdot \nabla g(u)$ for all $x\in S$
\end{defn}

\begin{thm}[\cite{Hanson1999} theorem 2.1]
	Consider the problem $\min_{x\in S} f(x)$ subject to $g(x) \leq 0$ for some functions $f: \Re^n\to \Re$ and $g:\Re^n \to \Re^m$ that are differentiable on $S$.  For $u\in S$ to be optimal, it is sufficient that the KKT conditions are satisfied at $u$ and $f$ and $g$ satisfy
	\begin{align*}
		f(x) - f(u) &\geq {\eta}(x) \cdot \nabla f(u)\\
		-g_t(u) & \geq {\eta}(x) \cdot \nabla g_t(u)
	\end{align*}
	for every active component $g_t$ of $g$ for some common invexity kernel ${\eta}:\Re^n \to \Re^n$.
\end{thm}

\begin{lem} \label{lem:invexity}
	The KKT conditions \eqref{eq:npmle Lagrange foc} are necessary and sufficient for the isotonic maximum likelihood problem \eqref{eq:isonpmle}.
	\proof
	The proof proceeds as follows. First, we establish that $-\sL$ and the isotonicity constraints are invex functions, a generalization of convex functions (see \cref{def:invex}) which may be equivalently characterized as the collection of differentiable functions for which every stationary point is a global minimum. Second, we show that these functions are invex with respect to the same invexity kernel using Gale's theorem of the alternative, as suggested by \cite{Hanson1981}. Finally, we conclude that the KKT conditions for the constrained minimization of $-\sL$ are sufficient by a direct application of theorem 2.1 of \cite{Hanson1999}.
	
	First, $-\sL$ is invex because every stationary point is a global minimum. To see this, we note that $-\sL$ is convex at its (unique) stationary point because
	\begin{align*}
		-\pdd{\sL}{{\alpha}_{\parens{t}}} &= \frac{t-2}{({\alpha}_{\parens{t}}-b_{\parens{t}})^2} - \frac{t-1}{({\alpha}_{\parens{t}} - b_{\parens{t-1}})^2}\\
		&=  \frac{t-1}{({\alpha}_{\parens{t}}-b_{\parens{t-1}})({\alpha}_{\parens{t}}-b_{\parens{t}})} -  \frac{t-1}{({\alpha}_{\parens{t}} - b_{\parens{t-1}})^2} > 0 \,,
	\end{align*}
	where the second equality follows by substitution using the first-order condition, and the inequality follows from $b_{\parens{t}}>b_{\parens{t-1}}$. Thus, there is a unique local minimum. Finally, there are no minima at the boundaries because $-\pd{\sL}{{\alpha}_{\parens{t}}}$ is eventually positive as ${\alpha}_{\parens{t}}$ tends to infinity and $-\sL$ diverges to infinity as $\alpha_{\parens{t}}$ approaches $b_{\parens{t}}$. Thus, every stationary point is a global minimum. Hence, $-\sL$ is invex. Then, by definition of invexity, there exists a vector-valued invexity kernel ${\eta}$ such that $L(\balTml) - L(\tilde{{\alpha}}) \geq {\eta}\parens{\tilde {\alpha}}\cdot\nabla L(\tilde {\alpha})$ for all $\tilde{{\alpha}}$. The constraints on $\tilde{{\alpha}}_{\parens{t}}$ are linear and therefore invex, as well.
	
	Second, we must further show that there exists a \emph{common} invexity kernel with respect to which $-\sL$ and the (active) constraint functions $c_t(\tilde {\alpha}) = -\tilde {\alpha}_{\parens{t}} + \tilde {\alpha}_{\parens{t-1}}$ are (type I) invex at the solution. The existence of such an invexity kernel is implied by the existence of a solution to the linear system $A {\eta}(\tilde {\alpha}) \leq C(\tilde {\alpha})$ for all $\tilde {\alpha}$, where $A$ is the $\Re^{T-2}\times\Re^{T-2}$ Jacobian matrix $(-\nabla \sL; \nabla c_4; \dots; \nabla c_T)$ evaluated at $\balTml$ and $C(\tilde {\alpha})=(\sL(\balTml)-\sL(\tilde {\alpha}),-c_4(\tilde{{\alpha}}), \dots, -c_T(\tilde {\alpha}))$. By Gale's theorem of the alternative, a solution to this system of inequalities exists if and only if there does not exist a vector $y\geq 0$ such that $y'A = 0$ and $C(\tilde \alpha)'y = -1$. Because $A$ is of the form
	\[
	\begin{pmatrix}
	-\pd{\sL}{\breve{\alpha}_{\parens{3}}^{\text{MLE}}} & -\pd{\sL}{\breve{\alpha}_{\parens{4}}^{\text{MLE}}} & -\pd{\sL}{\breve{\alpha}_{\parens{5}}^{\text{MLE}}} & -\pd{\sL}{\breve{\alpha}_{\parens{6}}^{\text{MLE}}} & \cdots\\
	1 & -1 & 0 & 0 & \cdots\\
	0 & 1 & -1 & 0 & \\
	0 & 0 & 1 & -1 & \\
	\vdots & & &\ddots &  \ddots
	\end{pmatrix}
	\]
	and because the stationarity condition implies $\lambda_{t}=\sum_{s=3}^{t-1}-\pd{\sL}{\breve{\alpha}_{\parens{s}}^{\text{MLE}}}$, we can show that $y'A=0$ has a solution only if $y$ is a scalar multiple of $(1,\lambda_4,\dots, \lambda_T)$. But then $C(\tilde \alpha)'y = -1$ does not have a solution for any $\tilde \alpha$ because y and $C(\tilde {\alpha})$ are nonnegative vectors for all $\tilde {\alpha}$. Thus, the objective and active constraints are (type I) invex with respect to some shared invexity kernel.

	Finally, any solution to the KKT conditions is a global minimum by theorem 2.1 in \cite{Hanson1999}. \qed
\end{lem}

\begin{proof}[Proof of \cref{thm:PAVA converges to MLE}]
	The final iterate of the PAVA algorithm satisfies the KKT conditions by \cref{lem:dual active set}, which are sufficient for the global maximum by \cref{lem:invexity}.
\end{proof}

\begin{just*}[of \cref{eq:npmle limit distribution}]
We provide a sketch of the proof and a derivation of the limit distribution.   A full proof would be more careful, especially about issues pertaining to uniformity.  However, there is nothing special about the present scenario and a full rigorous proof would be lengthy but routine.

We remind the reader that ${\alpha}_T^\mle\parens{p}= {\alpha}_T^\mle\cparens{G_c\parens{b}}$ can be characterized in terms of the inverse of the solution to an `inverse regression' problem, namely to find the solution ${\beta}_T\parens{v}$ that minimizes
\[
\S_T\parens{b,v} = \frac1T\sum_{t=2}^T \parens[\bigg]{ \frac{t-2}{v-b_{(t)}} - \frac{t-1}{v-b_{(t-1)}}}
\one\parens{b_{(t)} \leq b},
\]
over a region of $b$'s for which $v-b$ is bounded away from zero.  Note that ${\beta}_T$ is an estimate of the bid function ${\beta}\parens{v}$. We first obtain the limit distribution of $\sqrt[3]{T}\cparens{ {\beta}_T\parens{v}- {\beta}\parens{v}}$.  We then invoke the results of \citet{jun2015classical} to obtain properties of the inverse.

We first obtain an approximation of the form
\(
\S_T\parens{b,v} \simeq  \S_T^*\parens{b,v} + \S_T^\circ\parens{b} + \S\parens{b,v},
\)
for functions $\S_T^*,\S_T^\circ,\S$ introduced below, where $\simeq$ means that the omitted terms are negligible.  Taking $\sqrt[3]{T}$--consistency as given, we then argue that
\(
T^{2/3} \sparens[\big]{\S_T\cparens{{\beta}\parens{v}+x/\cbT,v} - \S_T\cparens{{\beta}\parens{v},v}}
\)
converges to a limiting Gaussian process plus a quadratic in $x$, whose minimizer as a function of $x$ has a (scaled) Chernoff distribution.  Applying equations (15) and (16) in \citet{jun2015classical} then yield the stated limit distribution.

Noting that $\max_t \abs{b_{(t)}-Q_{ct}} = O_p\parens{T^{-1/2}}$ for $Q_{ct}=Q_c\parens{t/T}$, we have (uniformly in $b,v$),
\begin{multline*}
\S_T\parens{b,v} = \frac1T\sum_{t=2}^T \parens[\bigg]{ \frac{t-2}{v-b_{(t)}} - \frac{t-1}{v-b_{(t-1)}}}
\one\parens{b_{(t)} \leq b}
=\\
\frac1T \sum_{t=2}^T \frac{t}{v-b_{(t)}} \one\parens{b_{(t)}\leq b< b_{(t+1)}}
- \frac1T \sum_{t=2}^T \frac{2}{v-b_{(t)}} \one\parens{b_{(t)}\leq b}
\\
=
\underbrace{\frac1T \sum_{t=2}^T \frac{t}{v-Q_{ct}} \one\parens{b_{(t)}\leq b<b_{(t+1)}}}_{\mathrm{I}}
- 
\underbrace{\frac1T \sum_{t=2}^T \frac{2}{v-Q_{ct}} \one\parens{b_{(t)}\leq b }}_{\mathrm{II}}
\\
+
\underbrace{\sum_{t=2}^T \frac tT \frac{b_{(t)}-Q_{ct}}{\parens{v-b_{(t)}}^2} \one\parens{b_{(t)}\leq b< b_{(t+1)}}}_{\mathrm{III}}
-
\underbrace{\frac2T \sum_{t=2}^T \frac{b_{(t)}-Q_{ct}}{\parens{v-b_{(t)}}^2}\one\parens{b_{(t)}\leq b}}_{\mathrm{IV}} +O_p\parens{T^{-1}}.
\end{multline*}
Now, term I is
\<
\frac1T \sum_{t=2}^T \frac{t}{v-Q_{ct}} \one\parens[\Big]{ \frac tT \leq G_{cT}\parens{b} < \frac{t+1}T }
\simeq
\frac{G_{cT}\parens{b}}{v-Q_c\cparens{G_T\parens{b}}} 
\simeq
2\frac{G_{cT}\parens{b}-G_c\parens{b}}{v-b} + \frac{G_c\parens{b}}{v-b}, \label{eq:npmle I}
\>
where $\simeq$ means that asymptotically negligible terms were omitted.  Further, term II is
\begin{multline} \label{eq:npmle II}
\simeq \int_0^{G_{cT}\parens{b}} \frac2{v-Q_c\parens{p}} \dif p
\simeq
2\int_0^b \frac{g_c\parens{\tilde b}}{v-\tilde b} \dif \tilde b
+
2 \frac{G_{cT}\parens{b}-G_c\parens{b}}{v-b}
\simeq\\
2\frac{G_c\parens{b}}{v-b} - 2\int_0^b \frac{G_c\parens{\tilde b}}{\parens{v-\tilde b}^2} \dif \tilde b
+
2 \frac{G_{cT}\parens{b}-G_c\parens{b}}{v-b}
\end{multline}
Term III is
\< \label{eq:npmple III}
\simeq \sum_{t=2}^T \frac tT   \frac{b_{(t)}-Q_{ct}}{\parens{v - Q_{ct}}^2} \one\parens{Q_{ct}\leq b< Q_{c,t+1}}
\simeq
G_c\parens{b} \frac{Q_{cT}\cparens{G_c\parens{b}} - b}{\parens{v- b}^2}
\simeq
-  \frac{G_{cT}\parens{b} - G_c\parens{b}}{\parens{v- b}}.
\>
Finally, term IV is
\< \label{eq:npmple IV}
\simeq \frac2T \sum_{t=2}^T c_t\parens{v}\parens{b_{(t)}-Q_{ct}} \one\parens{Q_{ct}\leq b}
\simeq 2 \int_0^{G_c\parens{b}} \frac{Q_{cT}\parens{p}-Q_c\parens{p}}{\cparens{v-Q_c\parens{p}}^2} \dif p
\simeq
-2 \int_0^b \frac{G_{cT}\parens{\tilde b}-G_c\parens{\tilde b}}{\parens{v-\tilde b}^2} \dif \tilde b.
\>
Adding the right hand sides in \cref{eq:npmle I,eq:npmple III} and subtracting the right hand sides in \cref{eq:npmle II,eq:npmple IV} from the sum yields after integration by parts on one of the nonstochastic terms
\<
\S_T\parens{b,v}\simeq - \underbrace{\frac{G_{cT}\parens{b}-G_c\parens{b}}{v-b}}_{\S_T^*\parens{b,v}} + \underbrace{2 \int_0^b \frac{G_{cT}\parens{s}-G_c\parens{s}}{\parens{v -s}^2} \dif s}_{\S_T^\circ\parens{b,v}}
+
\underbrace{\int_0^b \frac{G_c\parens{s}-g_c\parens{s} \parens{v-s}  }{\parens{v-s}^2} \dif s}_{\S\parens{b,v}}.
\>
Now, by \cref{ass:Gaussian process},
\[
\sqrt{T}\parens{\S_T^*+\S_T^\circ} \convw \underbrace{2\int_0^b \frac{\G^*\parens{s}}{\parens{v-s}^2} \dif s}_{\S_1^R} - \underbrace{\frac{\G^*\parens{b}}{v-b}}_{\S_2^R},
\]
as a process indexed by $\parens{b,v}$.  Thus,
\[
\maligned{
T^{2/3}\sparens[\big]{\S_1^R\cparens[\big]{{\beta}\parens{v}+t/\cbT,v}-\S_1^R\cparens{{\beta}\parens{v},v}} &= \opone, \\
T^{2/3} \cparens[\big]{\S_2^R\parens[\big]{{\beta}\parens{v}+t/\cbT,v}-\S_2^R\cparens{{\beta}\parens{v},v}} &\convw \sqrt{g_c\parens{b}} \cparens{v -{\beta}\parens{v}}^{-1} \G^B,
}
\]
where $\G^B$ is a standard two--sided Brownian motion.  

Further,
\(
T^{2/3} \cparens[\big]{\S\parens{b+t/\cbT,v}-\S\parens{b,v}} = \S''\parens{b,v} t^2/2 + o\parens{1},
\)
where the derivatives are taken with respect to $b$.
Putting everything together suggests that under \cref{eq:Hstar simple} for $b={\beta}\parens{v}$,
\[
\cbT\cparens{ {\beta}_T\parens{v} - {\beta}\parens{v}} \convw \argmin_x \parens[\bigg]{
	\frac{\sqrt{g_c\parens{b}}}{v -b} \G^B\parens{x}
	+ 
	\frac{S''\parens{b}}2 x^2}
=
\parens[\bigg]{\frac{4g_c\parens{b}}{\parens{v-b}^2 \cparens{\S''\parens{b}}^2}}^{1/3} \C,
\]
where $\C$ is a standard Chernoff.   Note that when $\S''$ is evaluated at ${\beta}\parens{v}$, we get
\(
\S''\parens{b} = \cparens{ 2 g_c^2\parens{b} -g_c'\parens{b} G_c\parens{b} } g_c\parens{b} / G_c^2\parens{b}.
\)
By equations (15) and (16) of \citet{jun2015classical} we have that for $b=Q_c\parens{p}$,
\begin{multline*}
\cbT\cparens{{\alpha}_T^\mle\parens{p}-{\alpha}\parens{p}} 
\convd
\parens[\bigg]{2- \frac{G_c\parens{b}g_c'\parens{b}}{g_c^2\parens{b}}}
\parens[\bigg]{\frac{4g_c\parens{b}}{\parens{{\alpha}\parens{p}-b}^2 \cparens{S''\parens{b}}^2}}^{1/3} \C
=
\\
\cparens[\bigg]{ \frac{4G_c^2\parens{b}}{g_c^3\parens{b}} \parens[\bigg]{2- \frac{G_c\parens{b}g_c'\parens{b}}{g_c^2\parens{b}}}}^{1/3} \C
=
\sqrt[3]{4 {\zeta}^2\parens{p} \cparens[\big]{2Q_c'\parens{p}+Q_c''\parens{p}p}} \C,
\end{multline*}
as claimed. \qed
\end{just*}

\subsection{Smoothing}

\begin{proof}[Proof of \cref{thm:ehat}]
First convexity. Substitution of $t=\parens{s-p}/h$ yields
\[
  \ehat\parens{p} = \int_{-\infty}^\infty \beT\parens{p+sh} k\parens{s} \dif s.
\]
Thus, for any $0<\lambda<1$ and any $p_\ell < p_h$,
\begin{multline*}
 \ehat\cparens[\big]{\lambda p_\ell +\parens{1-\lambda} p_h} =
\int_{-\infty}^\infty 
\beT\cparens[\big]{\lambda p_\ell+\parens{1-\lambda} p_h + sh} 
k\parens{s} \dif s \leq
\\
 \int_{-\infty}^\infty \cparens[\big]{
\lambda \beT\parens{p_\ell+ sh}
+
\parens{1-\lambda} \beT\parens{p_h+sh}
} 
k\parens{s} \dif s 
=
\lambda \ehat\parens{p_\ell} + \parens{1-\lambda} \ehat\parens{p_h}.
\end{multline*}

Now convergence.  We have
\begin{multline*}
\sqrt{T} \cparens[\big]{ \ehat\parens{p} - e\parens{p} }
=
 \sqrt{T} \parens[\bigg]{ \lazyint{\beT\parens{s} k_h\parens{p-s}} - e\parens{p}} 
=\\
\sqrt{T} \parens[\bigg]{
	 \lazyint{ \ehat\parens{p+sh} k\parens{s}} - e\parens{p}}
=\\
\underbrace{\sqrt{T} \lazyint{
		\cparens[\big]{
			\beT\parens{p+sh} - \beT\parens{p}-e\parens{p+sh}+e\parens{p}	
		} k\parens{s}} }_{\mathrm{I}}
\\
+\underbrace{\sqrt{T} \cparens[\big]{ \beT\parens{p} -e\parens{p}}}_{\mathrm{II}} 
+\underbrace{ \sqrt{T} \lazyint{  \cparens[\big]{e\parens{p+sh} - e\parens{p}} k\parens{s}}}_{\mathrm{III}}
\end{multline*}
Term II is what we want left over.  Term III is $\oo$ by a standard kernel bias expansion and \cref{ass:Qc differentiable}.  Finally, term I is $\opone$ by \cref{thm:ebreve}.

We end by establishing convergence of $\ahT$.  Note that
\begin{multline*}
\sqrt{Th} \cparens[\big]{ \ahT\parens{p} - \alpha\parens{p} }
=
- \sqrt{Th} \parens[\bigg]{ \lazyint{\beT\parens{s} k_h'\parens{s-p}} + \alpha\parens{p}} 
=\\
-\sqrt{Th} \parens[\bigg]{
\frac{1}{h} \lazyint{ \beT\parens{p+sh} k'\parens{s}} + \alpha\parens{p}}
=\\
\underbrace{-\sqrt{\frac{T}{h}} \lazyint{
\cparens[\big]{
\beT\parens{p+sh} - \beT\parens{p}-e\parens{p+sh}+e\parens{p}	
} k'\parens{s}} }_{\mathrm{I}}
\\
\underbrace{-\sqrt{\frac{T}{h}} \cparens[\big]{ \beT\parens{p} -e\parens{p}}
\lazyint{k'\parens{s}}}_{\mathrm{II}} 
\\
\underbrace{- \sqrt{\frac{T}{h}} \lazyint{  \cparens[\big]{e\parens{p+sh} k'\parens{s} \dif s - \sqrt{Th} \alpha\parens{p}}}}_{\mathrm{III}}
\end{multline*}
Term II is zero by the assumptions on the kernel.  Further, using a standard kernel derivative bias expansion, term III is $e'''\parens{p}\Xi/2 + o\parens{1}$.	
Finally, term I has by the continuous mapping theorem the same asymptotic distribution as
\[
 -\frac{1}{\sqrt{h}} \lazyint{ \cparens[\big]{\G\parens{p+sh} - \G\parens{p}} k'\parens{s} },
\]
which converges to a normal distribution with variance
\[
\lim_{h\to 0} \int_{-\infty}^\infty \int_{-\infty}^\infty 
 \frac{1}{h}\cparens[\big]{
 H\parens{p+sh,p+\tilde s h}
 - H\parens{p+sh,p}
 - H\parens{p,p+\tilde s h}
 + H\parens{p,p} 
} k'\parens{s} k'\parens{\tilde s}\dif \tilde s \dif s,	 
\]
which equals $\sV\parens{p}$, as asserted.
\end{proof}

\begin{proof}[Proof of \cref{lem:V}]
Suppose first that $\tilde s \geq s \geq 0$. Then,
\[
\maligned{
H\parens{p+sh,p+\tilde sh} & = {\zeta}\parens{p+sh}{\zeta}\parens{p+\tilde sh} \cparens{ p\parens{1-p} + \parens{1-p}sh - p \tilde sh - s\tilde s h^2}, 
\\
H\parens{p+sh,p} & = {\zeta}\parens{p+sh}{\zeta}\parens{p} \cparens{ p\parens{1-p}-psh}, 
\\
H\parens{p,p+\tilde sh} & = {\zeta}\parens{p+\tilde sh}{\zeta}\parens{p} \cparens{ p\parens{1-p}-p\tilde sh}, 
\\
H\parens{p,p} & = {\zeta}^2\parens{p} p\parens{1-p}.
}
\]
Thus,
\(
\lim_{h\to 0} \sH_h\parens{p,s,\tilde s} = {\zeta}^2\parens{p} s.
\)
Repeating the argument for the other cases yields 
\(
\lim_{h\to 0}  \sH_h\parens{p,s,\tilde s}
=
\zeta^2\parens{p}\abs{\Med\parens{0,s,\tilde s}}.
\)
Now,
\begin{multline*}
\int_{-\infty}^\infty \int_{-\infty}^\infty \abs{\Med\parens{0,s,\tilde s}}
k'\parens{s} k'\parens{\tilde s} \dif \tilde s \dif s =
\\
 -\int_{-\infty}^0 \int_{-\infty}^0 \max\parens{\tilde s,s} k'\parens{s} k'\parens{\tilde s} \dif \tilde s \dif s
 +
 \int_0^\infty \int_0^\infty \min\parens{\tilde s,s} k'\parens{s} k'\parens{\tilde s} \dif \tilde s \dif s
 =
 \\
 -2  \int_{-\infty}^0 k'\parens{s} s \int_{-\infty}^s  k'\parens{\tilde s} \dif \tilde s \dif s
 +
2 \int_0^\infty k'\parens{s} s\int_s^\infty k'\parens{\tilde s} \dif \tilde s \dif s
=
\\
-2 \int_{-\infty}^\infty k'\parens{s} s k\parens{s} \dif s =
 \kappa_2,
\end{multline*}
where the last equality follows using integration by parts. 
\end{proof}

\begin{proof}[Proof of \cref{lem:V symmetric}]
    Let $p>0$ and $z= p^{1/\parens{n-1}}$.   To obtain \cref{eq:H symmetric simplification}, simply note that ${\zeta}\parens{p} = Q_c'\parens{p} p = z Q'\parens{z}/\parens{n-1}$.
    Suppose first that $\tilde s \geq s \geq 0$.  Then,\footnote{The expansions below are informal to reduce length, but we have verified that they obtain if they are conducted more rigorously.} 
\[
\maligned{
    H\parens{p+sh,p+\tilde sh} & \simeq 
    \frac{
\parens{n-1} z\parens{1-z} + z^{2-n} sh - z^{3-n} sh - z^{3-n} \tilde s h}{n\parens{n-1}}
\cparens{p^2 +p\parens{s+\tilde s}h} \times \\
& \quad
\cparens{ Q'^2\parens{z}+ Q''\parens{z} Q'\parens{z}z^{2-n}\parens{s+\tilde s}h/\parens{n-1}},
    \\
    H\parens{p+sh,p} & \simeq \frac{
        \parens{n-1} z\parens{1-z}  - z^{3-n} sh }{n\parens{n-1}}
    \parens{ p^2+p sh} 
     \cparens{ Q'^2\parens{z}+ Q''\parens{z} Q'\parens{z}z^{2-n} sh/\parens{n-1}}, 
    \\
    H\parens{p,p+\tilde sh} &   \simeq \frac{
        \parens{n-1} z\parens{1-z}  - z^{3-n} \tilde sh }{n\parens{n-1}}
    \parens{ p^2+p \tilde sh}  \cparens{ Q'^2\parens{z}+ Q''\parens{z} Q'\parens{z}z^{2-n} \tilde sh/\parens{n-1}}, 
    \\
    H\parens{p,p} & =\frac{
        \parens{n-1} z\parens{1-z}  }{n\parens{n-1}}
    p^2 Q'^2\parens{z},
}
\]
where $\simeq$ means that $o\parens{h}$ terms are omitted.  Thus,
\(
\lim_{h\to 0} \sH_h\parens{p,s,\tilde s} = z^n  Q'^2\parens{z} s/ \cparens{{n\parens{n-1}} }.
\)
Repeating the same argument for other $s,\tilde s,0$ orderings we get
\[
\lim_{h\to 0} \sH_h\parens{p,s,\tilde s} = z^n Q'^2\parens{z} |\Med\parens{s,\tilde s,0}| / \cparens{{n\parens{n-1}} }
\]
Now use the integration argument from the proof of \cref{lem:V} to obtain the claimed result. 
\end{proof}

\begin{proof}[Proof of \cref{lem:V asymmetric}]
The proof follows the same path as that of \cref{lem:V symmetric} and is hence omitted.   
\end{proof}

\subsection{Boundary kernels}

\begin{proof}[Proof of \cref{lem:boundary kernel}]
	Let $\sI_j = \busyint{s^j {\phi}\parens{s}}$.  Note that
	$\sI_0= {\Omega}_{{\psi}0}$,  $\sI_1= - {\Omega}_{{\psi}1}$, $\sI_2= {\Omega}_{{\psi}2}+ {\Omega}_{{\psi}0}$.  Thus, we need 
	\[
	\left\{
	\begin{array}{rcll}
	{\omega}_{{\psi}1}{\Omega}_{{\psi}0} &- & {\omega}_{{\psi}2} {\Omega}_{{\psi}1} & = 1, \\
	-{\omega}_{{\psi}1}{\Omega}_{{\psi}1}&+ &{\omega}_{{\psi}2}\parens{{\Omega}_{{\psi}0}+{\Omega}_{{\psi}2}} & = 0.
	\end{array}
	\right.
	\]
	Solve for ${\omega}_{{\psi}1},{\omega}_{{\psi}2}$.
\end{proof}

\begin{lem} \label{lem:complicated boundary kernel}
Let $\sI_j$ be defined as in the proof of \cref{lem:boundary kernel} and let
$
 k_{{\psi}h}\parens{s \bkdelim p} = \sum_{j=1}^3 s^{j-1} {\omega}_{{\psi}j} {\phi}\parens{s},
$    
where
\[
\begin{bmatrix}
{\omega}_{{\psi}1}\\
{\omega}_{{\psi}2}\\
{\omega}_{{\psi}3}
\end{bmatrix}
=
\begin{bmatrix}
 \sI_0 & \sI_1 & \sI_2 \\
 \sI_1 & \sI_2 & \sI_3 \\
 \sI_2 & \sI_3 & \sI_4
\end{bmatrix}^{-1}
\begin{bmatrix}
1\\
0\\
1
\end{bmatrix}
\]
Then $k_{{\psi}h}$ satisfies the requirements of \cref{eq:boundary kernel requirements} everywhere, including at the boundaries.
\proof
Trivial. \qed
\end{lem}

\begin{lem} \label{lem:tedious kernel integrals}
    For any $\bl {\upsilon}< \bar {\upsilon}$,
    \[
    \maligned{
        \int_{\bl {\upsilon}}^{\bar {\upsilon}} sk'\parens{-s} \dif s &=
        \bl {\upsilon} k\parens{-\bl {\upsilon}} - \bar {\upsilon} k\parens{-\bar {\upsilon}} + \int_{\bl {\upsilon}}^{\bar {\upsilon}} k\parens{-s} \dif s, 
        \\
        \int_{\bl {\upsilon}}^{\bar {\upsilon}} sk'\parens{-s} k\parens{-s} \dif s &=
        \frac{1}{2}
        \parens[\Big]{
            \bl {\upsilon} k^2\parens{-\bl {\upsilon}} - \bar {\upsilon} k^2\parens{-\bar {\upsilon}} + \int_{\bl {\upsilon}}^{\bar {\upsilon}} k^2\parens{-s} \dif s}.
    }
    \]    
    \proof
    Follows using integration by parts. \qed
\end{lem}

\begin{proof}[Proof of \cref{thm:ehatpsi}]
%
Consider $\ahTpsi$.	  We have
\begin{multline} \label{eq:psitheorem1}
 \sqrt{Th}  \cparens{ \ahTpsi\parens{p} - {\alpha}\parens{p}} =
 \sqrt{Th} \parens[\bigg]{\psi'\parens{p} \uint[s]{{\alpha}_T\parens{s} k_{{\psi}h}\parens[\Big]{\frac{{\psi}\parens{p}- {\psi}\parens{s}}h\Bigm\bkdelim p }} - \alpha\parens{p}}
 \\
 =
\underbrace{\sqrt\frac Th \psi'\parens{p} \uint[s]{\cparens{{\alpha}_T\parens{s} - {\alpha}\parens{s}} k_{{\psi}h}\parens[\Big]{\frac{{\psi}\parens{p}- {\psi}\parens{s}}h\Bigm\bkdelim p }} }_{\mathrm{I}}
\\
+ 
\underbrace{\sqrt\frac Th \parens[\bigg]{
\psi'\parens{p} \uint[s]{{\alpha}\parens{s}k_{{\psi}h}\parens[\Big]{\frac{{\psi}\parens{p}- {\psi}\parens{s}}h\Bigm\bkdelim p }} - {\alpha}\parens{p}}}_{\mathrm{II}}.	
\end{multline}
Let ${\Psi}\parens{s} = {\alpha}\cparens{{\psi}^{-1}\parens{s}} / {\psi}'\cparens{{\psi}^{-1}\parens{s}}$.  Then term II in \cref{eq:psitheorem1} becomes by substitution of $s \leftarrow \cparens{{\psi}\parens{p}- {\psi}\parens{s} }/h$,
\[
 \sqrt{Th} \parens[\bigg]{
 	{\psi}'\parens{p} \busyint{{\Psi}\cparens{{\psi}\parens{p}+sh} k_{{\psi}h}\parens{-s \bkdelim p}} - {\alpha}\parens{p}}=
 \\
 \frac{\Xi {\psi}'\parens{p} {\Psi}''\cparens{{\psi}\parens{p}}}{2}
 +\oo
\]
for all $0< p < 1$
by a standard kernel bias expansion.  For $p=0,1$ the asymptotic bias differs by a multiplicative constant.  Note that 
$
 {\psi}'{\Psi}''
 $
 equals \cref{eq:fugly2},
which produces the asserted asymptotic bias.

Now term I in \cref{eq:psitheorem1}.  Integration by parts produces
\begin{multline}\label{eq:psitheorem2}
 \underbrace{\sqrt\frac{T}{h} {\psi}'\parens{p}\cparens{ \beT\parens{1}-e\parens{1}} k_{{\psi}h}\parens[\Big]{\frac{{\psi}\parens{p}-{\psi}\parens{1}}h \Bigm\bkdelim p}}_{\mathrm{I}}
 \\
 +
 \underbrace{\sqrt\frac{T}{h^3} {\psi}'\parens{p} \uint[s]{ {\psi}'\parens{s}
\cparens{\beT\parens{s}-e\parens{s} - \beT\parens{p}+e\parens{p}} k_{{\psi}h}'\parens[\Big]{\frac{{\psi}\parens{p}-{\psi}\parens{s}}h \Bigm\bkdelim p}}
}_{\mathrm{II}}
\\
+
 \underbrace{\sqrt\frac{T}{h^3} {\psi}'\parens{p} \cparens{\beT\parens{p}-e\parens{p}} \uint[s]{{\psi}'\parens{s}k_{{\psi}h}'\parens[\Big]{\frac{{\psi}\parens{p}-{\psi}\parens{s}}h \Bigm\bkdelim p}}}_{\mathrm{III}}.
\end{multline}
Term I in \cref{eq:psitheorem2} vanishes because $\beT\parens{1}$ converges at rate $T$.  Term III equals
\begin{multline} \label{eq:psitheorem3}
 \sqrt{\frac Th} {\psi}'\parens{p} \cparens{\beT\parens{p}-e\parens{p}} k_{{\psi}h}\parens[\Big]{\frac{{\psi}\parens{p}-{\psi}\parens{0}}h \Bigm\bkdelim p} \\
 - \sqrt{\frac Th} {\psi}'\parens{p} \cparens{\beT\parens{p}-e\parens{p}} k_{{\psi}h}\parens[\Big]{\frac{{\psi}\parens{p}-{\psi}\parens{1}}h \Bigm\bkdelim p}.
\end{multline}
For fixed $0\leq p\leq 1$, \cref{eq:psitheorem3} is $\opone$ by the conditions on $k_{{\psi}h}$.

Finally, term II in \cref{eq:psitheorem2}.  Consider fixed $0\leq p\leq 1$. Substitute $s \ot \cparens{{\psi}\parens{s}-{\psi}\parens{p}}/h$ to obtain
\begin{multline*}
\sqrt\frac Th {\psi}'\parens{p} \busyint{ 
	\parens[\Big]{\beT\sparens[\big]{{\psi}^{-1}\cparens{{\psi}\parens{p}+sh}}-e\sparens[\big]{{\psi}^{-1}\cparens{{\psi}\parens{p}+sh}} 
		- \beT\parens{p}+e\parens{p}} k_{{\psi}h}'\parens{-s \bkdelim p}}
\\
\simeq
\frac{{\psi}'\parens{p}}{\sqrt{h}} \busyint{ 
	\parens[\Big]{\G\parens[\Big]{p+ \frac{sh}{{\psi}'\parens{p}}}
		- \G\parens{p}} k_{{\psi}h}'\parens{-s \bkdelim p}},
\end{multline*}
which by a tedious repetition of the arguments of \cref{thm:ehat} has a limiting mean zero normal distribution with variance
\[
\lim_{h\to 0 }{\psi}'^2\parens{p}\busyint{ \busyint[\tilde s]{ 
\sH_h\parens{p,s/{\psi}'\parens{p},\tilde s/{\psi}'\parens{p}}
        k_{{\psi}h}'\parens{-\tilde s \bkdelim p}
    }\;
    k_{{\psi}h}'\parens{-s \bkdelim p}
},
\]
which under \cref{eq:Hstar simple} simplifies to
\begin{multline}\label{eq:psitheorem4}
{\zeta}^2\parens{p} {\psi}'\parens{p} \lim_{h\to 0} \parens[\bigg]{
    \bar {\upsilon}_{{\psi}h} k_{{\psi}h}^2\parens{\bar {\upsilon}_{{\psi}h} \bkdelim p} 
    -
    \bl {\upsilon}_{{\psi}h} k_{{\psi}h}^2\parens{\bl {\upsilon}_{{\psi}h} \bkdelim p} 
    \\
    -
    k_{{\psi}h}\parens{-\bar {\upsilon}_{{\psi}h} \bkdelim p} \int_0^{\bar {\upsilon}_{{\psi}h}} k_{{\psi}h}\parens{-s \bkdelim p} \dif s
    -
    k_{{\psi}h}\parens{-\bl {\upsilon}_{{\psi}h} \bkdelim p} \int_{\bl {\upsilon}_{{\psi}h}}^0 k_{{\psi}h}\parens{-s \bkdelim p} \dif s
    +
    \int_{\bl {\upsilon}_{{\psi}h}}^{\bar {\upsilon}_{{\psi}h}} k_{{\psi}h}^2\parens{-s \bkdelim p} \dif s
}
\\
=
{\zeta}^2\parens{p} {\psi}'\parens{p}\lim_{h\to 0} \int_{\bl {\upsilon}_{{\psi}h}}^{\bar {\upsilon}_{{\psi}h}} k_{{\psi}h}^2\parens{-s \bkdelim p} \dif s,
\end{multline}
as promised.  For $0<p<1$, the right hand side in \cref{eq:psitheorem4} reduces to
${\zeta}^2\parens{p} {\psi}'\parens{p} / \sqrt{\pi}$.\joris{The variance may be discontinuous at $p=1$.}  
\end{proof}

\begin{proof}[Proof of \cref{thm:ebarpsi}]
Consider $\bar {\alpha}_{T{\psi}}$.	  We have
\begin{multline} \label{eq:psitheoremalt1}
	\sqrt{Th}  \cparens{ \bar {\alpha}_{T{\psi}}\parens{p} - {\alpha}\parens{p}} =
	\sqrt{Th} \parens[\bigg]{\uint[s]{{\psi}'\parens{s} {\alpha}_T\parens{s} k_{{\psi}h}\parens[\Big]{\frac{{\psi}\parens{p}- {\psi}\parens{s}}h\Bigm\bkdelim p }} - \alpha\parens{p}}
	\\
	=
	\underbrace{\sqrt\frac Th  \uint[s]{\psi'\parens{s}\cparens{{\alpha}_T\parens{s} - {\alpha}\parens{s}} k_{{\psi}h}\parens[\Big]{\frac{{\psi}\parens{p}- {\psi}\parens{s}}h\Bigm\bkdelim p }} }_{\mathrm{I}}
	\\
	+ 
	\underbrace{\sqrt\frac Th \parens[\bigg]{
	 \uint[s]{{\psi}'\parens{s}{\alpha}\parens{s}k_{{\psi}h}\parens[\Big]{\frac{{\psi}\parens{p}- {\psi}\parens{s}}h\Bigm\bkdelim p }} - {\alpha}\parens{p}}}_{\mathrm{II}}.	
\end{multline}
Let ${\Psi}\parens{s} = {\alpha}\cparens{{\psi}^{-1}\parens{s}}$. Then term II in \cref{eq:psitheoremalt1} becomes by substitution of $s \leftarrow \cparens{{\psi}\parens{p}- {\psi}\parens{s} }/h$ and a standard kernel bias expansion,
\begin{multline} \label{eq:psitheoremalt1.5}
\sqrt{Th} \parens[\bigg]{
 \busyint{{\Psi}\cparens{{\psi}\parens{p}+sh} k_{{\psi}h}\parens{-s\bkdelim p}} - {\alpha}\parens{p}}=
\frac{\Xi  {\Psi}''\cparens{{\psi}\parens{p}}}{2}
\busyint{ s^2 k_{{\psi}h}\parens{-s\bkdelim p}} + \oo
\\
=
\frac{\Xi  {\Psi}''\cparens{{\psi}\parens{p}}}{2} \lim_{h\to 0} \frac{ \parens{{\Omega}_{{\psi}0}+{\Omega}_{{\psi}2}}^2 - {\Omega}_{{\psi}1} \parens{3{\Omega}_{{\psi}1}+{\Omega}_{{\psi}3}}}{{\Omega}_{{\psi}0}^2 + {\Omega}_{{\psi}0} {\Omega}_{{\psi}2}- {\Omega}_{{\psi}1}^2} +\oo.
\end{multline}
The limit in the right hand side in \cref{eq:psitheoremalt1.5} equals one for all $0<p<1$ and equals $\parens{\pi-4}/\parens{\pi-2}$ for $p=0,1$.  Since
\[
{\Psi}''= \frac{{\alpha}''}{{\psi}'^2} -  \frac{{\alpha}' {\psi}''}{{\psi}'^3},
\]
we get the asserted asymptotic bias.

Now term I in \cref{eq:psitheoremalt1}.  Integration by parts produces
\begin{multline}\label{eq:psitheoremalt2}
	\underbrace{\sqrt\frac{T}{h} {\psi}'\parens{1}\cparens{ \beT\parens{1}-e\parens{1}} k_{{\psi}h}\parens[\Big]{\frac{{\psi}\parens{p}-{\psi}\parens{1}}h \Bigm\bkdelim p}}_{\mathrm{I}}
	\\
	+
	\underbrace{\sqrt\frac{T}{h^3}  \uint[s]{ {\psi}'^2\parens{s}
			\cparens{\beT\parens{s}-e\parens{s} - \beT\parens{p}+e\parens{p}} k_{{\psi}h}'\parens[\Big]{\frac{{\psi}\parens{p}-{\psi}\parens{s}}h \Bigm\bkdelim p}}
	}_{\mathrm{II}}
	\\
	+
	\underbrace{\sqrt\frac{T}{h^3}  \cparens{\beT\parens{p}-e\parens{p}} \uint[s]{{\psi}'^2\parens{s}k_{{\psi}h}'\parens[\Big]{\frac{{\psi}\parens{p}-{\psi}\parens{s}}h \Bigm\bkdelim p}}}_{\mathrm{III}}
\\
-
	\underbrace{\sqrt\frac{T}{h}  \uint[s]{ {\psi}''\parens{s}
		\cparens{\beT\parens{s}-e\parens{s} - \beT\parens{p}+e\parens{p}} k_{{\psi}h}\parens[\Big]{\frac{{\psi}\parens{p}-{\psi}\parens{s}}h \Bigm\bkdelim p}}
}_{\mathrm{IV}}
\\	
-\underbrace{\sqrt\frac{T}{h}  \cparens{\beT\parens{p}-e\parens{p}} \uint[s]{{\psi}''\parens{s}k_{{\psi}h}\parens[\Big]{\frac{{\psi}\parens{p}-{\psi}\parens{s}}h \Bigm\bkdelim p}}}_{\mathrm{V}}
\end{multline}
Term I in \cref{eq:psitheoremalt2} vanishes because $\beT\parens{1}$ converges faster than $\sqrt{T/h}$.  Term III equals
\begin{multline} \label{eq:psitheoremalt3}
	\sqrt{\frac Th} {\psi}'\parens{0} \cparens{\beT\parens{p}-e\parens{p}} k_{{\psi}h}\parens[\Big]{\frac{{\psi}\parens{p}-{\psi}\parens{0}}h \Bigm\bkdelim p} \\
	- \sqrt{\frac Th} {\psi}'\parens{1} \cparens{\beT\parens{p}-e\parens{p}} k_{{\psi}h}\parens[\Big]{\frac{{\psi}\parens{p}-{\psi}\parens{1}}h \Bigm\bkdelim p}
	\\
+ 	\sqrt{\frac Th}\cparens{\beT\parens{p}-e\parens{p}} \int_0^1 {\psi}''\parens{s} k_{{\psi}h}\parens[\Big]{\frac{{\psi}\parens{p}-{\psi}\parens{s}}h \Bigm\bkdelim p} \dif s.
\end{multline}
For fixed $0<p<1$, \cref{eq:psitheoremalt3} is $\opone$ by the conditions on $k_{{\psi}h}$ and at $p=0,1$, the superconsistency of $\beT\parens{p}$ takes care of the problem.

By a change of variables, terms IV and V in \cref{eq:psitheoremalt2} are $\opone$.

Finally, term II in \cref{eq:psitheoremalt2}. Consider fixed $0< p < 1$.  Let ${\Psi}\parens{p;sh}={\psi}^{-1}\cparens{{\psi}\parens{p}+sh}$. Substitute $s \ot \cparens{{\psi}\parens{s}-{\psi}\parens{p}}/h$ to obtain
\begin{multline*}
	\sqrt\frac Th  \busyint{ {\psi}'\cparens{{\Psi}\parens{p;sh}}
		\parens[\Big]{\beT\cparens[\big]{{\Psi}\parens{p;sh}}-e\cparens[\big]{{\Psi}\parens{p;sh}} 
			- \beT\parens{p}+e\parens{p}} k_{{\psi}h}'\parens{-s \bkdelim p}}
	\\
	\simeq
	\frac{{\psi}'\parens{p}}{\sqrt{h}} \busyint{ 
		\cparens[\Big]{\G\parens[\Big]{p+ \frac{sh}{{\psi}'\parens{p}}}
			- \G\parens{p}} k_{{\psi}h}'\parens{-s \bkdelim p}},
\end{multline*}
which has a limiting mean zero normal distribution with covariance kernel
\[
{\psi}'^2\parens{p}
\lim_{h\to 0}
\busyint{ 
\busyint[\tilde s]{
	 k_{{\psi}h}'\parens{-s \bkdelim p}k_{{\psi}h}'\parens{-\tilde s \bkdelim p} \sH_h\cparens{p,s/{\psi}'\parens{p},s'/{\psi}'\parens{p}}
	}
}.
\]
Note that $\bl {\upsilon}_{\psi} \to -\infty$, $\bar {\upsilon}_{\psi} \to\infty$, $k'_{{\psi}h} \to {\phi}'$ as $h \to 0$.
\end{proof}

\begin{proof}[Proof of \cref{lem:computation alpha bar}]
   We have
\begin{multline*}
  \bar {\alpha}_{T{\psi}}\parens{p} = \frac1h \uint[s]{ {\psi}'\parens{s} {\alpha}_T\parens{s} k_{{\psi}h}\parens[\Big]{ \frac{ {\psi}\parens{p}-{\psi}\parens{s}}h \Bigm \bkdelim p}}  
  \\
  =
  \frac1h \sum_{j=1}^T {\alpha}_{Tj} \int_{\frac{j-1}T}^\frac jT  {\psi}'\parens{s}  k_{{\psi}h}\parens[\Big]{ \frac{{\psi}\parens{p}-{\psi}\parens{s}}h \Bigm \bkdelim  p } \dif s 
 = 
  \sum_{j=1}^T {\alpha}_{Tj} \int_{{\upsilon}_{j-t}\parens{p}}^{{\upsilon}_j\parens{p}} k_{{\psi}h}\parens{-s \bkdelim p} \dif s
  \\
  = \sum_{j=1}^T {\alpha}_{Tj}  {\Lambda}_{{\psi}j}\parens{p}. 
\end{multline*} 
The second statement in \cref{lem:computation alpha bar} is a natural property of the normal distribution.
\end{proof}

\drop{
\begin{proof}[Proof of \cref{lem:bs psi}]
Note that\joris{lemma no longer in text, so this proof should be removed}
\begin{multline*}
 \hat A\parens{p} = \frac{1}{h} \uint[s]{ \breve A_T\parens{s} {\psi}'\parens{s} k_{{\psi}h}\parens[\Big]{\frac{{\psi}\parens{p}-{\psi}\parens{s}}h \Bigm\bkdelim p }}
 =\\
 \frac{1}{h} \sum_{j=1}^T \int_{\frac{j-1}T}^{\frac jT} \sparens[\bigg]{
{\alpha}_{Tj} s - \cparens[\Big]{ \breve e_{T,j-1}+ \parens[\Big]{s- \frac{j-1}T} {\alpha}_{Tj}}}  {\psi}'\parens{s} k_{{\psi}h}\parens[\Big]{\frac{{\psi}\parens{p}-{\psi}\parens{s}}h \Bigm\bkdelim p } \dif s
= \\
\sum_{j=1}^T S_{Tj} {\Lambda}_{{\psi}j}\parens{p},
\end{multline*}   
which yields the desired result. 
\end{proof}
}

\subsection{Another boundary correction}

\newcommand{\rated}{\epsilon_{dT}}

Below, let $\rated$ denote the convergence rate of $\hat d$.

\begin{lem} \label{lem:ahat minus a}
    Suppose that for $0\leq \hat a,a\leq 1$, $\hat a-a = O_p\parens{{\epsilon}_T}$.  Then
 \[
 \beT\parens{\hat a}- e\parens{a}
 =
 \beT\parens{a}-e\parens{a} + e\parens{\hat a}-e\parens{a} + O_p\parens[\big]{\sqrt{\parens{-{\epsilon}_T\log {\epsilon}_T}/T}}.
 \]
 \proof
This is just a rearrangement of
$
 \beT\parens{\hat a}- e\parens{\hat a} - \beT\parens{a}+ e\parens{a} = O_p\parens[\big]{\sqrt{\parens{-{\epsilon}_T\log {\epsilon}_T}/T}},
 $  which is Levy's modulus of continuity theorem. \qed
\end{lem}

\begin{lem} \label{lem:phihat minus phi}
     Let ${\psi}$ satisfy condition (iii) in \Cref{thm:boundary correction reflection}. For any $0\leq C<\infty$,
\(
 \sup_{0\leq s \leq C} \abs{ \hat {\rho}\cparens{{\psi}\parens{1+sh}} - {\rho}\cparens{{\psi}\parens{1+sh}}} = O_p\parens{h^2 {\epsilon}_{dT}}.
\)    
\proof
Follows immediately from writing out noting that $\hat {\rho}, {\rho}$ are third degree polynomials and noting that ${\psi}\parens{1}=0$. \qed
\end{lem}

\begin{lem} \label{lem:phi minus sh}
     Let ${\psi}$ satisfy condition (iii) in \Cref{thm:boundary correction reflection} and ${\rho}$ defined as in \Cref{section:kz}. For any $0\leq C<\infty$,
\(
 \sup_{0\leq s \leq C} \abs{{\rho}\cparens{{\psi}\parens{1+sh}} -sh} + \sup_{0\leq s \leq C} \abs{\hat {\rho}\cparens{{\psi}\parens{1+sh}} -sh} = O_p\parens{h^2}.
\)    
\proof
Follows from the mean value theorem and the fact that $\hat d$ is bounded in probability. \qed
\end{lem}

\begin{lem} \label{lem:excruciating}
 \begin{multline*}
   \beT\sparens{1-\hat {\rho}\cparens{{\psi}\parens{1+sh}}} - e\parens{1-sh}
   =\\
   \cparens{\beT\parens{1-sh}-e\parens{1-sh} } + \cparens{e\sparens{1-\hat {\rho}\cparens{{\psi}\parens{1+sh}}}- e\parens{1-sh}}
   +
   o_p\parens{\sqrt{h/T}}.
\end{multline*}
\proof
Follows directly from \cref{lem:ahat minus a,lem:phihat minus phi,lem:phi minus sh}.  \qed
\end{lem} 
    
\begin{lem} \label{lem:hideous}
     For any $0\leq C<\infty$,
\[
\sup_{0\leq t \leq C} \abs[\Big]{
\beT\parens{1+th} - 2 \beT\parens{1} + \beT\parens{1-th}
- e\parens{1+th} + 2 e\parens{1} - e\parens{1-th}
-
\diam\sB^R_T\parens{t}
}
= o_p\parens[\big]{ \sqrt{h/T}},
\]
where $\diam\sB_T^R\parens{t}=0$ for $t\leq 0$ and for $t>0$ it is
\begin{multline*}
\diam\sB^R_T\parens{t} = \frac{ e\sparens{1-\hat {\rho}\cparens{{\psi}\parens{1+th}}} - e\sparens{1-{\rho}\cparens{{\psi}\parens{1+th}}}}
{{\psi}'\parens{1+th}} - \\
h \int_0^t \parens[\Big]{
e\sparens{1-\hat {\rho}\cparens{{\psi}\parens{1+sh}}}
-
e\sparens{1-{\rho}\cparens{{\psi}\parens{1+sh}}}} \frac{{\psi}''\parens{1+sh}}{{\psi}'^2\parens{1+sh}} \dif s.
\end{multline*}
\proof
Using integration by parts we get
\begin{multline} \label{eq:eT1+th}
  \beT\parens{1+th}  
  =
  \beT\parens{1}+
\int_0^{th} {\alpha}_T\parens{1+s} \dif s 
=\\
\beT\parens{1}+
\int_0^{th}
{\alpha}_T\sparens[\big]{ 1 - \hat {\rho}\cparens{{\psi}\parens{1+s}}} \hat {\rho}'\cparens{1+{\psi}\parens{1+s}} \dif s =
\\
2\beT\parens{1} - \frac{ \beT\sparens{1-\hat {\rho}\cparens{{\psi}\parens{1+th}}} }{{\psi}'\parens{1+th}}
-
h\int_0^t  \beT\sparens{1-\hat {\rho}\cparens{{\psi}\parens{1+sh}}}\frac{{\psi}''\parens{1+sh}}{{\psi}'^2\parens{1+sh}} \dif s.
\end{multline}
Now, by \cref{lem:excruciating} we have
\[
 \beT\sparens{1-\hat {\rho}\cparens{{\psi}\parens{1+sh}}} 
 =
 \beT\parens{1-sh} - e\parens{1-sh} + e\sparens{1-\hat {\rho}\cparens{{\psi}\parens{1+sh}}} + o_p\parens[\big]{\sqrt{h/T}},
\]
uniformly in $0\leq s\leq C$.  Thus, \cref{eq:eT1+th} is
\begin{multline} \label{eq:horrible}
2\beT\parens{1} - \cparens{\beT\parens{1-th} - e\parens{1-th}}
-
\frac{e\sparens{1-\hat {\rho}\cparens{{\psi}\parens{1+th}}}}{{\psi}'\parens{1+th}}
-\\
h\int_0^t e\sparens{1-\hat {\rho}\cparens{{\psi}\parens{1+sh}}} \frac{{\psi}''\parens{1+sh}}{{\psi}'^2\parens{1+sh}} \dif s
+ o_p\parens[\big]{\sqrt{h/T}}.
\end{multline}
Repeat \cref{eq:eT1+th} for $e$ in lieu of $\beT$ and subtract from \cref{eq:horrible}. \qed

%
%
\end{lem}

\begin{lem}\label{lem:Bdiamond basics}
	For $\diam \sB_T^R$ defined in \cref{lem:hideous} and any $0\leq C<\infty$,
	$\sup_{0\leq s \leq C} \abs{\diam \sB_T^R\parens{s} + {\alpha}\parens{1}\parens{\hat d -d}s^2h^2} = o_p\parens{h^2 {\epsilon}_{dT}+h^3}$.
	\proof    
	Note that by the mean value theorem and the definitions of ${\rho},\hat {\rho},{\psi}$,
	\[
	e\sparens{1-\hat {\rho}\cparens{{\psi}\parens{1+sh}}} - e\sparens{1-{\rho}\cparens{{\psi}\parens{1+sh}}}
	=
	- {\alpha}\parens{1} \cparens{ \hat {\rho}''\parens{0} - {\rho}''\parens{0}} \frac{s^2h^2}2 + o_p\parens{h^3}.
	\]
	Further,
	${\rho}''\parens{0} = 2d$ and  	$\hat {\rho}''\parens{0} = 2\hat d$.  The stated result then follows from the fact that $\hat d-d = O_p\parens{{\epsilon}_{dT}}$.
	 \qed
\end{lem}

\begin{lem}\label{lem:Bdiamond}
    For $\diam \sB_T^R$ defined in \cref{lem:hideous} and  any $0\leq C<\infty$,
$\sup_{0\leq s \leq C} \abs{\diam \sB_T^R\parens{s}} = O_p\parens{h^2 {\epsilon}_{dT}} + o_p\parens{h^3}$.
\proof    
This is an immediate consequence of \cref{lem:Bdiamond basics}. \qed
\end{lem}

    
\begin{lem} \label{lem:nbh}
    Uniformly in $0\leq t\leq 1$,\footnote{For $t>1$, we get standard asymptotics.}
    \begin{multline} \label{eq:nbh1}
 \frac1h \int_{-\infty}^\infty k\parens[\Big]{ \frac{{\psi}\parens{1-th}-{\psi}\parens{s}}h} {\psi}'\parens{s}
 \cparens{ {\alpha}_T\parens{s} - {\alpha}\parens{s}} \dif s  
\simeq\\
-\frac h8 {\alpha}\parens{1} \parens{1-t}^3 \parens{t+3}\parens{\hat d - d}	+ 
\frac1h \int_t^1 k'\parens{s}\sparens{ \beT\cparens{1-\parens{s-t}h} - e\cparens{1-\parens{s-t}h}}
\dif s
\\-
\frac1h \int_{-1}^t k'\parens{s}\sparens{ \beT\cparens{1+\parens{s-t}h} - e\cparens{1+\parens{s-t}h}}
\dif s,
 \end{multline}
where $\simeq$ means that any omitted terms are asymptotically negligible.
\proof
The left hand side in \cref{eq:nbh1} is by integration by parts equal to
\begin{multline} \label{eq:nbh2}
 \frac1{h^2} \int_{-\infty}^\infty
 k'\parens[\bigg]{ \frac{{\psi}\parens{1-th}-{\psi}\parens{s}}h} {\psi}'^2\parens{s} \cparens{ \beT\parens{s}-e\parens{s}} \dif s
 \\	-
  \frac1h \int_{-\infty}^\infty
 k\parens[\bigg]{ \frac{{\psi}\parens{1-th}-{\psi}\parens{s}}h} {\psi}''\parens{s} \cparens{ \beT\parens{s}-e\parens{s}} \dif s.	
\end{multline}
The first term in \cref{eq:nbh2} dominates the second term, so we deal with the first term only.  The first term in \cref{eq:nbh2} is for
$\varsigma_{ts}\parens{h}={\psi}^{-1}\cparens{{\psi}\parens{1-th}+sh}$ equal to 
\begin{multline} \label{eq:nbh3}
-\frac1h \int_{-1}^1
k'\parens{s} {\psi}'\cparens{\varsigma_{ts}\parens{h}}  
\parens[\big]{ \beT\cparens{\varsigma_{ts}\parens{h}}-e\cparens{\varsigma_{ts}\parens{h} }} \dif s
\\
\simeq 
-\frac1h \int_{-1}^1
k'\parens{s} {\psi}'\cparens{\varsigma_{ts}\parens{0}+\varsigma_{ts}'\parens{0}h}  
\parens[\big]{ \beT\cparens{\varsigma_{ts}\parens{0}+\varsigma_{ts}'\parens{0}h} -e\cparens{\varsigma_{ts}\parens{0}+\varsigma_{ts}'\parens{0}h}}  \dif s
\\
\simeq 
-\frac1h \int_{-1}^1
k'\parens{s}
 \sparens{\beT\cparens{1+\parens{s-t}h}  - e\cparens{1+\parens{s-t}h}} \dif s.
\end{multline}
Since $\beT\parens{1}$ is a super--consistent estimator of $e\parens{1}$, we get by \cref{lem:hideous} that 
\begin{multline} \label{eq:nbh4}
-\frac1h \int_t^1
k'\parens{s}
\sparens{\beT\cparens{1+\parens{s-t}h}  - e\cparens{1+\parens{s-t}h}} \dif s
\simeq \\
-\frac1h \int_t^1 k'\parens{s}\parens[\Big]{
	\diam \sB_T^R\parens{s-t}
	- 
	\sparens{ \beT\cparens{1-\parens{s-t}h} - e\cparens{1-\parens{s-t}h}}
} \dif s
\overset{\text{\cref{lem:Bdiamond basics}}}{\simeq}
\\
h {\alpha}\parens{1} \parens{\hat d-d}\int_t^1 k'\parens{s} \parens{s-t}^2 \dif s
	+ 
	\frac1h \int_t^1 k'\parens{s}\sparens{ \beT\cparens{1-\parens{s-t}h} - e\cparens{1-\parens{s-t}h}}
 \dif s.
\end{multline}
Since $k$ is the Epanechnikov kernel, the right hand side in \cref{eq:nbh4} simplifies to
\[
-\frac h8 {\alpha}\parens{1} \parens{1-t}^3 \parens{t+3}\parens{\hat d - d}		+ 
\frac1h \int_t^1 k'\parens{s}\sparens{ \beT\cparens{1-\parens{s-t}h} - e\cparens{1-\parens{s-t}h}}
\dif s. \qedhere
\]
\end{lem}

\begin{lem} \label{lem:boundary correction bias}
Uniformly in $0\leq t\leq C$ for given $0<C<\infty$,
\[
\frac1h \int_{-\infty}^\infty
k\parens[\bigg]{ \frac{{\psi}\parens{1-th}-{\psi}\parens{s}}h} {\psi}'\parens{s} \cparens{{\alpha}\parens{s}-{\alpha}\parens{1-th}} \dif s
=
\cparens{{\alpha}''\parens{1} - {\alpha}'\parens{1}{\psi}''\parens{1}} \frac{h^2}{10} + o\parens{h^2}.
\]
\proof
Follows directly from a standard kernel bias expansion followed by an application of the mean value theorem, noting that
$\int_{-1}^1 k\parens{s} s^2 \dif s =1/5$. \qed
\end{lem}

\begin{proof}[Proof of \cref{thm:boundary correction reflection}]
The asymptotic bias is derived in \cref{lem:boundary correction bias}.  The second term in \cref{eq:boundary correction statement} corresponds to the first right hand side term in \cref{eq:nbh1}.  Thus, the only issue remaining is to show that $\sqrt{Th}$ times the sum of the second and third terms in \cref{eq:nbh1} have a zero mean limiting normal distribution with variance equal to $\sV^R\parens{t}$.

Now, using the shorthand ${\upsilon}_T=\beT - e$ and noting that $k$ is the Epanechnikov kernel, the second right hand side term in \cref{eq:nbh1} can be written as
\< \label{eq:bcr1}
 -\frac3{2h} \int_t^1 s {\upsilon}_T\cparens{1-\parens{s-t}h} \dif s 
 =
 -\frac3{2h} \int_0^{1-t} \parens{s+t} {\upsilon}_T\parens{1-sh} \dif s.
\>
The last right hand side term in \cref{eq:nbh1} is then
\<\label{eq:bcr2}
 \frac3{2h} \int_{-1}^t s {\upsilon}_T\cparens{1+\parens{s-t}h} \dif s
 =
  \frac3{2h} \int_0^{1+t} \parens{t-s} {\upsilon}_T\parens{1-sh} \dif s.
\>
Summing \cref{eq:bcr1,eq:bcr2} yields
\<\label{eq:bcr3}
-\frac3{2h} \int_0^{1-t} 2s {\upsilon}_T\parens{1-sh} \dif s + \frac3{2h} \int_{1-t}^{1+t} \parens{t-s} {\upsilon}_T\parens{1-sh} \dif s.
\>
Recall that $\diam k_t\parens{s} = \parens{3/2}\cparens{\parens{t-s} \one\parens{1-t \leq s\leq 1+t}-2s \one\parens{0\leq s\leq 1-t}}$.  Then \cref{eq:bcr3} becomes
\[
 \frac1h \int_0^{1+t} \diam k_t\parens{s} {\upsilon}_T\parens{1-sh} \dif s,
\]
which leads to the asserted limit distribution using techniques similar to the ones employed in proofs of e.g.\ \cref{thm:ebarpsi} above.

Now the simplification of $\sV^R\parens{t}$ if $H^*$ has the indicated form.  Note first that
\(
\sH_h\parens{1,-s,-\tilde s}
=
{\zeta}^2\parens{1}\min\parens{s,\tilde s}.
\)
Hence
\begin{multline}
\sV^R\parens{t} = {\zeta}^2\parens{1}\int_0^{1+t}\int_0^{1+t} \diam k_t\parens{s} \diam k_t\parens{\tilde s} \min\parens{s,\tilde s}\dif \tilde s \dif s
=
\\
2 {\zeta}^2\parens{1}\parens[\bigg]{ \int_0^{1-t}\!\!\!\!\!\! \diam k_t\parens{s} s \int_s^{1-t} \!\!\!\!\!\!\diam k_t\parens{\tilde s} \dif \tilde s \dif s
    +
   \int_0^{1-t} \!\!\!\!\!\! \diam k_t\parens{s} s \int_{1-t}^{1+t} \!\!\!\!\!\!\diam k_t\parens{\tilde s} \dif \tilde s \dif s 
   +
   \int_{1-t}^{1+t} \!\!\!\!\!\!\diam k_t\parens{s} s \int_s^{1+t} \!\!\!\!\!\!\diam k_t\parens{\tilde s} \dif \tilde s \dif s}. \label{eq:bcr4}
\end{multline}
Now,
\[
\maligned{
 \int_s^{1+t} \diam k_t\parens{\tilde s} \dif \tilde s & = 3\cparens{\parens{t-s}^2-1}/4, && \text{ if } 1-t\leq s\leq 1+t, \\
 \int_s^{1-t} \diam k_t\parens{\tilde s} \dif \tilde s & = 3\cparens{s^2-\parens{1-t}^2}/2, && \text{ if } 0 \leq s \leq 1-t,
}
\]
which implies that \cref{eq:bcr4} equals
\[
2{\zeta}^2\parens{1} \parens[\Big]{
	\frac35 \parens{1-t}^5
	+
	3t\parens{1-t}^4
	+
	\frac3{10} t^2\parens{-9t^3+30t^2-35t+15}
}
=
\frac35 {\zeta}^2\parens{1}
\cparens[\big]{ 2-t^2 \parens{t^3-5t+5}},
\]
as claimed.
\end{proof}

\subsection{Derivatives}
\label{app:derivatives}

\begin{lem} \label{lem:alpha symmetric}
In the symmetric case,
\[
 F_p\parens{p} = p^{1/\parens{n-1}}, \quad f_p\parens{p} = \frac{F_p^{2-n}\parens{p}}{n-1} , \quad f_p'\parens{p} = \frac{2-n}{\parens{n-1}^2} F_p^{3-2n}\parens{p}, \quad
 f_p''\parens{p} = \frac{\parens{2-n}\parens{3-2n}}{\parens{n-1}^2} F_p^{4-3n}\parens{p}.
\]
Further,
\<
\maligned{
{\alpha}\parens{p} & = \frac1{n-1} \sparens{ F_p\parens{p} Q'\cparens{F_p\parens{p}} + \parens{n-1} Q\cparens{F_p\parens{p}}}, \\
{\alpha}'\parens{p} & = \frac{F_p^{2-n}}{\parens{n-1}^2} \parens{nQ' + F_p Q''}, \\
{\alpha}''\parens{p} & = \frac{F_p^{3-2n}}{\parens{n-1}^3} \cparens[\big]{ \parens{2-n}nQ' + 3 F_p Q'' + F_p^2 Q'''}, \\
{\alpha}'''\parens{p} & =
\frac{F_p^{4-3n}}{\parens{n-1}^4} \cparens[\big]{
\parens{3-2n}\parens{2-n}n Q' + \parens{12-4n-n^2} F_p Q'' + \parens{8-2n} F_p^2 Q''' + F_p^3 Q''''
}.
}	
\>
\proof
Trivial, but messy, calculus. \qed
\end{lem}

\begin{proof}[Proof of \cref{thm:derivatives away from the boundary}]
Note that
\<
{\alpha}_{T{\psi}}'\parens{p} = \frac{{\psi}'\parens{p}}{h^2} \int {\psi}'\parens{s} {\alpha}_T\parens{s} k'\parens[\bigg]{\frac{{\psi}\parens{p}-{\psi}\parens{s}}h} \dif s.
\>
First, note that by integration by parts and substitution we have
\begin{multline*}
 \frac{{\psi}'\parens{p}}{h^2} \int {\psi}'\parens{s} {\alpha}\parens{s} k'\parens[\bigg]{\frac{{\psi}\parens{p}-{\psi}\parens{s}}h} \dif s
 \simeq
 \frac{{\psi}'\parens{p}}h \int {\alpha}'\parens{s} k\parens[\bigg]{\frac{{\psi}\parens{p}-{\psi}\parens{s}}h} \dif s =
 \\
 {\psi}'\parens{p} \int \frac{{\alpha}'}{{\psi}'}\sparens[\big]{{\psi}^{-1}\cparens{{\psi}\parens{p}+sh}} k\parens{s} \dif s
 \simeq
 {\alpha}'\parens{p} + 
 \frac{h^2}2 
 {\psi}'\parens{p} \sparens[\Big]{\frac{{\alpha}'}{{\psi}'}\cparens{{\psi}^{-1}}}''\cparens{{\psi}\parens{p}} \int k\parens{s} s^2 \dif s,
\end{multline*}
which yields the asserted bias after noting that the Epanechnikov kernel is a density with variance $1/5$.  Finally, skipping some steps that repeat steps taken in the proofs of earlier theorems and using integration by parts and substitution plus the properties of the Epanechnikov kernel,
\begin{multline*}
\frac{{\psi}'\parens{p}}{h^2} \int {\psi}'\parens{s} \cparens{{\alpha}_T\parens{s}-{\alpha}\parens{s}} k'\parens[\bigg]{\frac{{\psi}\parens{p}-{\psi}\parens{s}}h} \dif s
\simeq\\
\frac{{\psi}'\parens{p}}{h^3} \int \cparens{{\psi}'\parens{s}}^2 \cparens{\beT\parens{s}-e\parens{s} - \beT\parens{p} +e\parens{p}} k''\parens[\bigg]{\frac{{\psi}\parens{p}-{\psi}\parens{s}}h} \dif s
\\
\simeq
\frac{{\psi}'^3\parens{p}}{h^2} \int_{-1}^1 \cparens[\Big]{\beT\parens[\Big]{p + \frac{sh}{{\psi}'\parens{p}}}-e\parens[\Big]{p + \frac{sh}{{\psi}'\parens{p}}} - \beT\parens{p}+e\parens{p}} k''\parens{-s} \dif s
\\
= 
-\frac{3{\psi}'^3\parens{p}}{2h^2} \int_{-1}^1 \cparens[\Big]{\beT\parens[\Big]{p + \frac{sh}{{\psi}'\parens{p}}}-e\parens[\Big]{p + \frac{sh}{{\psi}'\parens{p}}} -\beT\parens{p}+e\parens{p}} \dif s
\\
\sim
-\frac{3{\psi}'^3\parens{p}}{2h^2\sqrt{T}} \int_{-1}^1 \cparens[\Big]{\G\parens[\Big]{p + \frac{sh}{{\psi}'\parens{p}}}-\G\parens{p}} \dif s,
\end{multline*}
which has the asserted limit distribution. 
\end{proof}

\subsection{Derivatives near the boundary}

\begin{lem} \label{lem:nbhd}
    Uniformly in $0\leq t\leq 1$,
    \begin{multline} \label{eq:nbhd1}
        \frac1{h^2} \int_{-\infty}^\infty k'\parens[\Big]{ \frac{{\psi}\parens{1-th}-{\psi}\parens{s}}h} {\psi}'\parens{s}
        \cparens{ {\alpha}_T\parens{s} - {\alpha}\parens{s}} \dif s  
        \simeq\\
        \frac{{\alpha}\parens{1}}2 \parens{1-t}^3 \parens{\hat d - d}	+ 
        \frac3{2h^2} \int_t^1 \sparens{ \beT\cparens{1-\parens{s-t}h} - e\cparens{1-\parens{s-t}h}}
        \dif s
        \\-
        \frac3{2h^2} \int_{-1}^t \sparens{ \beT\cparens{1+\parens{s-t}h} - e\cparens{1+\parens{s-t}h}}
        \dif s,
    \end{multline}
    \proof
    The line of proof is the same as \cref{lem:nbh} but with $k'$ instead of $k$, noting that $k''$ is constant whereas $k'$ is odd. \qed
\end{lem}

\begin{lem} \label{lem:alpha' bias near one}
   Uniformly in $0\leq t\leq C$ for given $0<C<\infty$,
   \begin{multline} \label{eq:alpha' bias near one2}
   \frac{{\psi}'\parens{1-th}}{h^2}\int_{-\infty}^\infty
   k'\parens[\bigg]{ \frac{{\psi}\parens{1-th}-{\psi}\parens{s}}h} {\psi}'\parens{s} {\alpha}\parens{s} \dif s
    =  o\parens{h^2} + {\alpha}'\parens{1-th} +
   \\
\frac{h^2}{80}\parens[\Big]{
8 \cparens[\big]{{\alpha}'''_\uparrow\parens{1}+3{\alpha}'\parens{1}{\psi}''^2\parens{1} - {\alpha}'\parens{1}{\psi}'''\parens{1} - 3{\alpha}''\parens{1}{\psi}''\parens{1}  }
+
\\
\parens{4+t}\parens{1-t}^4 \cparens{{\alpha}'''_\uparrow\parens{1}-{\alpha}'''_\downarrow\parens{1}}},
   \end{multline}
where ${\alpha}_\uparrow''',{\alpha}_\downarrow'''$ denote the third left and right derivatives, respectively.
   \proof
   Let $z_{th}\parens{s} = {\psi}^{-1}\cparens{{\psi}\parens{1-th} +sh}$.  Then $z_{th}\parens{0}=1-th$, $z_{th}'\parens{0}=h / {\psi}'\parens{1-th}$,
   $z_{th}''\parens{0}= -h^2 {\psi}''\parens{1-th} / {\psi}'^3\parens{1-th}$.
  The left hand side in \cref{eq:alpha' bias near one2} is
   \<\label{eq:alpha' bias near one2}
     \frac{{\psi}'\parens{1-th}}h \int_{-\infty}^\infty k\parens[\bigg]{ \frac{{\psi}\parens{1-th}-{\psi}\parens{s}}h} {\alpha}'\parens{s} \dif s 
     =
     {\psi}'\cparens{z_{th}\parens{0}}\int_{-1}^1 k\parens{s} \frac{ {\alpha}'\cparens{z_{th}\parens{s}}}{{\psi}'\cparens{z_{th}\parens{s}}} \dif s.
     \>
    Now, for $|s|\leq 1$ we have, uniformly in $s$,
   \begin{multline*}
      {\psi}'\cparens{z_{th}\parens{0}} \frac{ {\alpha}'\cparens{z_{th}\parens{s}}}{{\psi}'\cparens{z_{th}\parens{s}}}
       -
        {\alpha}'\cparens{z_{th}\parens{0}}
       =
       o\parens{h^2} +
h s \parens[\Big]{\frac{ {\alpha}''}{{\psi}'} - \frac{{\alpha}'{\psi}''}{{\psi}'^2}}
\\
+
\frac{h^2s^2}2
\cparens[\big]{
{\alpha}'''_\uparrow\parens{1} 
-
3 {\alpha}'' {\psi}'' - {\alpha}' {\psi}''' + 3 {\alpha}' {\psi}''^2
}
+
\one\parens{s>t}
\frac{h^2\parens{s-t}^2}2  \cparens{ {\alpha}'''_\downarrow\parens{1}-{\alpha}'''_\uparrow\parens{1}}
,
   \end{multline*}
where all omitted arguments of the ${\alpha},{\psi}$ functions are $1-th$.  
Hence the right hand side in \cref{eq:alpha' bias near one2} is
\[
o\parens{h^2} + {\alpha}'\parens{1-th}
+
\frac{h^2}2{\kappa}_2^*  \parens{{\alpha}'''_\uparrow+ 3{\alpha}'{\psi}''^2-{\alpha}'{\psi}'''-3{\alpha}''{\psi}''}
+
\frac{h^2}2 \parens{{\alpha}'''_\downarrow-{\alpha}'''_\uparrow} \int_t^1 k\parens{s} \parens{s-t}^2 \dif s,
\]
where the ${\alpha}$'s and ${\psi}$'s are evaluated at 1.  
%
Finally, observe that for the Epanechnikov kernel, ${\kappa}_2^*=1/5$ and, 
\[
\int_t^1 k\parens{s} \parens{s-t}^2 \dif s = 
\frac{\parens{4+t}\parens{1-t}^4}{40}. \qedhere
\]
%
%

 \qed
\end{lem}

\begin{proof}[Proof of \cref{thm:derivatives near the boundary}]
	\Cref{lem:alpha' bias near one} provides the formula for the asymptotic bias.  For the asymptotic distribution, we start from \cref{lem:nbhd}.  Take ${\upsilon}_T$ to have the same meaning as in the proof of \cref{thm:boundary correction reflection}.  Note that the sum of the last two terms in \cref{eq:nbhd1} equals
	\[
	\frac3{2h^2} \int_0^{1-t} {\upsilon}_T\parens{1-sh} \dif s
	- 
	\frac3{2h^2} \int_0^{1+t} {\upsilon}_T\parens{1-sh} \dif s
	=
	\frac3{2h^2} \int_{1-t}^{1+t} {\upsilon}_T\parens{1-sh} \dif s,
	\]
	which produces the promised asymptotic distribution. 
	
Under \cref{eq:Hstar simple} the asymptotic variance simplifies to
\begin{multline*}
\frac94 {\zeta}^2\parens{1} 
\int_{1-t}^{1+t} \int_{1-t}^{1+t} \min\parens{s,\tilde s} \dif \tilde s \dif s 
=
\frac92 {\zeta}^2\parens{1}  \int_{1-t}^{1+t} s\int_s^{1+t} \dif \tilde s \dif s
=
\frac92 {\zeta}^2\parens{1} \int_{1-t}^{1+t}
s\parens{1+t-s} \dif s\\
=
\frac34 {\zeta}^2\parens{1} \sparens[\big]{
3\parens{1+t} \cparens{\parens{1+t}^2-\parens{1-t}^2}
-
2\cparens{ \parens{1+t}^3 - \parens{1-t}^3}
}
=
3 {\zeta}^2\parens{1} t^2 \parens{3-t},
\end{multline*}
as asserted. 
\end{proof}

\subsection{Distribution of win--probabilities}
\begin{proof}[Proof of \cref{thm:Fp}]
%
Note that
\begin{multline} \label{eq:Fp1}
	\sqrt{T} \sparens[\big]{
G_T\cparens{\hat Q_c\parens{p}} - G\cparens{Q_c\parens{p}}	
}
=
\sqrt{T} \sparens[\big]{
G_T\cparens{Q_{cT}\parens{p}} - G_T\cparens{Q_c\parens{p}}
-G\cparens{Q_{cT}\parens{p}} + G\cparens{Q_c\parens{p}}
}
\\
+
\sqrt{T} \sparens[\big]{
	G_T\cparens{Q_c\parens{p}} - G\cparens{Q_c\parens{p}}
}
+
\sqrt{T} \sparens[\big]{
	G\cparens{Q_{cT}\parens{p}} - G\cparens{Q_c\parens{p}}
}.
\end{multline}
The first right hand side term in \cref{eq:Fp1} is \opone since $\sqrt{T}\parens{G_T-G}$ converges weakly to a Gaussian process and $Q_{cT}$ is consistent for $Q_c$.  The third right hand side term expands as
\(
 g\cparens{Q_c\parens{p}} \cparens{ Q_{cT}\parens{p}-Q_c\parens{p}}
\)
plus terms of (uniformly) lesser order.  Noting that $Q_{cT}$ converges superconsistently at the boundaries, the stated result then follows from the independence of $G_T$ and $Q_{cT}$. 
	\end{proof}

\subsection{Derived objects}
\label{sec:object proofs}

\begin{proof}[Proof of \cref{thm:bs symmetric not smooth}]
 Using integration by parts we get
 \begin{multline*}
 \sqrt{T} \parens{ \bshat-\bs} = \sqrt{T} \uint{ \sparens[\big]{\cparens{{\alpha}_T\parens{p}-{\alpha}\parens{p}} p - \cparens{\beT\parens{p}-e\parens{p}}} \fpp}
 =
 \\
 \sqrt{T} \cparens{ \beT\parens{1}-e\parens{1}} f_p\parens{1}
 -
 \sqrt{T} \uint{ \cparens{ \beT\parens{p}-e\parens{p}} \cparens{ f_p'\parens{p}p + 2 f_p\parens{p}} }
 =
 \\
\sqrt{T} \uint{ \cparens{ \beT\parens{p}-e\parens{p}} \frac{n}{\parens{n-1}^2} p^{\parens{2-n}/\parens{n-1}} } + \opone.
 \end{multline*}
    Apply \cref{thm:ebreve,lem:V symmetric}.
\end{proof}

\begin{lem} \label{lem:alphaone}
${\alpha}_T\parens{1} = O_p\parens{1}$.\joris{requires $T H^*\parens{1-1/T,1-1/T}=O_p\parens{1}$.  Not sure what we could/should say about this.}
\proof
Since $\G\parens{1}=0$ a.s., we have that for any $C> {\alpha}\parens{1}$,
\begin{multline*}
\lim_{T\to\infty} \Pr\cparens{ {\alpha}_T\parens{1} > 3C} = \lim_{T\to\infty} \Pr\cparens{ {\alpha}_T\parens{1}-\alpha\parens{1}>2C}
=
\\
\lim_{T\to\infty} \Pr\sparens[\big]{
T \cparens{ \beT\parens{1}-\beT\parens{1-1/T} - e\parens{1}+e\parens{1-1/T}} 
+ T\cparens{e\parens{1}-e\parens{1-1/T}} - {\alpha}\parens{1} > 2C
}
\leq
\\
\lim_{T\to\infty} \Pr\cparens[\big]{
-\sqrt{T}\G\parens{1-1/T} > C
}=
\lim_{T\to\infty}
\Phi\parens[\bigg]{ - \frac{C}{\sqrt{TH\parens{1-1/T,1-1/T}}}} 
=\Phi\parens{-cC},
\end{multline*}	
for some $c<\infty$ independent of $C$.  Take $C\to\infty$ to make the right hand side zero. \qed	
\end{lem}

\begin{lem} \label{lem:bs asymmetric not smooth}
    $\uint{\cparens{ {\alpha}_T\parens{p}-{\alpha}\parens{p}}^2 f_p\parens{p}}=\opone$.
   \proof
  We have uniform convergence of ${\alpha}_T$ by \cref{thm:ebreve} except at the boundaries.  Since ${\alpha}_T,{\alpha}$ are nondecreasing and nonnegative, we only have to worry about ${\alpha}_T$ near one.
Now, let $I_m=\int_0^{\bar p} \cparens{ {\alpha}_T\parens{p}-{\alpha}\parens{p}}^2 f_p\parens{p} \dif p$ and
$I_r = \int_{\bar p}^1 \cparens{ {\alpha}_T\parens{p}-{\alpha}\parens{p}}^2 f_p\parens{p} \dif p$ for a $0<\bar p<1$ to be manipulated later.
Now, for any ${\epsilon}>0$ and $0<C<\infty$,
\< \label{eq:bs asns0}
 \Pr\parens{I_m+I_r>2{\epsilon}} 
 \leq
 \Pr\parens{I_m>{\epsilon}} + \Pr\cparens{I_r>{\epsilon}, {\alpha}_T\parens{1}\leq C} + \Pr\cparens{ {\alpha}_T\parens{1}>C}
\>
Take $C= {\epsilon}/\parens{1-\bar p}$.  Then the second right hand side probability in \cref{eq:bs asns0} equals zero. Then take $\limsup_{T\to\infty}$ in \cref{eq:bs asns0}, followed by $C\to\infty$ to obtain the stated result. \qed
\end{lem}

\begin{proof}[Proof of \cref{thm:bs asymmetric not smooth}]
Note that
\begin{multline} \label{eq:bs asns1}
\sqrt{T}\parens{\bshat - \bs} 
=
\underbrace{\sqrt{T} \uint{ \sparens{\cparens{{\alpha}_T\parens{p}-{\alpha}\parens{p}}  p - \cparens{\beT\parens{p}-e\parens{p}}}f_p\parens{p} }}_{\mathrm{I}}
+
\\
\underbrace{\uint[\G_{Tp}\parens{p}]{A\parens{p} }}_{\mathrm{II}}
+
\underbrace{\uint[\G_{Tp}\parens{p}]{\cparens{{\alpha}_T\parens{p}-{\alpha}\parens{p}}  }}_{\mathrm{III}}
-
\underbrace{\uint[\G_{Tp}\parens{p}]{\cparens{\beT\parens{p}-e\parens{p}}  }}_{\mathrm{IV}},
\end{multline}
where $\G_{Tp} = \sqrt{T} \parens{F_{pT}-F_p}$.  First, note that $\G_{Tp} \convw \G_p$ by \cref{thm:Fp}.  Thus,
since the class of right--continuous step functions is Donsker and by \cref{lem:bs asymmetric not smooth}, \citet[lemma 19.24]{vandervaart2000asymptotic} implies that term III in \cref{eq:bs asns1} is \opone.\footnote{Lemma 19.24 in \citet{vandervaart2000asymptotic} is stated specifically for empirical processes, but its proof relies merely on continuity properties and the fact that $F_{pT}$ is not an empirical distribution function is hence immaterial ($F_{pT}$ is the empirical distribution function of estimated $p$'s, not of the $p$'s themselves).}  Further, term IV is $o_p\parens{1}$ by \cref{thm:Fp}.

Now, term I in \cref{eq:bs asns1} is using integration by parts equal to
\< \label{eq:bs asns2}
\sqrt{T} \cparens{ \beT\parens{1}-e\parens{1}} f_p\parens{1}
-
\sqrt{T} \uint{ \cparens{ \beT\parens{p}-e\parens{p}} \cparens{ p f_p'\parens{p} +2 \fpp }}.
\>
Note that the first term in \cref{eq:bs asns2} is \opone.  Term II in \cref{eq:bs asns1} can likewise be written as
\< \label{eq:bs asns3}
-\uint{{\alpha}'\parens{p}p \G_{Tp}\parens{p}}.
\>
Now, combining the above results with the proof of \cref{thm:Fp}, it follows that
\begin{multline*}
 \sqrt{T}\parens{\bshat-\bs} 
 =
 - \uint{ \cparens[\big]{ p^2 f_p'\parens{p} + 2p\fpp + {\alpha}'\parens{p}p g\cparens{ Q_c\parens{p}} }
\sqrt{T}\cparens{Q_{cT}\parens{p}-Q_c\parens{p} }
}\\
-
\uint{ {\alpha}'\parens{p}p \sqrt{T}\sparens{G_T\cparens{Q_c\parens{p}} - G\cparens{Q_c\parens{p}}}}
+
\opone,
\end{multline*}
which has a mean zero normal limit with variance $\sV_{BS}^a$,
where we have used the fact that $G,G_c$ are estimated using different data such that $G_T$ and $Q_{cT}$ are independent. 

To establish \cref{eq:bs efficiency bound}, consider $\int_0^1\int_0^1 {\Gamma}_2\parens{p} {\Gamma}_2\parens{p^*} H_1\cparens{Q_c\parens{p},Q_c\parens{p^*}} \dif p^* \dif p$, which we now show to equal the first term in \cref{eq:bs efficiency bound}: showing that the remainder of \cref{eq:bs asymmetric variance} is equal to the second term in \cref{eq:bs efficiency bound} follows the same path, but is messier.  

Let $I_j = \int_0^{\bar b} \cparens{ G_c^2\parens{b}/g_c\parens{b}}^j g\parens{b} \dif b$.
First, use integration by parts to obtain
\< \label{eq:bs effing bound1}
 \uint{ {\Gamma}_2\parens{p} G\cparens{Q_c\parens{p}}}
 =
  Q_c'\parens{1} - I_1.
\>
Then,
\[
 \int_p^1 {\Gamma}_2\parens{t} \dif t = Q_c'\parens{1} - Q_c'\parens{p} p^2,
\]
whence
\begin{multline} \label{eq:bs effing bound2}
 \uint{ \uint[p^*]{ {\Gamma}_2\parens{p} {\Gamma}_2\parens{p^*} G\sparens{Q_c\cparens{\min\parens{p,p^*} }}}}
 =
 2 \uint{ {\Gamma}_2\parens{p}G\cparens{Q_c\parens{p}}\int_p^1 {\Gamma}_2\parens{t} \dif t}=
 \\
 = Q_c'^2\parens{1} - 2 Q_c'\parens{1} I_1 + I_2.
\end{multline}
Subtract the square of \cref{eq:bs effing bound1} from \cref{eq:bs effing bound2} to obtain
\(
 I_2 -I_1^2,
\)
as promised.  To see that \cref{eq:bs efficiency bound} is in fact the semiparametric efficiency bound note that for any hypothetical parameter vector ${\theta}$ indexing $g,g_c$,
\begin{multline*}
 \partial_{\theta} \BS = \partial_{\theta} \int_0^{\bar b} \frac{G_c^2\parens{b}g\parens{b}}{g_c\parens{b}} \dif b
 =
 \\
 \int_0^{\bar b} \frac{G_c^2\parens{b}}{g_c\parens{b}} \partial_{\theta} g\parens{b} \dif b
 -
 \int_0^{\bar b} \frac{G_c^2\parens{b} g\parens{b}}{g_c^2\parens{b}} \partial_{\theta} g_c\parens{b} \dif b
 -
 2 \int_0^{\bar b} \int_0^b \frac{G_c\parens{t} g\parens{t}}{g_c\parens{t}} \dif t \partial_{\theta} g_c\parens{b} \dif b
 =
 \\
 \Expr[\Big]{ \frac{G_c^2\parens{b}}{g_c\parens{b}} \partial_{\theta} \log g\parens{b} }
 -
 \Expc[\bigg]{ \parens[\Big]{\frac{G_c^2\parens{b_c} g\parens{b_c}}{g_c^2\parens{b_c}} 
+ 2 \int_0^{b_c} \frac{G_c\parens{t} g\parens{t}}{g_c\parens{t}}\dif t} 
\partial_{\theta} \log g_c\parens{b_c} }
\\
=
\Expc[\bigg]{ \parens[\bigg]{ \frac{G_c^2\parens{b}}{g_c\parens{b}} -
 \frac{G_c^2\parens{b_c} g\parens{b_c}}{g_c^2\parens{b_c}} 
		- 2 \int_0^{b_c} \frac{G_c\parens{t} g\parens{t}}{g_c\parens{t}}\dif t}
	\partial_{\theta} \log \cparens{g\parens{b}g_c\parens{b_c}}},
\end{multline*}
which yields the stated bound by the arguments in \citet[page 106]{newey1990semiparametric}.  We have ignored the possibility that the upper bound can depend on ${\theta}$ but that is irrelevant since the upper bound can be estimated at a rate faster than $\sqrt{T}$.
\end{proof}

\begin{proof}[Proof of \cref{thm:bs asymmetric smoothed}]
The proof is largely a repeat of that of \cref{thm:bs asymmetric not smooth}.  The main difference concerns
\begin{multline*}
\sqrt{T}\int_0^1 \sparens[\big]{ \cparens{\hat {\alpha}_{T{\psi}}\parens{p}-{\alpha}\parens{p}}p - \cparens{\hat e_{T{\psi}}\parens{p}-e\parens{p}}} f_p\parens{p}\dif p	
\\
=
\sqrt{T}\cparens{ \hat e_{T{\psi}}\parens{1}-e\parens{1}} f_p\parens{1}
-
\frac{\sqrt{T}}h
\uint{
	\int_{-\infty}^\infty \cparens{ e\parens{s}-e\parens{p} } {\psi}'\parens{s}
	k\parens[\bigg]{ \frac{ {\psi}\parens{p}-{\psi}\parens{s}}h} \dif s p f_p'\parens{p}
}
\\
-
\frac{\sqrt{T}}h
\uint{
\int_{-\infty}^\infty \cparens{ \beT\parens{s}-e\parens{s} } {\psi}'\parens{s}
k\parens[\bigg]{ \frac{ {\psi}\parens{p}-{\psi}\parens{s}}h} \dif s p f_p'\parens{p}
}
\\
=
o_p\parens{1}
-
\uint{ \sqrt{T} \cparens{ \beT\parens{p} - e\parens{p}} p f_p'\parens{p}
},
\end{multline*}
where we have omitted a few steps entailing nothing more than substitution and simple expansions, including a nonparametric kernel bias expansion. 
\end{proof}

\begin{proof}[Proof of \cref{thm:mean v symmetric}]
Note that
\begin{multline*}
\sqrt{T} \uint[F_p\parens{p}]{ \cparens{{\alpha}_T\parens{p}-{\alpha}\parens{p}}} =
\sqrt{T} \cparens{ \beT\parens{1}-e\parens{1}} f_p\parens{1}
-
\sqrt{T} \uint{ \cparens{ \beT\parens{p}-e\parens{p}} f_p'\parens{p} }\\
\convd
\uint{ \G\parens{p} f_p'\parens{p}},
\end{multline*}
by \cref{thm:ebreve}.  Note that $f_p'\parens{p} = \parens{2-n} p^{\parens{3-2n}/\parens{n-1}} / \parens{n-1}^2$.
\end{proof}

\begin{proof}[Proof of \cref{thm:mean v}]
The proof follows with minor adjustments by repeating the steps in the proof of \cref{thm:bs asymmetric not smooth}.
\end{proof}

\end{document}